%% file: qudit-quantum.tex
\documentclass[a4paper,onecolumn,11pt,unpublished]{quantumarticle}
\pdfoutput=1
\usepackage[utf8]{inputenc}
\usepackage[english]{babel}
\usepackage[T1]{fontenc}
\usepackage[numbers,sort&compress]{natbib}
\input{qupit.sty} % only if this does not also load hyperref

\definecolor{citeDarkGreen}{RGB}{0, 100, 0}
\definecolor{linkDarkRed}{RGB}{180, 0, 0}
\hypersetup{citecolor=citeDarkGreen, linkcolor=linkDarkRed}

\title{A Complete and Natural Rule Set for Multi-Qudit Clifford Circuits in All Odd Prime Dimensions}

\author[1]{Xiaoning Bian}
\email{bian@dal.ca}

\author[2,3]{Sarah Meng Li}
\email{sarah.li@uva.nl}

\author[4]{Neil J. Ross}
\email{neil.jr.ross@dal.ca}

\author[2,3]{John van de Wetering}
\email{john@vdwetering.name}

\author[5]{Yuming Zhao}
\email{yuming@math.ku.dk}

\affil[1]{Tsinghua University, Beijing, China}
\affil[2]{University of Amsterdam, Amsterdam, Netherlands}
\affil[3]{QuSoft, Amsterdam, Netherlands}
\affil[4]{Dalhousie University, Halifax, Canada}
\affil[5]{University of Copenhagen, Copenhagen, Denmark}

\begin{document}

\maketitle

\begin{abstract}
We present a complete set of $16$ rewrite rules for multi-qudit Clifford circuits in all odd prime dimensions. Completeness means that any two Clifford circuits representing the same linear map can be rewritten into each other using these rules. Each rule involves at most three qudits and admits an intuitive interpretation. To establish completeness, we first work at the symplectic level, using the isomorphism between the symplectic group $\mathrm{Sp}(2n,\mathbb{Z}_p)$ and the quotient of the Clifford group by the Pauli group. We construct a circuit normal form that captures the stabiliser tableau of a Clifford operator and is unique up to Pauli correction. Using this normal form, we derive a complete set of symplectic relations, which we then lift to Clifford relations by incorporating Pauli corrections. We also formally verify completeness in the Agda proof assistant.
\end{abstract}

\begingroup
\hypersetup{linkcolor=black}
\tableofcontents
\endgroup

\input{scripts/1-introduction}
\input{scripts/2-foundation}

\input{scripts/3-assemble}
\input{scripts/4-completeness}

\input{scripts/5-conclusion}

\bibliographystyle{plain}
\bibliography{qupit}

\onecolumn
\appendix

\input{scripts/appendix/sectiontwoproofs}
\input{scripts/appendix/clifford}
\input{scripts/appendix/comp-def}
\input{scripts/appendix/push-normal}
\input{scripts/appendix/normal-form}
\input{scripts/appendix/boxrelations}

\end{document}

%% file: scripts/1-introduction.tex
\section{Introduction}
\label{sec:introduction}

Many physical platforms for qubits naturally possess higher energy levels that can encode qudits, enabling greater information density at the cost of increased control complexity~\cite{chi2022programmable,PhysRevLett.129.160501,Rin21,romanova2025measurement,wang2020qudits}. In particular, when $p$ is an odd prime, qudit systems exhibit rich algebraic structures that result in unique error correction capabilities~\cite{PhysRevLett.113.230501,FG24}, lower overhead magic state distillation~\cite{PhysRevX.2.041021,PS24}, stronger quantum correlations~\cite{Bru08,Des22}, and improvements to various quantum algorithms~\cite{PhysRevA.96.012306,PhysRevA.103.032417}. Ultimately, there is a trade-off between the engineering cost of accessing higher dimensions and the computational advantages they provide. By characterising qudit circuits in a systematic manner, we can better understand this balance and exploit the potential of higher-dimensional quantum information processing.

Recently, algebraic approaches to quantum circuits have gained significant attention. An important problem in this line of work is to develop sound and complete equational theories for circuit families, meaning any two circuits representing the same linear map can be transformed into each other using a finite set of rules. Such presentations enable syntactic reasoning about circuits and underpin applications ranging from circuit optimisation to formal verification~\cite{bravyi2021clifford,xu2022quartz,xu2023synthesizing,jiang2024pattern,chareton2021automated,huang2026equivalence}. While there has been substantial progress in the qubit setting~\cite{cnotdihedral,toffolihadamard,Bian2021GeneratorsAR,cliffordt,cliffordcs,greylyn2014generators,dyadicinteger,makary2021generators,selinger2015generators}, analogous results for qudit circuits remain comparatively underdeveloped.

Among quantum circuit fragments, the Clifford group plays a particularly prominent role. Clifford unitaries are central to the stabiliser formalism~\cite{aaronson2004improved,gheorghiu2014standard}, fault-tolerant quantum computing~\cite{Got99,gottesmanbook,quditqec}, efficient classical simulation~\cite{harper2025gcamps,nest2008classical}, and practical circuit optimisation techniques~\cite{bravyi20226,bravyi2021clifford}. As a result, obtaining a complete and finite equational theory for qudit Clifford circuits is a natural and foundational problem.

Note that there is a complete equational theory for all (including non-unitary) qudit Clifford maps in the form of the qudit Clifford ZX-calculus~\cite{Booth:2022isz,poor2023qupit,wang2018qutrit}.
One might hope that we could easily adapt this into a complete calculus for unitary circuits, just by restricting the allowed maps. However, as ZX-diagrams represent general linear maps this does not work in any straightforward way. In general, completeness of a graphical calculus for a broad class of linear maps does not yield a calculus restricted to unitary maps. Consider for instance that a complete ZX-calculus for universal qubit linear maps was found in 2017~\cite{jeandel2018diagrammatic,ng2017universal}, but that it was not until 2022 that a complete calculus of universal quantum circuits appeared (which used very different methods)~\cite{clement2023complete}. Moreover, there is no complete calculus for, for instance, Toffoli-Hadamard circuits, even though there is one for the corresponding set of linear maps~\cite{backens2023completeness}.

\subsection{Main Results}
\label{subsec:results}

Building on the characterisation of the qutrit Clifford group~\cite{qutrit2024}, where a qutrit is a three-dimensional ($p=3$) qudit, we present a complete set of rewrite rules for multi-qudit Clifford circuits in all odd prime dimensions. There are $16$ \emph{Clifford relations} in total, each involving at most three qudits and admitting an intuitive interpretation.

\begin{T7}
    \completeness
\end{T7}

\begin{figure}[!thb]
    \[
    \scalebox{0.6}{\input{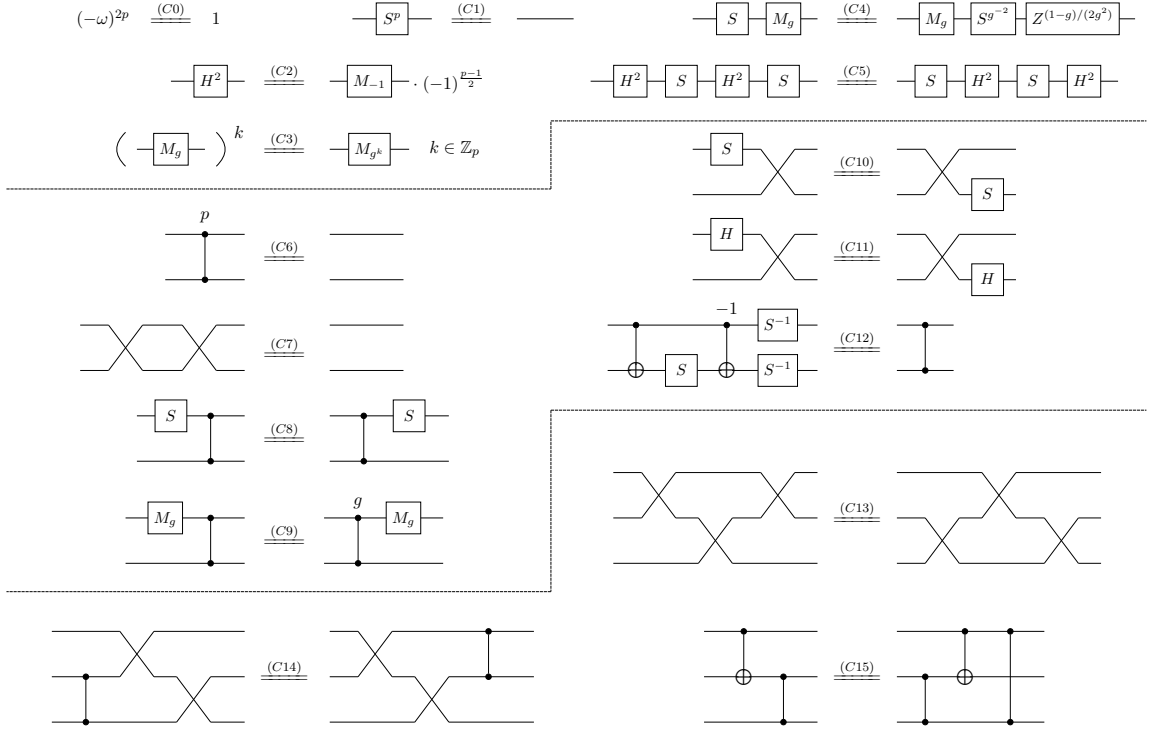}}
    \]
  \caption{A complete set of rewrite rules for $n$-qudit Clifford circuits with scalar generator $-\omega$, where $\omega = e^{2\pi i/p}$, $p$ is an odd prime, and $n\in \N$. We have $-1=(-\omega)^p$ and $\omega = (-\omega)^{p+1}$. $g$ is a chosen generator of the multiplicative group $\mathbb{Z}_p^*$. Derived generators such as $M_g$, $X$, $Z$, \SWAP, and \CX are defined in \cref{fig:derived-generators1}.}
  \label{fig:rewriterules6}
\end{figure}

\begin{figure}[!htb]
    \[
    \scalebox{.7}{\input{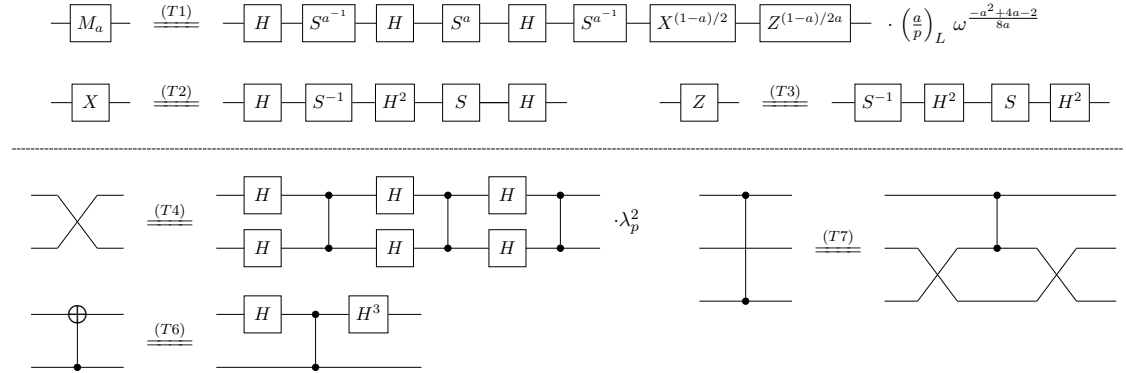}}
    \]
    \caption{A subset of the derived Clifford generators, expressed as circuits over the generating set $\{-\omega,H,S,\CZ\}$. These generators are formally defined in \cref{fig:derived-generators2}. Here, $a$ is an arbitrary element of the multiplicative group $\mathbb{Z}_p^*$ and $\lambda_p  = e^{(p-1)\pi i/4}$. $\left(\frac{a}{p}\right)_L$ denotes the Legendre symbol~\cite[Appendix A]{prakash2021normal}, which is $1$ if $a$ is a quadratic residue modulo $p$, that is, $a = x^2$ for some $x \in \Z_p$, and $-1$ otherwise.}
    \label{fig:derived-generators1}
\end{figure}  

To prove \cref{thm:completeness}, we introduce a symplectic normal form that corresponds to the stabiliser tableau of a Clifford operator up to Pauli corrections. This normal form is substantially easier to work with compared to the more involved normal form for qutrit Clifford operators. We then find a set of $42$ parametrised relations which suffice to reduce any Clifford circuit to this normal form. Finally, we show how these relations can be further reduced to a small set of $16$ Clifford relations in \cref{fig:rewriterules6}. In addition to the Clifford generators $-\omega$, $H$, $S$ and $\CZ$, we also use \emph{derived generators} listed in \cref{fig:derived-generators1}, such as \SWAP, \CX, the \emph{multiplier} $M_a$, and the remote $\CZ$ gate $\CIZ$. $M_a$ is defined by its action on the computational basis states as $M_a\ket{x} = \ket{ax}$, where $a$ is any element of $\Z_p^*$ and multiplication is over $\Z_p$. Using these derived generators, we can present the rewrite rules as intuitive gate-level interactions, mostly in the form of commutation and quasi-commutation rules.

\begin{remark}
\label{rem:minus-one}
The symbol $-1$ is used in three ways in \cref{fig:rewriterules6,fig:derived-generators1}. As a global scalar, as in $(-1)^{(p-1)/2}$ in \eqref{eq:H2-soundness}, it is the complex number $e^{\pi i}$. In the subscript of a multiplier, it is an element of $\Z_p$: $M_{-1}$ maps $\ket{x}$ to $\ket{-x}$. Similarly, a negative exponent on a gate of order $p$ is read in $\Z_p$: in \eqref{eq:CX-CZ-soundness}, $\CX^{-1}$ denotes $\CX^{p-1}$ and $S^{-1}$ denotes $S^{p-1}$. Finally, an inverse inside an exponent, such as $a^{-1}$ in \eqref{eq:def-Mg} or $g^{-2}$ in \eqref{eq:Mg-S-soundness}, is the modular inverse in $\Z_p$. Thus $S^{a^{-1}}$ is an integer power of $S$, not an $a$-th root of $S$. Fractions in exponents, such as $(1-g)/(2g^2)$ in \eqref{eq:Mg-S-soundness}, are read in the same way.
\end{remark}

\paragraph{Interpretation of the rewrite rules}
The multi-qudit rewrite rules of \cref{fig:rewriterules6} mostly correspond to equations that look familiar. For three qudits, we only have equations corresponding to commuting a CZ gate through a set of SWAP gates \eqref{eq:SWAP-CZ-soundness} or a CNOT gate \eqref{eq:CZ-CX-soundness}, together with the Yang–Baxter equation \eqref{eq:SWAP-SWAP-soundness}~\cite{poulain-dandecy-2020-fusion}. For two qudits, we have the rules for commuting generators through the CZ gate \eqref{eq:CZ-S-soundness}-\eqref{eq:CZ-Mg-soundness}, through the SWAP gate \eqref{eq:SWAP-S-soundness}-\eqref{eq:SWAP-H-soundness}, as well as the order of the CZ and SWAP gates \eqref{eq:CZ-soundness}-\eqref{eq:SWAP-soundness}. Finally, \eqref{eq:CX-CZ-soundness} is the path-sum decomposition of a CZ gate: the quadratic phase $q(j+\ell)$ induced by conjugating $S$ with a CNOT gate decomposes as $q(j)+q(\ell)+j\ell$, where $q(j)=\frac{j(j-1)}{2}$. After cancelling the local $S$-phases, one is left with the bilinear CZ phase $j\ell$.

The equations for a single qudit are perhaps less familiar. \eqref{eq:omega-soundness} and \eqref{eq:S-soundness} are simply the order of the scalar generator $-\omega$, which is $2p$; and the order of the S gate, which is $p$. \eqref{eq:H2-soundness}, which relates $H^2$ to $M_{-1}$, can be used to derive an Euler decomposition of the Hadamard in terms of Z and X rotations. \eqref{eq:Mg-soundness} relates doing iterated multiplication with a multiplier to doing the multiplication in one go. This rule could be replaced by an equation $M_aM_b=M_{ab}$ for all $a,b \in \mathbb{Z}_p^*$ which would give us $(p-1)^2$ equations instead of just the $p$ equations that \eqref{eq:Mg-soundness} represents. In particular, all equations involving $g$ in \cref{fig:rewriterules6}, \eqref{eq:Mg-S-soundness} and \eqref{eq:CZ-Mg-soundness}, continue to hold if we replace $g$ by any element of $\mathbb{Z}_p^*$, and so we could equally state these rules without choosing a $g$, at the cost of having more rules that are families of equations instead of a single equation. Finally, note that \eqref{eq:SH2-SH2-soundness} states that the gate $H^2SH^2$ commutes with $S$. Such a rule was also needed for qutrits, and was shown to be necessary in that setting~\cite{blake2026simpler}. The soundness of every rule in \cref{fig:rewriterules6}, namely that its two sides implement the same linear map, is established through path-sum calculations~\cite{amy2019towards,koh2017computing} in \cref{subsec:soundness}.

Some interesting observations follow from these rewrite rules. First, there is no complicated dependency on the dimension of the qudit $p$: the only dependencies are the order of the scalar $-\omega$ and of the $S$ and CZ gates, and the number of equations represented by \eqref{eq:Mg-soundness}, one for each $k \in \Z_p$. Second, the rewrites only act on at most three qudits, and this number does not increase with $p$ or $n$. This is in contrast to the rewrite rules for universal qubit circuits where rewrite rules on an arbitrary number of qubits must be included~\cite{clement2023complete}.

\paragraph{Parametrised Relations and Symplectic Lifting}
Extending qutrit Clifford completeness~\cite{qutrit2024} to arbitrary odd prime dimensions requires a systematic method for deriving and simplifying parametrised box relations. For qutrits, $\Z_3$ has only two nonzero elements, each is its own inverse, which simplifies many calculations. For general $p$, we must handle parameters and their inverses symbolically. 

We combine this parametrisation with a divide-and-conquer approach based on $\Clifford_n/\Pauli_n \cong \symplectic{2n,\Z_p}$, separating the symplectic group presentation from Pauli and scalar corrections. We first reduce the box relations to a complete set of symplectic relations, then incorporate Pauli corrections to establish completeness up to global phase, and finally account for scalars to obtain exact Clifford completeness.

We formally verify the completeness of \cref{fig:rewriterules6} in the Agda proof assistant. The formalisation is explained in~\cite{bian2026verified} and the code is available in the GitHub repository~\cite{qupitgit}. Our results connect group theory with circuit rewriting and provide a foundation for automated verification, optimisation, and synthesis of qudit Clifford circuits.

\subsection{Proof Overview}
\label{subsec:method}
We define unitary operators by their actions on the computational basis states. The single-qudit Pauli operators are defined as $X\ket{j}=\ket{j+1}$, $Z\ket{j}=\omega^j\ket{j}$, where $\omega = e^{2\pi i/p}$ and addition is over $\Z_p$. The $n$-qudit Pauli group $\Pauli_n$ is generated by $X$ and $Z$ acting on individual qudits, together with the scalar $-\omega$. The qudit \emph{$H$, $S$, and $\CZ$ gates}~\cite{dickinson1982eigenvectors,glaudellthesis,Got99,knill1996group,knill1996non,prakash2021normal} are defined in \eqref{eq:qudit-gates}, where $i = e^{2\pi i/4}$ is the imaginary unit and $\lambda_p = e^{(p-1)\pi i/4}$. 

\begin{equation}
    H : \ket{j} \mapsto \frac{1}{\lambda_p\sqrt{p}}\sum_{\ell=0}^{p-1}\omega^{j\ell}\ket{\ell}, \quad
    S : \ket{j} \mapsto \omega^{\frac{j(j-1)}{2}}\ket{j},\quad CZ:\ket{j}\ket{\ell} \mapsto \omega^{j\ell}\ket{j}\ket{\ell}.
    \label{eq:qudit-gates}
\end{equation}

The $n$-qudit Clifford group $\Clifford_n$ is the normaliser of $\Pauli_n$ in the group of $p^n$-dimensional unitary operators. It is generated by $H$, $S$, and $\CZ$ gates, together with the scalar $-\omega$, via matrix multiplication and tensor product~\cite{farinholt2014ideal}. For simplicity, we ignore all other global phases. It is well known that the qudit Clifford group modulo Paulis is isomorphic to the symplectic group $\symplectic{2n,\Z_p}$~\cite{appleby2005symmetric}. This correspondence allows Clifford generators to be represented by symplectic matrices over the finite field $\Z_p$, forming the basis of stabiliser tableau representation of Clifford operators~\cite{aaronson2004improved,Got99}. Building on this correspondence, we generalise the prior Clifford completeness results~\cite{qutrit2024,makary2021generators,selinger2015generators} for qubits and qutrits to all odd primes.

\paragraph{A Unique Normal Form for Multi-Qudit Clifford Circuits}

To find a complete set of rewrite rules, we first define a unique normal form for Clifford circuits. We introduce a normal form with two components: a \emph{symplectic normal form} characterising the stabiliser tableau of the Clifford circuit, followed by a \emph{Pauli normal form} giving the necessary Pauli correction. This decomposition reflects the semidirect product structure of the projective Clifford group $\widehat{\Clifford}_n$, obtained by identifying Clifford operators that differ by a global phase, $\widehat{\Clifford}_n\cong \symplectic{2n, \Z_p} \ltimes \left(\Z_p\right)^{2n}$~\cite{weil1964certains}.

We show that symplectic normal forms $N^{(n)}$ are in one-to-one correspondence with symplectic matrices $M\in \symplectic{2n,\Z_p}$. $N^{(n)}$ is defined inductively using structured layers of circuit fragments, which we call \emph{normal boxes}. These are labelled $A$, $B$, $D$, and $E$, each implementing a specific action on Pauli generators. We retain the box labels used in the qubit and qutrit constructions~\cite{selinger2015generators,qutrit2024}, but absorb the $C$ and $F$ boxes from each layer into our Pauli normal form. An illustration of the multi-qudit Clifford normal form is given in \cref{fig:normal-form}.

\begin{figure}[!tbh]
    \centering
     \scalebox{0.65}{\input{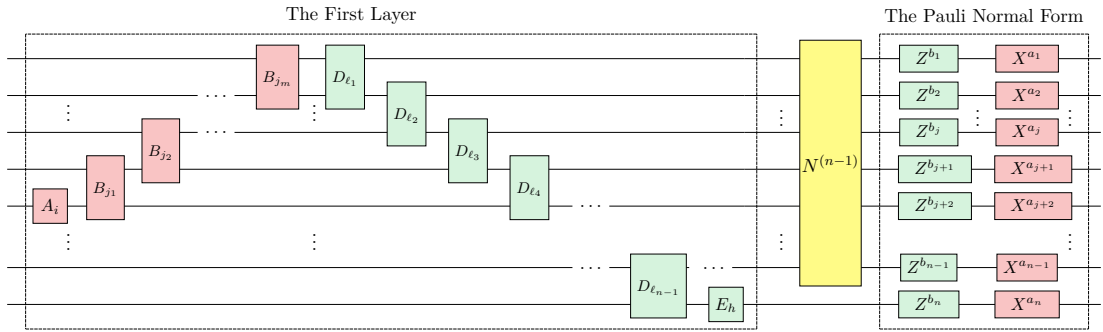}}
    \caption{The inductive construction of the normal form for any $C \in \Clifford_n$, up to a global phase. Each layer of the normal form alternates the Z-normal (red) and X-normal (green) circuits, progressively eliminating nontrivial entries in the stabiliser tableau of $C$. The first layer implements the conjugation action of $C$ on the $n$-th qudit up to Pauli correction. By induction, the normal form $N^{(n-1)}$ implements the remaining symplectic action on the first $n-1$ qudits. Finally, the Pauli normal form carries out the Pauli correction.}
    \label{fig:normal-form}
\end{figure}

\paragraph{Establishing Qudit Clifford Completeness}

To obtain a complete set of Clifford relations, it suffices to find a set of rewrite rules that transforms any Clifford circuit into its unique normal form. We proceed in three steps. First, we show that if a Clifford generator is applied before a normal form, the resulting circuit can be rewritten into a new normal form. This is achieved by showing how each generator can be ``pushed'' through each distinct normal box, yielding equations of the form illustrated in \eqref{eq:normal-box-rel}. We refer to them as the \emph{box relations}.
\begin{equation}
    \scalebox{0.72}{\input{figures/NormalForm/normal-box-rel.tikz}}
\label{eq:normal-box-rel}
\end{equation}  
Here, $\dir$ denotes the \emph{residual dirty gates}. They are composed of Clifford generators whose form depends on the box parameters $a$, $b$, and the dimension $p$. We continue pushing these gates through the remaining boxes until we obtain a new symplectic normal form. In total, there are $42$ box relations, parametrised by the indices of the corresponding normal boxes.

Second, we obtain a compact presentation of the multi-qudit symplectic group using $18$ relations, from which we derive all $42$ box relations, first by hand and then in Agda; representative derivations are given in the supplement~\cite{Supplement2026}. Finally, we incorporate Pauli corrections and derive the rules governing their interaction with Clifford generators, establishing completeness up to global phase. Accounting for scalar phases then yields the complete set of Clifford relations in \cref{fig:rewriterules6}.

\subsection{Paper Organisation}
\label{subsec:org}

The remainder of this paper is organised as follows. In \cref{sec:foundations}, we introduce the algebraic and circuit-theoretic preliminaries. We present the normal forms of Clifford circuits in \cref{sec:assemble} and establish completeness of the rewrite rules in \cref{sec:box-relations}. In \cref{sec:conclusion}, we summarise our results and discuss future directions. Most technical proofs are deferred to the appendices.

%% file: scripts/2-foundation.tex
% --------------------------------------------------------------------
\section{Preliminaries}
\label{sec:foundations}

Here, we introduce the algebraic background and conventions used throughout the paper. In \cref{ssec:rings}, we introduce the rings and roots of unity underlying our odd-prime qudit formalism. In \cref{ssec:operators}, we define the Pauli and Clifford operators, review basic identities, and recall the symplectic structure of the Clifford group. In \cref{ssec:circuits}, we move from operators to circuits by specifying Clifford circuit families, their interpretation maps (unitary, projective, and symplectic), as well as the notions of soundness and completeness for rewriting systems that will be used in later sections. Proofs that are lengthy or purely technical are deferred to \cref{sec:sectiontwoproofs,sec:clifford}.

% --------------------------------------------------------------------
\subsection{Rings}
\label{ssec:rings}

We write $\N$ for the collection of non-negative integers, $\Z$ for the ring of integers, and $\Z_m$ for the ring of integers modulo $m\in\N$. If $a\in\Z_m$ is invertible, we sometimes denote its multiplicative inverse by $1/a$. Throughout the paper, $p$ denotes an arbitrary positive odd prime of $\Z$, and $\Z_p^*$ denotes the multiplicative group of non-zero elements of $\Z_p$.

Recall that a complex number $c\in \C$ is a \emph{root of unity} if $c^m=1$ for some positive integer $m$. We will be especially interested in $p$-th roots of unity.

\begin{definition}
  \label{def:omega}
  We write $\omega_p$ for the \emph{primitive $p$-th root of unity} $\omega_p = e^{2\pi i/p}$.
\end{definition}

In what follows, we drop the subscript and write $\omega$ for $\omega_p$, leaving $p$ implicit. Note that $\omega^p = 1$, $\omega^\dagger = \omega^{-1} = \omega^{p-1}$, and $\sum_{j=0}^{p-1}\omega^j = 0$.

\begin{remark}
  \label{rem:modular-inverse-exponent}
  Exponents of $\omega$, and of gates of order $p$, are read in $\Z_p$. In particular, a fraction $u/v$ in such an exponent, with $p\nmid v$, can be written as $uv^{-1} \bmod p$. For example, $\omega^{1/a}$ with $a\in\Z_p^*$ is the integer power $\omega^{a^{-1}}$, not $e^{2\pi i/(ap)}$. See also \cref{rem:minus-one}.
\end{remark}

\begin{definition}
  \label{def:cyclotomic integers}
  The ring $\Z[\omega]$ of \emph{cyclotomic integers} is the smallest subring of $\C$ containing $\omega$. 
\end{definition}

Since $\omega^\dagger = \omega^{p-1}$, the ring $\Z[\omega]$ is closed under complex conjugation. Every element of $\Z[\omega]$ can be written as a linear combination of $\omega^j,0\leq j\leq p-1$, so that
\begin{equation}
  \Z[\omega] = \Biggl\{\sum_{j=0}^{p-1}a_j\omega^{j}  \mid a_j \in \Z\Biggr\}.
  \label{eq:unique-fac}
\end{equation}

\begin{definition}
  \label{def:ring}
  The ring $\Z[1/p,\omega]$ is the smallest subring of $\C$ containing $1/p$ and $\omega$.
\end{definition}

The ring $\Z[1/p,\omega]$ is the localisation of $\Z[\omega]$ at the multiplicative set $\{1, p, p^2, \ldots\}$.
\begin{equation}
  \Z[1/p,\omega] = \biggl\{\frac{u}{p^\ell}\mid u \in \Z[\omega], \ell \in \N\biggr\}.
  \label{eq:unique-fac-local}
\end{equation}

\begin{definition}
\label{def:lambda}
We define $\lambda_p  := e^{(p-1)\pi i/4}$. In particular,
 \begin{equation*}
 \lambda_p=\begin{cases}
       1, &  p\equiv 1 \pmod{8} \\
       i, &  p\equiv 3 \pmod{8} \\
      -1, &  p\equiv 5 \pmod{8} \\
       -i, &  p\equiv 7 \pmod{8}.
     \end{cases}
    \label{eq:phase}
 \end{equation*}
\end{definition}

The ring $\mathbb{Z}[1/p,\omega]$ has two properties that will be useful later. \cref{lem:lambda-in-ring,lem:roots-of-unity} are proved in \cref{sec:sectiontwoproofs}.

\begin{proposition}
  \label{lem:lambda-in-ring}
  % $\frac{1}{\lambda_p\sqrt{p}}\in\Z[1/p,\omega]$.
  \lambdainring
\end{proposition}    

\begin{proposition}
  \rootsofunity
  \label{lem:roots-of-unity}
  % The roots of unity in $\Z[1/p,\omega]$ are $\pm \omega^t$, $t \in \Z_p$.
\end{proposition}    

% --------------------------------------------------------------------
\subsection{Operators}
\label{ssec:operators}
We now introduce the operators of primary interest in this work, namely the \emph{Pauli} operators~\cite{Got99} and the \emph{Clifford} operators~\cite{dickinson1982eigenvectors,glaudellthesis,Got99,knill1996group,knill1996non,prakash2021normal}.

% We now introduce the operators that will be of interest in the rest of the paper: the \emph{Pauli} operators~\cite{Got99} and the \emph{Clifford} operators~\cite{dickinson1982eigenvectors,glaudellthesis,Got99,knill1996group,knill1996non,prakash2021normal}.

In what follows, we write $I_m$ for the $m\times m$ identity matrix, and we omit the subscript $m$ when it is not relevant or can be inferred from the context. We say that an operator $U$ is a \emph{scalar}, if it is a scalar multiple of the identity operator, i.e., if $U = \lambda I$, for some $\lambda \in \C$. For simplicity, we denote $U$ by $\lambda$ in such a case. When $\lambda = 1$, the identity operator $I$ may itself be written as $1$.
%if $U = \lambda I$, for some $\lambda \in \C$. For simplicity, we denote $U$ by $\lambda$ in such a case; in particular, taking $\lambda = 1$, the identity operator $I$ may itself be written as $1$.
%and when it is clear from the context, we abuse the notation and use $1$ to denote both the integer $1$ and the identity operator.

We write $\symplectic{2n,\mathbb{Z}_p}$ for the \emph{symplectic group}. For $m \in \N$, we write $\U_m \coloneq \{U \in \C^{m\times m} \mid U^\dagger U = I_m\}$ for the \emph{unitary group} of $m\times m$ complex matrices under matrix multiplication. We write $\U(1) \coloneq \{\alpha \in \C \mid \lvert \alpha \rvert = 1\}$ for the group of complex numbers of magnitude one under multiplication. The scalars $\{\alpha I_m \mid \alpha \in \U(1)\}$ form a subgroup of $\U_m$ called the \emph{group of global phases}. Following the convention above of denoting a scalar $\alpha I_m$ by $\alpha$, we also denote this subgroup by $\U(1)$.

\begin{definition}
    Let $m \in \N$. $\U_m(\Z[1/p,\omega])$ is the group of $m\times m$ unitaries with entries in $\Z[1/p,\omega]$. 
    \label{def:group}
\end{definition}

Throughout the paper, we define unitary operators by specifying their action on the computational basis. This uniquely determines their action on all states. Unless stated otherwise, juxtaposition 
$AB$ denotes matrix multiplication of the unitaries $A$ and $B$.

% In the remainder of this paper, we will often define unitary operators by specifying their action on the computational basis states. The definition then extends linearly to the rest of the space.

% --------------------------------------------------------------------
\subsubsection{Pauli Operators}
\label{sssec:paulis}

\begin{definition}
    The single-qudit \emph{Pauli $X$ and $Z$ gates} are defined as follows.
\begin{align}
X:\ket{j} \mapsto \ket{j+1}
\qquad
\mbox{and} 
\qquad
Z:\ket{j}\mapsto \omega^j\ket{j},
\end{align}
where $0\leq j \leq p-1$ and addition is performed modulo $p$.
    \label{def:pauli}
\end{definition}

We record several properties of the Pauli $X$ and $Z$ gates. The corresponding derivations can be found in \cref{sec:sectiontwoproofs}. We have
\begin{equation}
\label{eq:order}
X^p = Z^p = 1.
\end{equation}
It follows from \cref{eq:order} that $X^\dagger = X^{p-1}$ and $Z^\dagger = Z^{p-1}$. In addition, we have $ZXZ^\dagger X^\dagger = \omega$, which implies that 
\begin{equation}
  \label{eq:XZ}
  ZX = \omega XZ, \qquad ZXZ^\dagger = \omega X, \qquad \mbox{and} \qquad XZX^\dagger = \omega^{-1} Z.
\end{equation}

\begin{definition}
  \label{def:Pauli-groups}
  For $n \in \N$, the $n$-qudit \emph{Pauli group} $\Pauli_n$ is the group of $n$-qudit unitary operators generated by $-\omega$, $X$, and $Z$ under composition and tensor product.
\end{definition}
% The group $\Pauli_n$ is not closed under tensor product.

\cref{def:Pauli-groups} states that among all of the operators that can be constructed using $-\omega$, $X$, and $Z$ through composition and tensor product, $\Pauli_n$ contains exactly those that act on $n$ qudits. Since $p$ is odd, the scalar $-\omega$ has order $2p$, with $(-\omega)^p = -1$ and $(-\omega)^{p+1} = \omega$ (see \cref{lem:gen-scalar}). We write
\[
\langle -\omega \rangle \coloneq \left\{ (-\omega)^t \mid t \in \Z_{2p}\right\} = \left\{ \pm\omega^t \mid t \in \Z_{p}\right\}
\]
for the cyclic group generated by $-\omega$. Using \cref{eq:order,eq:XZ} and the usual properties of the tensor product, one can show that
\[
\Pauli_n = \left\{ (-\omega)^c \left(X^{a_1}Z^{b_1} \otimes \cdots \otimes X^{a_n}Z^{b_n}\right) \mid c\in \Z_{2p} \mbox{ and } a_j, b_j\in \Z_p \mbox{ for all } 1 \leq j \leq n\right\}.
\]

As a consequence, we have $\lvert \Pauli_n\rvert = 2p^{2n + 1}$. Note that $\omega \in \Pauli_1$ can already be obtained from $X$ and $Z$ because $ZXZ^\dagger X^\dagger = \omega$, whereas the phase $-1 = (-\omega)^p$ cannot. Including $-\omega$ as a generator ensures that $\Pauli_n$ contains all the global phases of the Clifford group defined below (see \cref{lem:qupit-Clifford-phase}), so that $\Clifford_n/\Pauli_n \cong \symplectic{2n,\Z_p}$ in \cref{subsec:symplectic}.

% --------------------------------------------------------------------
\subsubsection{Clifford Operators}
\label{sssec:cliffords}

\begin{definition}
  \label{def:H-and-S}
  The single-qudit $H$ and $S$ gates and the two-qudit $\CZ$ gate are defined as follows, where $\lambda_p$ is given as in \cref{def:lambda}. 
  \[
  H : \ket{j} \mapsto \frac{1}{\lambda_p\sqrt{p}}\sum_{\ell=0}^{p-1}\omega^{j\ell}\ket{\ell},
  \qquad
  S : \ket{j} \mapsto \omega^{\frac{j(j-1)}{2}}\ket{j},
  \qquad
  \mbox{and}
  \qquad
  \CZ:\ket{j}\ket{\ell} \mapsto \omega^{j\ell}\ket{j}\ket{\ell},
  \]
  where $0\leq j,\ell \leq p-1$ and multiplication is performed modulo $p$.
\end{definition}

The gates $H$, $S$, and $\CZ$ are called the \emph{Hadamard} gate, the \emph{phase} gate, and the \emph{controlled-$Z$} gate, respectively. We use these gates and the global phase $-\omega$ to define the \emph{Clifford} group, similar to how $-\omega$ and the $X$ and $Z$ gates were used to define the Pauli group in \cref{def:Pauli-groups}.

\begin{definition}
  \label{def:Clifford-circuits}
  For $n \in \N$, the $n$-qudit \emph{Clifford group} $\Clifford_n$ is the group of $n$-qudit unitary operators generated by $-\omega$, $H$, $S$, and $\CZ$ under composition and tensor product.
\end{definition}

We record a few properties of Clifford operators. In particular, we show that the Clifford phases are exactly the integer powers of $-\omega$.

\begin{proposition}
\label{prop:hsquared}
For $0\leq j \leq p-1$, $H^2 \ket{j} = \lambda_p^{-2}\ket{-j}$. Here the negation $-j$ is performed modulo $p$.
\end{proposition}

Note that $H^2$ is a special case of a \emph{multiplier} since up to a global phase, $H^2\ket{j} = \ket{-j}$. Multipliers admit simple commutation rules with the Clifford generators, which help streamline the relation reduction in~\cite{qupitgit,Supplement2026}.

\begin{definition}
    For every $a\in\Z_p^*$, we define the \emph{multiplier by $a$} as the linear operator $M_a$ on $\C^p$ that acts on computational basis states as $M_a\ket{x} = \ket{ax}$.
\end{definition}

Since $a$ is invertible modulo $p$, $M_a$ permutes the computational basis states, so it is unitary and $M_a^\dagger = M_a^{-1} = M_{a^{-1}}$, where $a^{-1}\in \Z_p^*$ is the inverse of $a$ modulo $p$.
In \cref{lem:multiplier}, we show how to construct a multiplier using $H$, $S$, $X$, and $Z$ gates together with a global phase in $\langle -\omega\rangle$, so that it is indeed Clifford.

\begin{lemma}
    % For every $a\in\Z_p^*$, $M_a$ defined in \Cref{fig:derived-generators1} satisfies $M_a\ket{x}=\ket{ax}$ for all $x\in \Z_p$.
    \multiplier
    \label{lem:multiplier}
\end{lemma}

 Here, $\left(\frac{a}{p}\right)_L$ denotes the Legendre symbol~\cite[Appendix A]{prakash2021normal}, which is $1$ if $a$ is a quadratic residue modulo $p$, that is, $a = x^2$ for some $x \in \Z_p$, and $-1$ otherwise.

\begin{remark}
For $a\in\Z_p^*$, the scalar factor $\left(\frac{a}{p}\right)_{\!L}\omega^{\frac{-a^2+4a-2}{8a}}$ in \cref{lem:multiplier} is of the form $\pm\omega^t$ with $t\in\Z_p$, so it lies in $\langle -\omega\rangle$. Hence the on-the-nose multiplier $M_a\ket{x}=\ket{ax}$ belongs to the Clifford group $\Clifford_1$ for every $a\in\Z_p^*$. In particular, by \cref{prop:hsquared},
\[
  H^2
  =\lambda_p^{-2}M_{-1}
  =(-1)^{\frac{p-1}{2}}M_{-1},
\]
so that $M_{-1}=(-1)^{\frac{p-1}{2}}H^2$ with $(-1)^{\frac{p-1}{2}}\in\langle-\omega\rangle$. This is why we take $-\omega$, rather than $\omega$, as the scalar generator in \cref{def:Pauli-groups,def:Clifford-circuits}. Had we used $\omega$ instead, then by \cref{lem:det-H,lem:Clifford-syllable-determinant} every element of the resulting group would have determinant in $\{\omega^t \mid t\in\Z_p\}$. Since $p^n$ is odd, $\det(-I_{p^n})=-1$. $-1 \not\in \langle\omega\rangle$, since otherwise $-1=(-1)^p=\omega^{tp}=1$, a contradiction. When $p \equiv 3 \pmod 4$, $M_{-1}=-H^2$ and $\det(M_{-1})=-1$, then $M_{-1}$ would not be in the group generated by $\omega$, $H$, $S$, and $\CZ$, confirming the necessity of using $-\omega$ as the scalar generator.
\end{remark}

We also record several basic properties of the Clifford operators. To begin with, the generators have finite order, $(-\omega)^{2p} = S^p = \CZ^p = H^4 = 1$, and the determinants of $H$, $S$, and $\CZ$ are integer powers of $\omega$. More precisely, $\det(H)=\det(\CZ)=1$, whereas $\det(S)=\omega$ when $p=3$ and $\det(S)=1$ when $p\geq 5$. See \cref{subsec:clifford-properties} for details.

\begin{lemma}
  $\Clifford_n \subseteq \U_{p^n}(\Z[1/p,\omega])$.
  \label{lem:representable}
\end{lemma}    

\begin{proof}
This follows from the fact that the matrices representing $-\omega$, $H$, $S$, and $\CZ$ in the computational basis have entries in $\Z[1/p,\omega]$. In the case of $H$, this follows from \cref{lem:lambda-in-ring}.
\end{proof}

\begin{proposition}
  \label{lem:qupit-Clifford-phase}
  The Clifford phases are exactly the integer powers of $-\omega$. That is, for all $n\in\N$, we have $\Clifford_n \cap \U(1) = \langle -\omega\rangle = \{ \pm\omega^t \mid t \in \Z_{p}\}$.
\end{proposition}

\begin{proof}
    Since $-\omega$ is a generator of $\Clifford_n$ by \cref{def:Clifford-circuits}, we have $(-\omega)^t \in \Clifford_n \cap \U(1)$ for every $t \in \Z_{2p}$, so that $\langle -\omega\rangle \subseteq \Clifford_n \cap \U(1)$. It remains to prove the reverse inclusion. 
    
    All zero-qudit Clifford operators are of the form $(-\omega)^t$ for some $t\in \Z_{2p}$, so that the statement holds when $n=0$. Now let $n\geq 1$. By \cref{lem:det-H,lem:Clifford-syllable-determinant}, the determinants of $H$, $S$, and $\CZ$ are integer powers of $\omega$, and hence lie in $\langle-\omega\rangle$. Moreover, since $p^n$ is odd, the scalar $-\omega = -\omega I_{p^n}$ has determinant $(-\omega)^{p^n}=-1\in\langle-\omega\rangle$. Recall that $\det(AB)=\det(A)\det(B)$ and that $\det(A\otimes B)=\det(A)^{m}\det(B)^{k}$ for a $k\times k$ matrix $A$ and an $m\times m$ matrix $B$. Since $\langle-\omega\rangle$ is closed under products, inverses, and integer powers, it follows that every element of $\Clifford_n$ has determinant in $\langle-\omega\rangle$.

    Now consider a phase $cI_{p^n}\in \Clifford_n$. By \cref{lem:representable}, $c\in \Z[1/p,\omega]$, and by the above there exists $s \in \Z_{2p}$ such that
    \begin{equation}
    \label{eq:dets}
        c^{p^n} = \det(c I_{p^n}) = (-\omega)^s.
    \end{equation}
    Hence, $c^{2p^{n+1}} = (-\omega)^{2ps} = 1$, so $c$ is a root of unity in $\Z[1/p,\omega]$. Therefore, by \cref{lem:roots-of-unity}, $c=\pm\omega^t$ for some $t\in\Z_p$, that is, $c \in \langle -\omega\rangle$. Hence, $\Clifford_n \cap \U(1) \subseteq \langle -\omega\rangle$.
\end{proof}

Equivalently, if $U, V \in \Clifford_n$ satisfy $U = cV$ for some $c \in \U(1)$, then $c \in \langle -\omega\rangle$, since $cI_{p^n} = UV^\dagger \in \Clifford_n$. That is, any global phase by which two Clifford operators differ is an integer power of $-\omega$.

\subsubsection{Clifford Conjugation of Pauli Generators}
\label{subsec:Clifford-conj}

Another way to define the Clifford unitaries is as those unitaries that send the Pauli operators to Pauli operators under conjugation.

\begin{definition}
\label{def:Clifford_unitaries}
    For $n \in \N$, the \emph{normaliser of the Pauli group} $\Pauli_n$ is denoted by $\widetilde{\Clifford_n}$,
    \[
    \widetilde{\Clifford_n} \coloneq \{U \in \U_{p^n} \mid UPU^\dagger \in \Pauli_n \; \mbox{ for all } P \in \Pauli_n\}.
    \]
\end{definition}

We will denote the action of $\widetilde{\Clifford_n}$ by conjugation on the Pauli group $\Pauli_n$ as $C \bullet P \coloneq CPC^\dagger$ for $C\in\widetilde{\Clifford_n}$. Note that if we write $Q \coloneq C\bullet P$, we then also have $CP = QC$, and we can interpret this as ``pushing $P$ through $C$ gives us $Q$''. It will be useful to introduce a circuit notation for this.

\begin{definition}\label{def:conjugation}
    Let $C \in \widetilde{\Clifford_n}$ and $P,Q \in \Pauli_n$ with $Q =C\bullet P$. Write $P = P_1 \otimes \cdots \otimes P_n$ and $Q = Q_1 \otimes \cdots \otimes Q_n$ for $P_j, Q_j \in \Pauli_1$. We write this in circuit notation as follows:
    \[
    \scalebox{.8}{\input{figures/Preliminaries/CliffordConjugation.tikz}}
    \]
    % That is, pushing $P$ through $C$ gives us $Q$: $CP = QC$.
\end{definition}

\begin{example}
    By \cref{eq:XZ}, $X, Z \in \widetilde{\Clifford_1}$. Diagrammatically, we have
    \begin{equation}
        \scalebox{.8}{\input{figures/Preliminaries/Pauli-conj.tikz}}
        \label{eq:Pauli action}
    \end{equation}
\end{example}

A Clifford unitary is fully determined up to global phase by its action of conjugating Pauli operators. 
\begin{proposition}
  \label{prop:scalars}
  Let $C\in\widetilde{\Clifford_n}$. If $C \bullet P = P$ for all $P \in \Pauli_n$, then $C$ is a scalar (i.e.~proportional to the identity).
\end{proposition}
\begin{proof}
    Any complex $p \times p$-matrix can be written in the form $\sum_{i=1}^{p^2}c_iP_i, \; c_i \in \C, \; P_i \in \{X^aZ^b \mid a, b \in \Z_p\}$. It follows that $\Pauli_n$ spans the vector space formed by $p^n\times p^n$ complex matrices. Since $C \bullet P = P$ for all $P \in \Pauli_n$, $C \bullet M = CMC^\dagger = M$ for all operators $M$. Hence $CM = MC$ so that $C$ must commute with all matrices. The only matrices with this property are the scalars.
\end{proof}

Consequently, if $C, D \in \widetilde{\Clifford_n}$ satisfy $C \bullet P = D \bullet P$ for all $P \in \Pauli_n$, then $(D^\dagger C) \bullet P = P$ for all $P \in \Pauli_n$. By \cref{prop:scalars}, $C = \lambda D$ for some $\lambda \in \U(1)$. If moreover $C, D \in \Clifford_n$, then by \cref{lem:qupit-Clifford-phase}, $\lambda \in \langle -\omega \rangle$.

% Finally, we introduce a property of Clifford operators that will be used to establish the existence and uniqueness of the symplectic normal form defined in \cref{sec:phase-free-normal-form}. 

Since this conjugation is a group homomorphism, it suffices to know its action on a set of generators of $\Pauli_n$. A useful set of generators are those with only one $Z$ or $X$ gate acting on a single qudit, and identities everywhere else.

\begin{definition}
    Let $Z_j$ denote the $n$-qudit Pauli $P_1\otimes\cdots\otimes P_{n}$ with $P_\ell = I$ if $\ell \neq j$, and $P_j = Z$. We write $X_j$ for the analogous $n$-qudit Pauli by substituting $Z$ with $X$ in the above definition. Let $\mathcal{B}_{n}\coloneq\{X_j,Z_j \mid 1\leq j\leq n\}$ be the $n$-qudit \emph{Pauli basis} for $\mathcal{P}_n$. An element in $\mathcal{B}_{n}$ is called a \emph{Pauli generator}.
    \label{def:Pauli notation}
\end{definition}

Note that $\mathcal{B}_n\cup\{-\omega\}$ generates $\mathcal{P}_n$. Given any $n$-qudit Clifford operator $C$, the collection of operators $P_j = C\bullet Z_j,Q_j = C\bullet X_j,1\leq j \leq n$ fully determines $C$ (up to a global phase). Moreover, $P_j$ and $Q_j$ inherit the commuting and non-commuting relations from $Z_j$ and $X_j$. By \cref{eq:XZ}, $P_jQ_j = \omega Q_jP_j$. When $j\neq \ell$, $P_jQ_\ell = Q_\ell P_j$.

\cref{lem:automorphism,lem:automorphism-CZ} show that the Clifford generators $H$, $S$, and $\CZ$ are indeed in $\widetilde{\Clifford_n}$ by specifying their action on the Pauli generators under conjugation. Full proofs are given in \cref{sec:sectiontwoproofs}.

\begin{lemma}
  \HSaction
% $H, S \in \Clifford_1$. Specifically, $HXH^\dagger = Z, \quad HZH^\dagger = X^{-1}, \quad SXS^\dagger = XZ, \quad SZS^\dagger = Z$.
\label{lem:automorphism}
\end{lemma}
\begin{equation}
    \scalebox{.8}{\input{figures/Preliminaries/single-qupit-Clifford-auto2.tikz}}
    \label{eq:H-S-auto}
\end{equation}

\vspace{.5 cm}

\begin{lemma}
\label{lem:automorphism-CZ}
$\CZ \in \widetilde{\Clifford_2}$. Specifically,
\begin{equation}
    \scalebox{.7}{\input{figures/Preliminaries/two-qupit-Clifford-auto.tikz}}
    \label{eq:CZ-auto}
\end{equation}
\end{lemma}

Using the path-sum formalism~\cite{amy2019towards,koh2017computing}, we can directly verify that multipliers are Pauli normalising by checking their action on the Paulis $X$ and $Z$: 

\begin{equation}
    M_a Z M_{a^{-1}} = Z^{a^{-1}}, \quad M_a X M_{a^{-1}} = X^a.
    \label{eq:multiplier-auto}
\end{equation}

These results show that $\Clifford_n \subseteq \widetilde{\Clifford_n}$. In fact, these groups are the same up to the inclusion of global phases. This fact is well known~\cite[Theorem~7]{farinholt2014ideal}, but will also follow from our construction of an explicit normal form for any operator in $\widetilde{\Clifford_n}$ in \cref{sec:assemble}.

Note that \cref{def:Clifford_unitaries} includes all possible global phases $e^{i\alpha}$, making $\widetilde{\Clifford_n}$ uncountably infinite for all $n \in \N$.
In contrast, the global phases in $\Clifford_n$ are exactly powers of $-\omega$. The groups $\widetilde{\Clifford_n}$ and $\Clifford_{n}$ are then related as follows. First, note that since global phases commute with everything, we have that the phases form normal subgroups $\U(1)\trianglelefteq \widetilde{\Clifford_n}$ and $\langle -\omega \rangle \trianglelefteq \Clifford_{n}$. The quotients $\widetilde{\Clifford_n}/\U(1)$ and $\Clifford_{n}/\langle -\omega \rangle$ are then well-defined. By the fact above and \cref{lem:qupit-Clifford-phase}, the Clifford group is the normaliser of the Pauli group up to global phases,
\[
\widetilde{\Clifford_n}/\U(1) \ \cong\  \Clifford_{n}/\langle -\omega \rangle.
\]

% --------------------------------------------------------------------
\subsubsection{Symplectic Structure of the Clifford Group}
\label{subsec:symplectic}

\cref{prop:scalars} shows that two Cliffords that have the same action on the Paulis are the same up to global phase. For $C, D \in \widetilde{\Clifford_n}$, we can consider a slightly weaker notion where $C \bullet P = \mu_P (D\bullet P)$ for all Paulis $P$ where $\mu_P$ is some phase that is allowed to vary depending on $P$. In this case, $C=\lambda DQ$ for some global phase $\lambda$ and a Pauli $Q$ with $Q\bullet P = \mu_P P$, which is unique up to a global phase.

Conjugating by a Pauli only multiplies each Pauli by a phase, so once we ignore phases, Cliffords that differ by a Pauli act in the same way. Up to a phase, a Pauli is determined by its $Z$ and $X$ powers on each qudit, and since $Z$ and $X$ commute up to a phase, multiplying two Paulis simply adds these powers. Hence, the Paulis form the vector space $\Z_p^{2n}$ up to phases.

\begin{definition}
    Let $\Z_p^{2n}$ be an even-dimensional vector space over $\Z_p$. The \emph{symplectic inner product} $\Omega(\cdot, \cdot)$ is defined by $\Omega( (\vec a,\vec b), (\vec c,\vec d) ) = \vec a\cdot \vec d - \vec b\cdot \vec c$ where $(\vec a,\vec b),(\vec c,\vec d)\in \Z_p^{2n}$ are arbitrary vectors, and $\cdot$ is the regular dot-product of vectors. We say a matrix $M:\Z_p^{2n}\to \Z_p^{2n}$ is \emph{symplectic} when $\Omega(M \vec v, M\vec w) = \Omega(\vec v,\vec w)$ for all $\vec v,\vec w\in \Z_p^{2n}$. We denote the group of symplectic matrices by $\symplectic{2n,\Z_p}$.
\end{definition}

In this definition, a vector $\vec v = (\vec a,\vec b)$ corresponds to an $n$-qudit Pauli where the $Z$ powers of the Pauli correspond to $\vec a$, and the $X$ powers to $\vec b$. The value of the symplectic inner product $\Omega(\vec v,\vec w) = k$ then corresponds to the Paulis defined by $\vec v$ and $\vec w$ commuting up to a phase $\omega^k$. Since conjugation respects products and preserves this phase, it acts on $\Z_p^{2n}$ as a symplectic matrix.

Moreover, $\Pauli_n \leq \Clifford_n$, since $-\omega$ is a generator of both groups, $Z = H^2SH^2S^{-1}$ by \cref{lem:Z-soundness}, and $X = HSH^2S^{-1}H$ by \cref{lem:X-soundness}. Since $\Clifford_n \leq \widetilde{\Clifford_n}$, every element of $\Clifford_n$ normalises $\Pauli_n$. Hence $\Pauli_n \trianglelefteq \Clifford_n$, and the quotient $\Clifford_n/\Pauli_n$ is well-defined.

\begin{proposition}[\cite{hostens2005stabilizer,gross2006hudson}]
    \label{prop:symplectic-quotient}
    $\Clifford_{n}/\Pauli_{n} \cong \symplectic{2n,\Z_p}$.
\end{proposition}

% \cref{prop:scalars} justifies that two Clifford circuits implementing the same automorphism of the Pauli group differ only by a global phase, so constructing a circuit that matches the desired Pauli action suffices to prove existence (and uniqueness up to phase) of the normal form.

 % It is well known that $\Clifford_{n}/\Pauli_{n} \cong \symplectic{2n,\Z_p}$.

Instead of directly modding out the Paulis, we could first mod out the phases.

\begin{definition}
    Let $\widehat{\Pauli}_n = \Pauli_n/\langle-\omega\rangle$ be the $n$-qudit \emph{projective Pauli group} and $\widehat{\Clifford}_n =\Clifford_n/\langle-\omega\rangle$ be the $n$-qudit \emph{projective Clifford group}.
    \label{def:projective-groups}
\end{definition}
Modding out the phases first does not change the result: $\widehat{\Pauli}_n \trianglelefteq\widehat{\Clifford}_n$ and $\widehat{\Clifford}_n/\widehat{\Pauli}_n \cong \symplectic{2n,\Z_p}$. Moreover, $\widehat{\Clifford}_n$ contains a copy of $\symplectic{2n,\Z_p}$, so every element of $\widehat{\Clifford}_n$ is uniquely a projective Pauli times an element of this copy.

\begin{theorem}[\cite{weil1964certains}]
    $\widehat{\Clifford}_n\cong \symplectic{2n, \Z_p} \ltimes \left(\Z_p\right)^{2n}$.
    \label{thm:semidirect}
\end{theorem}

% Hence, the action of $H$ and $S$ on $\Pauli_1$ can be described by their action on the single-qudit Pauli generators.

% --------------------------------------------------------------------
\subsubsection{Derived Generators}
\label{sssec:derivedgens}

For convenience, \cref{fig:derived-generators2} defines a set of derived Clifford operators. They will be used to compactly present the qudit Clifford relations. The soundness of the phases, the Pauli $X$ and $Z$ gates, as well as the $\SWAP$ gate is established in \cref{lem:gen-scalar,lem:X-soundness,lem:Z-soundness,lem:SWAP-soundness}.

% \begin{figure}[!htb]
%     \[
%     \scalebox{.7}{\tikzfig{figures/Preliminaries/derived-generators0}}
%     \]
%     \caption{Circuit notation for the derived Clifford generators up to Pauli correction and global phase. $g$ generates $\Z_p^*$.}
%     \label{fig:derived-generators0}
% \end{figure}  

\begin{figure}[!htb]
    % The derived generators are named (T1)--(T9) inside the figure, in the order
    % in which they are drawn; these commands give them labels for \eqref.
    \rellabel{T}{eq:def-Mg}\rellabel{T}{eq:def-X}\rellabel{T}{eq:def-Z}%
    \rellabel{T}{eq:def-SWAP}\rellabel{T}{eq:def-CX}\rellabel{T}{eq:def-XC}%
    \rellabel{T}{eq:def-CIZ}\rellabel{T}{eq:def-CIX}\rellabel{T}{eq:def-XIC}%
    \[
    \scalebox{.7}{\input{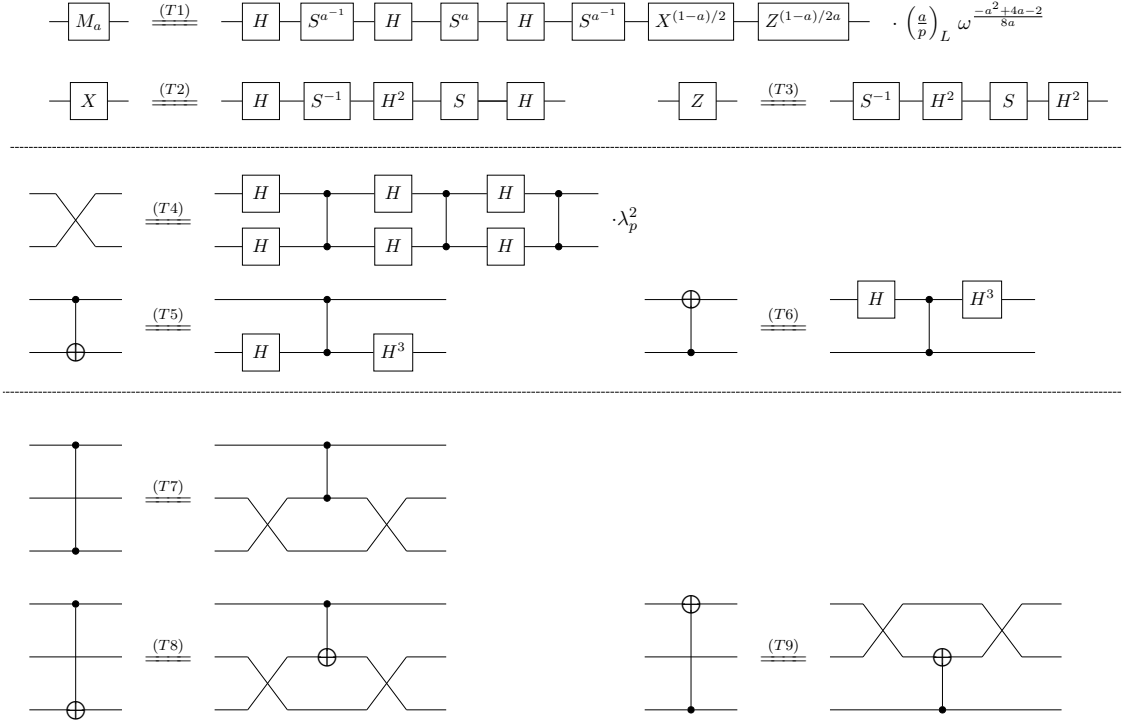}}
    \]
    \caption{A set of derived Clifford generators, expressed as circuits over the generating set $\{-\omega,H,S,\CZ\}$. Here, $a$ is an arbitrary element of the multiplicative group $\mathbb{Z}_p^*$ and $\lambda_p  = e^{(p-1)\pi i/4}$. $\left(\frac{a}{p}\right)_L$ denotes the Legendre symbol~\cite[Appendix A]{prakash2021normal}, which is $1$ if $a$ is a quadratic residue modulo $p$, that is, $a = x^2$ for some $x \in \Z_p$, and $-1$ otherwise.}
    \label{fig:derived-generators2}
\end{figure}    

% --------------------------------------------------------------------
\subsection{Circuits}
\label{ssec:circuits}

% We assume familiarity with the standard graphical notation for quantum circuits. For background, the reader may consult any standard textbook on quantum computation~\cite{NielsenChuang2010}.

In \cref{ssec:operators}, the Clifford generators $-\omega$, $H$, $S$, and $\CZ$, as well as the derived generators listed in \cref{fig:derived-generators2}, were treated as operators. In this subsection, we consider them as generators for circuits. The difference between these two perspectives is that when $X$ is viewed as an operator, we have for instance, $X^p = 1$. However, the circuit $X^p$ is not equal to the circuit of an identity matrix (the former contains $p$ gates, while the latter contains none).

We write $\circuitClifford$ for the collection of all circuits over the gate set $\{ -\omega, X, Z, H, S, \CZ\}$ and $\circuitClifford_n$ for the subset of $\circuitClifford$ containing only the circuits on exactly $n$ qudit wires. We similarly define $\circuitPauli$ and $\circuitPauli_n$ as the collection of circuits over the gate set $\{-\omega, X, Z\}$ and the collection of circuits in $\circuitPauli$ on exactly $n$ qudit wires. Note that $\circuitPauli \subseteq \circuitClifford$ and $\circuitPauli_n \subseteq \circuitClifford_n$. Here, $X$ and $Z$ are abbreviations for the circuits in \eqref{eq:def-X} and \eqref{eq:def-Z}, so every circuit in $\circuitClifford$ is a circuit over $\{-\omega, H, S, \CZ\}$.

% --------------------------------------------------------------------
\subsubsection{Interpretations}
\label{sssec:interpretations}

An \emph{interpretation map} for a family of circuits is a function which assigns an operator to every circuit. The standard interpretation is the \emph{unitary interpretation}, in which a circuit is interpreted as the unitary operator it represents.

\begin{definition}
\label{def:unitaryinterp}
Let $C\in\circuitClifford_n$. We define the \emph{unitary interpretation} $\llbracket C \rrbracket_u \in\Clifford_n$ of $C$ to be the unitary operator obtained by interpreting the gates in $\{ -\omega, X, Z, H, S, \CZ\}$ as the corresponding operators, the parallel composition of circuits as the tensor product of operators, and the sequential composition of circuits as the matrix multiplication of operators.
\end{definition}

Now recall that the quotients $\Clifford_n/\Pauli_n$ and $\Clifford_n/ \langle -\omega \rangle$ are well-defined and let $\pi_s:\Clifford_n \to \Clifford_n/\Pauli_n \cong \symplectic{2n,\Z_p}$ and $\pi_p: \Clifford_n \to \Clifford_n/\langle -\omega \rangle$ be the corresponding canonical projections. We use these projections to define two more interpretation maps.

\begin{definition}
\label{def:other_interp}
Let $C\in\circuitClifford_n$. We define: 
\begin{itemize}
\item the \emph{interpretation of $C$ up to Paulis} $\llbracket C \rrbracket_s \in\Clifford_n/\Pauli_n$ as $\llbracket C \rrbracket_s = \pi_s(\llbracket C \rrbracket_u)$ and
\item the \emph{interpretation of $C$ up to global phases} $\llbracket C \rrbracket_p \in\Clifford_n/\langle -\omega\rangle$ as $\llbracket C \rrbracket_p = \pi_p(\llbracket C \rrbracket_u)$.
\end{itemize}
\end{definition}

We sometimes refer to $\llbracket \;\cdot\; \rrbracket_s$ as the \emph{symplectic interpretation} (since  $\llbracket C \rrbracket_s$ belongs to the symplectic group) and to $\llbracket \;\cdot \;\rrbracket_p$ as the \emph{projective interpretation} (since $\llbracket C \rrbracket_p$ belongs to the projective group). The three interpretations introduced so far organise themselves into the diagram below.

\[
  \scalebox{1}{\input{figures/Maps/maps3.tikz}}
\]

\iffalse
\[
\begin{tikzcd}
	&& \circuitClifford_n \\
	\\
	\Clifford_n/\Pauli_n && \Clifford_n && \Clifford_n/\langle -\omega \rangle
	\arrow["{\llbracket \cdot \rrbracket_s}"', from=1-3, to=3-1]
	\arrow["{\llbracket \cdot \rrbracket_u}"{description}, from=1-3, to=3-3]
	\arrow["{\llbracket \cdot \rrbracket_p}", from=1-3, to=3-5]
	\arrow["{\pi_s}", from=3-3, to=3-1]
	\arrow["{\pi_p}"', from=3-3, to=3-5]
\end{tikzcd}
\]  
\fi

The interpretation maps $\llbracket \;\cdot\; \rrbracket_u$, $\llbracket \;\cdot\; \rrbracket_s$, and $\llbracket \;\cdot\; \rrbracket_p$ correspond to three ways to interpret Clifford circuits, and each interpretation comes with a level of granularity. We can now use these interpretation maps to define different types of equality among the elements of $\circuitClifford$.

\begin{definition}
\label{def:semanticrelations}
Let $C_1,C_2 \in \circuitClifford_n$. We write
\begin{itemize}
\item $C_1 \equiv C_2$ if $C_1$ and $C_2$ are the same circuit.
\item $C_1 \equiv_u C_2$ if $\llbracket C_1\rrbracket_u = \llbracket C_2\rrbracket_u$. 
\item $C_1 \equiv_p C_2$ if $\llbracket C_1\rrbracket_p = \llbracket C_2\rrbracket_p$. 
\item $C_1 \equiv_s C_2$ if $\llbracket C_1\rrbracket_s = \llbracket C_2\rrbracket_s$. 
\end{itemize}
\label{def:equality}
\end{definition}    

As \cref{def:semanticrelations} states: $C_1\equiv C_2$ when $C_1$ and $C_2$ are identical circuits; $C_1\equiv_u C_2$ when $C_1$ and $C_2$ represent the same unitary; $C_1\equiv_p C_2$ when $C_1$ and $C_2$ represent the same unitary up to a global phase; $C_1\equiv_s C_2$ when $C_1$ and $C_2$ represent the same unitary up to a Pauli. Note that we have the following implications
\[
C_1 \equiv C_2 \Rightarrow
C_1 \equiv_u C_2 \Rightarrow
C_1 \equiv_p C_2 \Rightarrow
C_1 \equiv_s C_2.
\]
None of the above implications can be reversed in general. However, we do know, for example, that if $C_1\equiv_s C_2$, then there exists a unique $P\in\Pauli_n$ such that $\llbracket C_1\rrbracket_u = \llbracket C_2\rrbracket_u P$.

Since $\circuitPauli \subseteq \circuitClifford$, all of the above interpretations restrict to $\circuitPauli$, giving well-defined interpretations of Pauli circuits.

% We note that the above interpretations restrict to $\circuitPauli$, i.e., to interpretations of Pauli circuits.

% --------------------------------------------------------------------
\subsubsection{Relations and Rewriting}
\label{sssec:relations}

A \emph{relation} on $\circuitClifford$ (or $\circuitPauli$) is a pair $(C,C')$ of elements of $\circuitClifford$ (or $\circuitPauli$) on the same number of qudits. If $\mathcal{R}$ is a set of relations on $\circuitClifford$, we write $\sim_{\mathcal{R}}$ for the smallest \emph{congruence} on $\circuitClifford$ that contains all of the relations in $\mathcal{R}$. By congruence, we mean an equivalence relation on $\circuitClifford$ that is compatible with composition and tensor product. For two circuits $D$ and $D'$, we have $D\sim_\mathcal{R} D'$ if and only if $D$ can be rewritten into $D'$ by applying the relations from $\mathcal{R}$ in either direction. That is, $D\sim_\mathcal{R} D'$ if and only if there is a sequence of rewrites
\[
D \to D_1 \to \cdots \to D_k \to D'
\]
such that each rewrite is an application of a relation in $\mathcal{R}$ that is applied to some local part of the circuits $D_i$.

Let $\llbracket \;\cdot\; \rrbracket_x$ be one of the three interpretations defined above (i.e., one of $\llbracket \;\cdot\; \rrbracket_u$, $\llbracket \;\cdot\; \rrbracket_p$, and $\llbracket \;\cdot \;\rrbracket_s$) and let $\equiv_x$ be the corresponding equivalence of circuits. We say that a collection $\mathcal{R}$ of relations is 
\begin{itemize}
\item \emph{sound with respect to $\llbracket\; \cdot \;\rrbracket_x$} if $C\sim_\mathcal{R} D$ implies $C\equiv_x D$ for all $C,D\in\circuitClifford$, and
\item \emph{complete with respect to $\llbracket \;\cdot \;\rrbracket_x$} if $C\equiv_x D$ implies $C\sim_\mathcal{R} D$ for all $C,D\in\circuitClifford$.
\end{itemize}

Taking $x = s$, $p$, and $u$ gives three completeness statements: symplectic completeness for $\Clifford_n/\Pauli_n\cong\symplectic{2n,\Z_p}$, projective Clifford completeness for $\Clifford_n/\langle-\omega\rangle$, and unitary Clifford completeness for exact equality in $\Clifford_n$. The symplectic presentation forgets Paulis and the projective presentation forgets global phases, whereas the unitary presentation retains all powers of $-\omega$. These are distinct quotient or phase-lifted presentations, rather than the same relation set interpreted simultaneously at all three levels.

%% file: figures/Maps/maps3.tikz
\begin{tikzpicture}
	\begin{pgfonlayer}{nodelayer}
		\node [style=none] (0) at (0, 4.5) {$\circuitClifford_n$};
		\node [style=none] (1) at (0, -0.75) {$\mathcal{C}_n$};
		\node [style=none] (2) at (-6.25, -0.75) {$\mathcal{C}_n/\mathcal{P}_n$};
		\node [style=none] (3) at (6.85, -0.75) {$\mathcal{C}_n/\langle -\omega\rangle$};
		\node [style=none] (4) at (0, 1.5) {};
		\node [style=none] (5) at (0, 0) {};
		\node [style=none] (6) at (0.2, 2) {$\llbracket \;\cdot\; \rrbracket_u$};
		\node [style=none] (7) at (-4.75, -0.75) {};
		\node [style=none] (8) at (-1, -0.75) {};
		\node [style=none] (9) at (-2.75, -1.5) {$\pi_s$};
		\node [style=none] (10) at (1, -0.75) {};
		\node [style=none] (11) at (5, -0.75) {};
		\node [style=none] (12) at (2.75, -1.5) {$\pi_p$};
		\node [style=none] (13) at (-0.5, 4) {};
		\node [style=none] (14) at (-5.5, 0) {};
		\node [style=none] (15) at (-4.25, 2) {$\llbracket \;\cdot\; \rrbracket_s$};
		\node [style=none] (16) at (0.5, 4) {};
		\node [style=none] (17) at (5.25, 0) {};
		\node [style=none] (18) at (4.25, 2) {$\llbracket \;\cdot\; \rrbracket_p$};
		\node [style=none] (20) at (8, -0.75) {};
		\node [style=none] (24) at (0, 4) {};
		\node [style=none] (25) at (0, 2.5) {};
	\end{pgfonlayer}
	\begin{pgfonlayer}{edgelayer}
		\draw [style=diredge] (4.center) to (5.center);
		\draw [style=diredge] (8.center) to (7.center);
		\draw [style=diredge] (10.center) to (11.center);
		\draw [style=diredge] (13.center) to (14.center);
		\draw [style=diredge] (16.center) to (17.center);
		\draw (24.center) to (25.center);
	\end{pgfonlayer}
\end{tikzpicture}

%% file: scripts/3-assemble.tex
\section{A Unique Normal Form for Multi-Qudit Clifford Circuits}
\label{sec:assemble}

In this section, we define a normal form in $\circuitClifford_n$ for elements in $\Clifford_n$. Similar to the Clifford normal form for qubits and qutrits~\cite{qutrit2024,makary2021generators,selinger2015generators}, our normal form for a qudit Clifford operator is based on uniquely representing its stabiliser tableau. Our construction differs in that we handle the symplectic and Pauli parts separately. In \cref{sec:phase-free-normal-form}, we first define the \emph{symplectic normal form}, as a normal form for elements in $\Clifford_n$ up to Pauli corrections, which is based on the isomorphism $\symplectic{2n, \Z_p} \cong \Clifford_n/ \Pauli_n\cong \widehat{\Clifford}_n / \widehat{\Pauli}_n$. For every symplectic matrix $M\in \symplectic{2n, \Z_p}$, there is a unique circuit $C\in\circuitClifford_n$ in the symplectic normal form such that $\llbracket C\rrbracket_s=M$. In \cref{subsec:phase-free-normal-form-pauli}, we then define a normal form for elements in $\widehat{\Pauli}_n$, and call this the \emph{Pauli normal form}. In \cref{subsec:assemble}, we derive the normal form of elements in $\widehat{\Clifford}_n$ by composing the symplectic and Pauli normal forms, using the semi-direct product structure established in \cref{thm:semidirect}. The normal form for $\Clifford_n$ follows by adding back in global phases.

To define the symplectic normal form, our first goal is to conjugate the Pauli generators to prescribed Pauli operators (up to phases) by using some elementary building blocks in $\circuitClifford_1$ and $\circuitClifford_2$, which we call \emph{normal boxes}. To ensure uniqueness, each Clifford operator must be built from these blocks in exactly one way (up to Pauli correction). We will consider four \emph{types} of normal boxes which we label as $A$, $B$, $D$, and $E$; see \cref{fig:normal-box-auto}. Each type of normal box has a subscript denoting which \emph{variant} it is. 
% Note that in \cref{fig:normal-box-auto}, Pauli propagation is tracked up to a global phase.

 \begin{figure}[!htb]
     \centering
      \[
\scalebox{.90}{\input{figures/NormalForm/normalBoxes-automorphisms-v2.tikz}}
    \]
    \vspace{-0.6cm}
     \caption{The types of normal boxes, and the specified action they must have on Paulis (up to phase). The colour of the box (red or green) specifies whether these are part of the Z-part, respectively the X-part, of the normal form.}
     \label{fig:normal-box-auto}
 \end{figure}

\begin{figure}[!htb]
    \centering
\scalebox{0.8}{\input{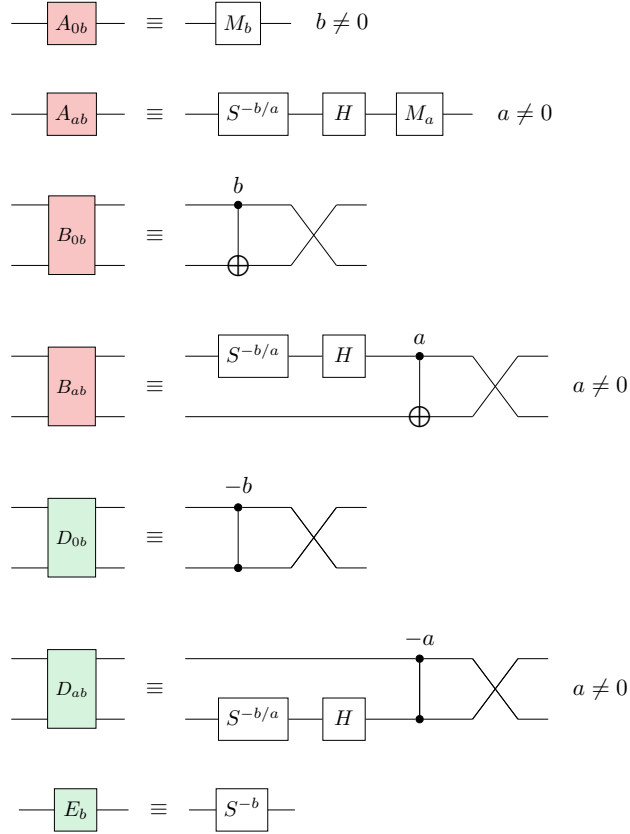}}
% \scalebox{0.8}{\tikzfig{figures/NormalForm/normalBoxesMacroSimp-v4}}
% \centering \scalebox{0.8}{\tikzfig{figures/agda-box-defs-rels/box_def_c}}
    \caption{The concrete implementations of the Z- and X-normal boxes. Here $a,b\in \Z_p$, subject to the conditions on the right of each circuit.}
    \label{fig:Z-and-X-Normal-Boxes-imp}
\end{figure}

The ``required actions'' the boxes need to satisfy in \cref{fig:normal-box-auto} are needed for the normal form construction to work, but they do not uniquely define the boxes. We identify some ``additional action'' that helps simplify some of the arguments later. But even so the symplectic maps of the boxes are not uniquely determined and care must be taken to find concrete implementations that give the simplest possible relations later in the completeness proof.

We give concrete implementations of the normal boxes in \cref{fig:Z-and-X-Normal-Boxes-imp}. Note that as opposed to the normal boxes used in the proofs of qubit and qutrit completeness~\cite{qutrit2024,makary2021generators,selinger2015generators}, these definitions are parametric in the labels of the variants. This feature plays an important role in simplifying the search for box relations later.

\subsection{The Symplectic Normal Form for Multi-Qudit Clifford Circuits}
\label{sec:phase-free-normal-form}

We want to use these normal boxes to define a unique normal form for symplectic matrices, i.e.~stabiliser tableaux up to Pauli correction. This requires us to show the normal form is \emph{universal}, meaning it can represent an arbitrary symplectic matrix $M$, and \emph{unique}, meaning that each $M\in \symplectic{2n, \Z_p}$ is uniquely represented by one specific instantiation of this normal form.

We will show universality by combining together the normal boxes into layers of circuits that we call Z-normal and X-normal. A single layer will ``fix up'' the action of the Paulis on a single fixed qudit. By then inductively doing this for all qudits we reduce it to have the identity action on all the Paulis. This procedure is analogous to Gaussian elimination.

In a bit more detail, let $C$ be a Clifford circuit that we will consider to be an automorphism $\phi$ of $\Pauli_n$ using the conjugation action $\phi(P) = \llbracket C\rrbracket_u\bullet P$.
We claim that we can then construct a $Z$-normal circuit $W_Z$ and an $X$-normal circuit $W_X$ (to be defined later) such that $\left(W_XW_Z\right)^{-1}$ acts as $\phi^{-1}$ on $Z_n$ and $X_n$. 
That is, if we define $C'= C\left(W_XW_Z\right)^{-1}$ and $\phi'(P) = \llbracket C'\rrbracket_u\bullet P$, then $\phi'(Z_n) = Z_n$ and $\phi'(X_n) = X_n$ (up to phase).
By then finding other $W_Z$ and $W_X$ we can keep updating $C$ using some normal circuit $N$ to make $CN^{-1}$ act as the identity on all $Z_j$ and $X_j$. Then by \cref{prop:scalars} $N$ must implement the same linear map as $C$.

% To obtain a normal form for $\phi$, it suffices to construct $W_Z$ and $W_X$ and then to recursively iterate this procedure with the updated automorphism $\phi'=\phi\left(W_XW_Z\right)^{-1}$.

% Now that we have the symplectic normal form, we need to show that  For universality, by \cref{prop:scalars}, it suffices to show that we can represent an arbitrary automorphism of $\Pauli_n$ by a normal form, up to Pauli correction.

%\label{fig:normal-components}
\begin{definition}\label{def:Z-and-X-normal}
    Let the $A$, $B$, $D$, and $E$ normal boxes be the concrete circuits defined in \cref{fig:Z-and-X-Normal-Boxes-imp}. We then define the following types of normal circuits:
    \[\scalebox{0.65}{\input{figures/PhaseFreeNormalForm/ZNormal.tikz}} \qquad 
    \scalebox{0.65}{\input{figures/PhaseFreeNormalForm/XNormal.tikz}}\]
    We call the circuits on the left \emph{Z-normal} and on the right \emph{X-normal}. In these circuits, each normal box can carry any label allowed in \cref{fig:normal-box-auto}, and the $A$ box of a Z-normal circuit may sit on any wire, with a ladder of $B$ boxes above it up to the first wire and no gates below it.
% \begin{itemize}
%     \item an $n$-qudit Clifford circuit is \emph{Z-normal} if it is of the form in \cref{fig:Z-normal}, and
%     \item an $n$-qudit Clifford circuit is \emph{X-normal} if it is of the form in \cref{fig:X-normal}.
% \end{itemize}
\end{definition}

The Z- and X-normal circuits are the building blocks of the symplectic normal form for an $n$-qudit Clifford circuit. A pair of Z- and X-normal circuits suffices to rewrite a pair of rows and columns of a symplectic matrix to the identity.

% \begin{figure}[!htb]
% % \renewcommand{\thefigure}{(a)}
% \begin{subfigure}{.35\textwidth}
%   \centering
%    \scalebox{0.65}{\tikzfig{figures/PhaseFreeNormalForm/ZNormal}}
%   \caption{The Z-normal circuit.}
%   \label{fig:Z-normal}
% \end{subfigure}%
% % \renewcommand{\thefigure}{(b)}
% \qquad
% \begin{subfigure}{.65\textwidth}
%   \centering
%    \scalebox{0.65}{\tikzfig{figures/PhaseFreeNormalForm/XNormal}}
%       \caption{The X-normal circuit.}
%     \label{fig:X-normal}
% \end{subfigure}
% \caption{The Z- and X-normal circuits that are the building blocks of the symplectic normal form for an $n$-qudit Clifford circuit. Here each normal box can have arbitrary labels. A pair of Z- and X-normal circuits suffices to rewrite a pair of rows and columns of a symplectic matrix to the identity.
% % Up to Pauli correction, $W_Z^{(n)}$ gets rid of the $C$ box on the top output wire and $W_X^{(n)}$ gets rid of the $F$ box on the bottom output wire.
% }%They differ from the normal circuits in \cref{fig:normal-components-phaseful} up to some Pauli correcrtion. 
% \label{fig:normal-components}
% \end{figure}

\begin{definition}\label{def:normal}\label{fig:macro-normal-circuit}
An $n$-qudit Clifford circuit is \emph{symplectically normal} if it is of the following form:
% in \cref{fig:macro-normal-circuit}.
    \[
\scalebox{0.8}{\input{figures/PhaseFreeNormalForm/Normal.tikz}}
    \]
    That is, we alternate layers of Z- and X-normal circuits, with each subsequent layer acting on one fewer qudit. 
\end{definition}

Note that we can view the symplectically normal circuits as being defined inductively, with the $n$-qudit symplectic normal form adding a layer of Z- and X-normal circuits to the left of an existing $(n-1)$-qudit symplectic normal form.

% \begin{figure}[!htb]

% \vspace{-.3cm}
%     \caption{The structure of a symplectic normal form: }
%     \label{fig:macro-normal-circuit}
% \end{figure}

\begin{remark}\label{fig:single-qupit-normal}
    When $n=1$, a symplectic normal form is composed of only $A$ and $E$ boxes:
    \[
    \input{figures/PhaseFreeNormalForm/Single-qupit-Normal.tikz}
    \]

     % A generic construction of $N_S^{(1)}$ is given in \cref{fig:single-qupit-normal}.
     Here $(a,b) \in \Z_p \times \Z_p \setminus \{(0,0)\}$ and $c \in \Z_p$ and this normal form can represent a single-qudit Clifford uniquely up to Pauli correction.
% \begin{figure}[!htb]
%     \centering
%     \[
% \tikzfig{figures/PhaseFreeNormalForm/Single-qupit-Normal}
% \]
%     \caption{For $(a,b) \in \Z_p \times \Z_p \setminus \{(0,0)\}$ and $c \in \Z_p$, a generic construction of the symplectic normal form of a single-qudit Clifford operator. It is a Clifford normal form up to Pauli correction.}
%     \label{fig:single-qupit-normal}
% \end{figure}
\end{remark} 

% \begin{remark}
%     In all proofs and statements below, when it is clear from the context, we drop the superscript and the symplectic annotation in $\widetilde{W_Z^{(n)}}$ $\left(\widetilde{W_X^{(n)}}\right)$, and simply write it as $W_Z$ $\left(W_X\right)$.
% \end{remark}    

% Now that we have the symplectic normal form, we need to show that it is \emph{universal}, meaning it can represent an arbitrary symplectic matrix $M$, and \emph{unique}, meaning that each $M\in \symplectic{2n, \Z_p}$ is uniquely represented by one specific instantiation of this normal form. For universality, by \cref{prop:scalars}, it suffices to show that we can represent an arbitrary automorphism of $\Pauli_n$ by a normal form, up to Pauli correction.

% Let $\phi$ be an automorphism of $\Pauli_n$. We claim that one can construct a $Z$-normal circuit $W_Z$ and an $X$-normal circuit $W_X$ such that $\left(W_XW_Z\right)^{-1}$ acts as $\phi^{-1}$ on $Z_n$ and $X_n$. To obtain a normal form for $\phi$, it suffices to construct $W_Z$ and $W_X$ and then to recursively iterate this procedure with the updated automorphism $\phi'=\phi\left(W_XW_Z\right)^{-1}$. 

The proofs of the statements below, which establish the universality and uniqueness of the normal form, are analogous to those for the qutrit case in~\cite{qutrit2024}, so we postpone them to \cref{sec:sectionthreeproofs}.

\begin{lemma}
  \lemznormal
    % For every $P \in \Pauli_n\setminus \{\omega^tI;\;t \in \Z_p\}$ be a non-scalar $n$-qudit Pauli operator, there exists a unique $Z$-normal circuit $W_Z$ such that $\llbracket W_Z\rrbracket_u\bullet P\equiv_p Z \otimes I \otimes \cdots \otimes I$.
\label{lem:Z-normal-reduction}
\end{lemma}

\begin{lemma}
  \lemxnormal
%   For every $Q\in \Pauli_n$ such that $(Z \otimes I \otimes \cdots \otimes I) Q = \omega Q(Z \otimes I \otimes \cdots \otimes I)$, there exists a unique $X$-normal circuit $W_X$ such that $\llbracket W_X\rrbracket_u\bullet Q \equiv_p I \otimes \cdots \otimes I \otimes X$.
\label{lem:X-normal-reduction}
\end{lemma}

\begin{lemma}
  \lemxnormalpropagateZ
% Every $X$-normal circuit $W_X$ satisfies $\llbracket W_X\rrbracket_u\bullet(Z \otimes I \otimes \cdots \otimes I \otimes I)\equiv_p I \otimes I \otimes \cdots \otimes I \otimes Z$.
\label{lem:X-normal-propagate-Z}
\end{lemma}

\begin{lemma}
    \lemnormalcircuitcompo
    % Let $P,Q\in \Pauli_n$ such that $PQ = \omega QP$. Then there exist unique normal circuits $W_Z$ and $W_X$ such that $\llbracket W_XW_Z\rrbracket_u\bullet P  \equiv_p I \otimes \cdots \otimes I \otimes Z$ and $\llbracket W_XW_Z\rrbracket_u\bullet Q  \equiv_p I \otimes \cdots \otimes I \otimes X$.
    \label{lem:normal-circuit-compositon}
\end{lemma}

\begin{proposition}
\propnormalform
% Let $\phi: \Pauli_n \rightarrow \Pauli_n$ be an automorphism of the Pauli group such that $\phi$ fixes scalars. There exists a unique Clifford circuit $C\in\circuitClifford_n$ in symplectic normal form such that for all $P\in \Pauli_n$, $\llbracket C \rrbracket_u\bullet P \equiv_p \phi(P)$.
    \label{prop:normal-form}
\end{proposition}

Since there is a bijection between the symplectic matrices and the normal forms, we can count the number of $n$-qudit normal forms to compute the cardinality of $\symplectic{2n, \Z_p}$. The proof of \cref{lem:cardinality} can be found in \cref{sec:sectionthreeproofs}.

\begin{lemma}
\corcardinality
% $\left\lvert \symplectic{2n, \Z_p} \right\rvert = \prod_{k=1}^n \left(p^{2k}-1\right) \cdot \left(p^{2k-1}\right).$
    \label{lem:cardinality}
\end{lemma}

\subsection{The Normal Form for Pauli Operators}
\label{subsec:phase-free-normal-form-pauli}
Here, we introduce a normal form for multi-qudit Pauli operators, which will serve as a basic building block in \cref{subsec:assemble}. This normal form is unique up to a global phase.

\begin{definition}\label{def:normal-pauli-phasefree}\label{fig:macro-normal-circuit-pauli-phasefree}
    An $n$-qudit circuit in $\circuitClifford_n$ is said to be a \emph{Pauli normal form} if it is of the following form:
    \[
    \scalebox{0.85}{\input{figures/PhaseFreeNormalForm/Normal-Pauli-phasefree.tikz}}
    \]
% form in \cref{fig:macro-normal-circuit-pauli-phasefree}. 
    % \label{def:normal-pauli-phasefree}
\end{definition}

% \begin{figure}[!htb]
% \[
% \scalebox{0.85}{\tikzfig{figures/PhaseFreeNormalForm/Normal-Pauli-phasefree}}
% \]
% \vspace{-.3cm}
%     \caption{The construction of a Pauli normal form.}
%     \label{fig:macro-normal-circuit-pauli-phasefree}
% \end{figure}

\begin{lemma}
    For every element $P\in\Pauli_n$, there exists a unique Pauli normal form $N_P^{(n)}$ such that $\llbracket N_P^{(n)} \rrbracket_u\equiv_p P$.
    \label{lem:pauli-normal-phasefree}
\end{lemma}    

\begin{proof}
This follows from the fact that every $P\in\Pauli_n$ can be uniquely written as $P=(-\omega)^t\bigotimes_{j=1}^nX^{a_j}Z^{b_j}$ for $t\in\Z_{2p}$ and $a_j,b_j\in\Z_p$.
\end{proof}    

\subsection{Composing the Symplectic and Pauli Normal Forms}
\label{subsec:assemble}
We now combine the symplectic and Pauli normal forms to obtain a canonical normal form for arbitrary multi-qudit Clifford circuits. Using the semi-direct product structure of the Clifford group, this composition separates the symplectic action from the residual Pauli correction, yielding a unique representation up to global phase. In what follows, $A;B$ denotes the sequential composition of gates in circuit (diagrammatic) order, meaning that $A$ is applied first and then $B$. This corresponds to the matrix product $BA$.

\begin{definition}\label{fig:CliffordNormal}
    An $n$-qudit Clifford circuit $N^{(n)}$ is a \emph{Clifford normal form} if $N^{(n)} \equiv N_S^{(n)};N_P^{(n)}$ where $N_P^{(n)}$ is a Pauli normal form and $N_S^{(n)}$ is a symplectic normal form: 
    \[
    \scalebox{0.74}{\input{figures/PhaseFreeNormalForm/Normal-compose.tikz}}
    \]
    % as shown in \Cref{fig:CliffordNormal},
\end{definition}

% \begin{figure}[!htb]
% \[
% \scalebox{0.74}{\tikzfig{figures/PhaseFreeNormalForm/Normal-compose}}
% \]
% \vspace{-.3cm}
%     \caption{The construction of a Clifford normal form.}
%     \label{fig:CliffordNormal}
% \end{figure}

\begin{proposition}
For every $n$-qudit Clifford operator $C\in\Clifford_n$, there exists a unique Clifford normal form circuit $N\in \circuitClifford_n$ such that $\llbracket N\rrbracket_u\equiv_p C$.
    \label{prop:compose}
\end{proposition}   

\begin{proof}
We first prove the existence. By \Cref{prop:normal-form}, there is a unique circuit $N_S\in\circuitClifford_n$ in the symplectic normal form such that $\llbracket N_S\rrbracket_u \bullet P \equiv_p C\bullet P$ for all $P\in\Pauli_n$. Hence $\llbracket N_S\rrbracket _s = \pi_s(C)$. This implies there is a unique Pauli $P\in\Pauli_n$ such that $P\cdot \llbracket N_S\rrbracket_u=C$. Then by \Cref{lem:pauli-normal-phasefree}, there is a unique circuit $N_P\in\circuitClifford_n$ such that $\llbracket N_P\rrbracket_u \llbracket N_S\rrbracket_u\equiv_p C$. It follows that $N\equiv N_S;N_P$ is a circuit in Clifford normal form such that $\llbracket N\rrbracket_u\equiv_p C$. 

Now we prove the uniqueness. Suppose $N'\equiv N'_S;N'_P$ is another circuit in $\circuitClifford_n$ in Clifford normal form such that $\llbracket N'\rrbracket_u\equiv_p C$, where $N'_S$ and $N'_P$ are symplectic normal form circuit and Pauli normal form circuit respectively. It follows that $\llbracket N'_S\rrbracket_s = \llbracket N'\rrbracket_s=\pi_s(C)=\llbracket N_S\rrbracket_s$. By the uniqueness of symplectic normal form, we must have $N'_S\equiv N_S$. This also implies $\llbracket N'_S\rrbracket_u=\llbracket N_S\rrbracket_u$. Since $\llbracket N'\rrbracket_u\equiv _p \llbracket N\rrbracket_u$, we must then have $\llbracket N'_P\rrbracket_u\equiv_p \llbracket N_P\rrbracket_u$ as Pauli operators in $\Pauli_n$. By the uniqueness of Pauli normal form, we must have $N'_P\equiv N_P$. We conclude that $N'\equiv N$.
\end{proof}    

% \begin{remark}
%     It is not necessary to distinguish between the precise and ``symplectic'' construction of the normal boxes. We could say let the \RHS of Figure 8.(a) to be $W_X^{(n)}$.
% \end{remark}   

Since each element of $\widehat{\Clifford}_n$ is in bijection with a unique Clifford normal form by \cref{prop:compose}, counting these normal forms yields the cardinality of $\widehat{\Clifford}_n$. Using \cref{lem:qupit-Clifford-phase}, which characterises the Clifford phases as the integer powers of $-\omega$, we then obtain the cardinality of the Clifford group $\Clifford_n$.

\begin{remark}
    By \cref{lem:cardinality,lem:pauli-normal-phasefree} and \cref{lem:qupit-Clifford-phase,prop:compose},
    \begin{equation*}
       \left\lvert \Clifford_n \right\rvert= 2p^{2n+1}\cdot\left\lvert \symplectic{2n, \Z_p} \right\rvert=2p^{2n+1}\prod_{k=1}^n \left(p^{2k}-1\right) \cdot \left(p^{2k-1}\right).
    \end{equation*}
\end{remark}    

%% file: scripts/4-completeness.tex
\section{A Complete Set of Multi-Qudit Clifford Relations}
\label{sec:box-relations}
In this section, we derive a complete set of rewrite rules for multi-qudit Clifford circuits. In \cref{subsec:symplectic-completeness}, we establish completeness at the symplectic level and then in \cref{sec:completenessv1}, we lift these results to unitary Clifford circuits by accounting for Pauli corrections.

\subsection{Establishing Symplectic Completeness}
\label{subsec:symplectic-completeness}
First, we show how to derive a complete set of relations that suffices to transform any multi-qudit Clifford circuit into its unique symplectic normal form. We call this the \emph{Clifford normalisation} process, and it proceeds as follows. Let $C \in \circuitClifford_n$ be a circuit composed of gates $g_1,\ldots, g_m$. Append it on the right with the normal form of the identity circuit ($N_{I}$). Starting from the right-most gate $g_m$, iteratively `push' each gate of $C$ into the normal form, updating it at each step (e.g., from $N_{I}$ to $N'$). This process continues until all gates have been absorbed into the updated normal form $N$:
\[
\scalebox{.7}{\input{figures/Normalization/normalization.tikz}}
\]

% Below, we sketch a simple example. 

Since the Pauli component plays no role in the present discussion, we suppress the distinction between the symplectic normal form $N_S$ and the Pauli normal form $N_P$, and simply write $N$ for the symplectic normal form.

\paragraph{Overview}
In \eqref{eq:identities} and \eqref{eq:n-qudit-identity-normal}, we present the symplectic normal form of the identity operator. Since a Clifford circuit is composed of $H$, $S$, and $\CZ$ gates, it suffices to show how a normal form absorbs these gates and is updated to a new normal form. Locally, this reduces to demonstrating how to push a Clifford gate through each of the normal boxes specified in \cref{fig:Z-and-X-Normal-Boxes-imp}. We call each such instance a \emph{box relation}.

\cref{fig:residual-dirty-gates}(a) gives an abstract illustration of \emph{pushing through a normal box}, which results in an updated normal box $M'$ and \emph{residual dirty gates}, denoted as $\dir$, appearing after $M'$. $M'$ is uniquely determined by $g$ and $M$, which in turn fixes the symplectic action of $\dir$. \cref{fig:residual-dirty-gates}(b) shows a concrete example of pushing an $H$ gate through a $B$ box. 
In the appendix, \cref{lem:HSA,lem:CZ-A,lem:CZ-A-B,lem:HI-B,lem:SI-B,lem:IS-B,lem:IH-D,lem:IS-D,lem:SI-D,lem:CZ-D,lem:CZ-BB,lem:ICZ-B,lem:CZ-DD} show how we can find $M'$ for all cases we need. These results are summarised in \cref{fig:two-qubit-box-relations}.  

\begin{figure}[!htb]
    \centering
    \[
     \scalebox{.9}{\input{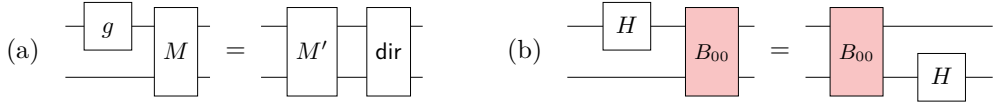}}
    \]
    \vspace{-.3 cm}
    \caption{An abstract and a concrete example of pushing a gate through a normal box $M$.}
\label{fig:residual-dirty-gates}
\end{figure}

In total, we find $42$ parametric box relations, which are listed in \cref{app:relations}. These relations cover every possible case of pushing through a normal box, making them sufficient to normalise any Clifford operator. Combined with \cref{prop:normal-form}, any two Clifford circuits that are equal as symplectic matrices can be rewritten into the same normal form using these box relations. This shows that the box relations are \emph{complete} for multi-qudit Clifford circuits up to Pauli corrections.

\paragraph{Technical details}
We need to show that the Clifford normalisation terminates, and that when this happens there are no more gates left on the right-hand side of the normal form. \cref{fig:dirty-normal-form} provides a schematic summary of what a normal form could look like during the normalisation process. It accounts for all possible cases when pushing through a normal box. 

\begin{definition}\label{def:dirty-normal-form}
We say that a circuit is in \emph{clean symplectic normal form} if it is of the form in \cref{fig:macro-normal-circuit}.
A circuit is in \emph{dirty symplectic normal form} if it is of the form in \cref{fig:dirty-normal-form}, which has the structure of a normal form, but with any number of additional Clifford gates allowed in all the designated spots, subject to the following rules:
\begin{itemize}
    \item $H$ gates can be present at any wire labelled $\circledalt{1}$;
     \item $S$ gates can be present at any wire labelled $\circledalt{1}$ or $\circledalt{2}$;
    % \item $X$ gates can be present at any wire labelled $\circled{2}$;
    % \item $Z$ gates can be present at any wire labelled $\circled{2}$, $\circled{3}$, $\circled{4}$, or $\circled{5}$;
    \item \CZ gates can be present at any pair of adjacent wires, provided that the top wire is labelled $\circledalt{1}$ or $\circledalt{2}$, and the bottom wire is labelled $\circledalt{1}$.
\end{itemize}
In the context of a dirty normal form, we will refer to the normal boxes as \emph{clean}, and all the other $H$, $S$, $X$, $Z$, and $\CZ$ gates as \emph{dirty}.
\end{definition}

\begin{figure}[!htb]
    \centering
     \scalebox{0.7}{\input{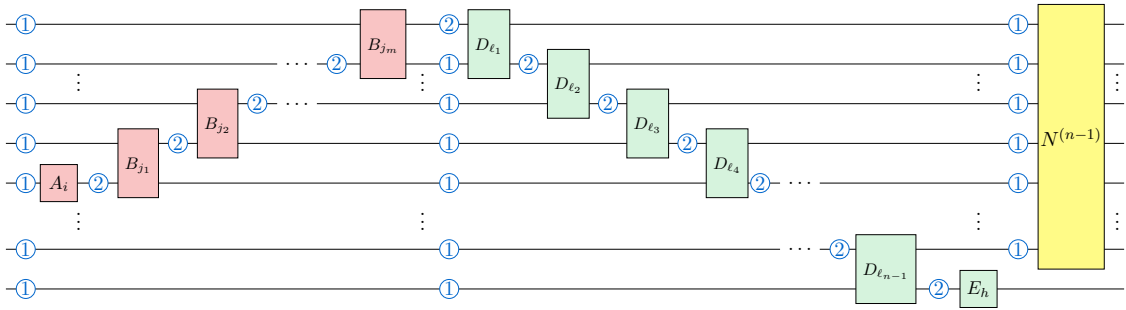}}
    \caption{An example of an $n$-qudit dirty symplectic normal form. Here, $N^{(n-1)}$ is in dirty normal form subject to the same constraints as $N^{(n)}$.}
    \label{fig:dirty-normal-form}
\end{figure}

The placement of dirty gates is restricted by which gates can be properly absorbed by each normal box, as this ultimately determines whether Clifford normalisation will terminate. 
In particular, the important point in \cref{fig:dirty-normal-form} is that there are no dirty gates allowed after the final $E_h$ box or after the $N^{(n-1)}$ so that when the dirty gates get pushed through starting with the left-most, they will eventually get removed completely.

In order to ensure the restrictions of \cref{def:dirty-normal-form} on the allowed dirty gates, we have the additional symmetries of the normal boxes in the third column of \cref{fig:normal-box-auto}. We can then confirm that all dirty gates follow the pattern described in \cref{def:dirty-normal-form} by manually inspecting the box relations in \cref{app:relations}.

\begin{lemma}
    Any dirty normal form can be converted to its symplectic normal form by applying the box relations in \cref{app:relations}, in the left-to-right direction, a finite number of times.
    \label{lem:conversion}
\end{lemma}

The proof is entirely analogous to that of~\cite[Lemma~4.2]{qutrit2024}.

\begin{proposition}
    Consider an $n$-qudit Clifford circuit composed of $H$, $S$, and \CZ gates. Any such circuit can be converted to its symplectic normal form by using the box relations in \cref{app:relations}, together with the equations:
    
    \begin{equation}
    \scalebox{.9}{\input{figures/Normalization/identities.tikz}} \label{eq:identities}
  \end{equation}
    \label{prop:normalization}
\end{proposition}

\begin{proof}
For $1 \leq k \leq n$, let $N_{I_k}$ be the normal form of the $k$-qudit identity operator. First note that 
\begin{equation}\label{eq:n-qudit-identity-normal}
    \scalebox{.7}{\input{figures/Normalization/n-quditIdentityNormal.tikz}}
\end{equation}
Indeed, an identity circuit can be converted to this normal form by applying \eqref{eq:identities} a finite number of times. Appending the given Clifford circuit on the right with the normal form in \eqref{eq:n-qudit-identity-normal}, we obtain a dirty normal form, which can be converted to its normal form by \cref{lem:conversion}.
\end{proof}

% Up to Pauli correction, \cref{prop:normal-form,prop:normalization} show that qudit Clifford operators are presented by the generators $H$, $S$, and $\CZ$, the derived generators defined in \cref{fig:derived-generators2}, the normal boxes given in \cref{fig:Z-and-X-Normal-Boxes-imp}, as well as the $42$ parametric box relations listed in \cref{app:relations}. This presentation is highly redundant, so we provide a reduced set of relations as listed in \cref{fig:rewriterules2}. This rule set remains symplectic complete for multi-qudit Clifford circuits, but it offers a more intuitive framework for us to understand the qudit Clifford gate interactions.

\cref{prop:normal-form,prop:normalization} show that qudit Clifford operators admit a presentation in terms of the generators $H$, $S$, and $\CZ$, together with the derived generators from \cref{fig:derived-generators2}, the normal boxes in \cref{fig:Z-and-X-Normal-Boxes-imp}, and the $42$ parametric box relations listed in \cref{app:relations} (up to Pauli corrections). This presentation is highly redundant, so we extract a reduced set of relations, shown in \cref{fig:rewriterules2}. The resulting rule set remains symplectically complete for multi-qudit Clifford circuits, while providing a more streamlined and intuitive account of qudit Clifford gate interactions.

\begin{figure}[!thb]
    \[
    \scalebox{0.6}{\input{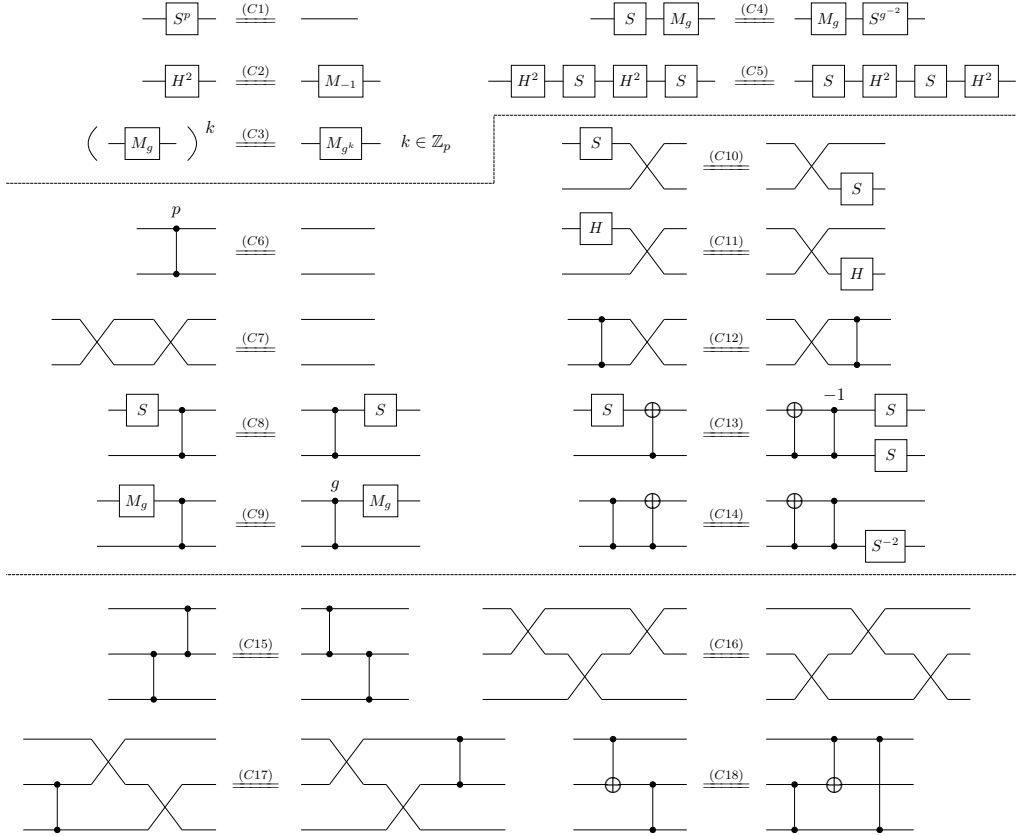}}
    \]
  \caption{A complete set of symplectic relations for $n$-qudit Clifford circuits where $p$ is an odd prime and $n\in \N$. Derived generators such as $M_g$, $X$, $Z$, \SWAP, and \CX are defined in \cref{fig:derived-generators1}.}
  \label{fig:rewriterules2}
\end{figure}

\begin{theorem}
    The $18$ rewrite rules listed in \cref{fig:rewriterules2} are complete for qudit Clifford circuits up to Pauli corrections.
    \label{prop:reduced-relations}
\end{theorem}
\begin{proof}
It suffices to show that these $18$ relations imply the full set of $42$ parametric box relations. All proofs have been formalised, and the corresponding code is available in the GitHub repository~\cite{qupitgit}. Representative examples of this relation reduction, together with the derived relations used as intermediate lemmas, are given in the supplement~\cite{Supplement2026}.
\end{proof}

\subsection{From Symplectic to Unitary Completeness}
\label{sec:completenessv1}

In \cref{subsec:symplectic-completeness}, we showed that Clifford circuits corresponding to the same symplectic matrix can be rewritten into the same normal form using the relations of \cref{fig:rewriterules2}. To lift this result to unitary completeness, we account for the Pauli corrections that are not visible at the symplectic level. By combining symplectic completeness with the Pauli normal form from \cref{subsec:phase-free-normal-form-pauli}, we then obtain a complete and sound rule set of Clifford circuits for the unitary interpretation. 

Our construction will follow a generic group-theoretic approach: given a group $G$ with a normal subgroup $N$ and a presentation of both $N$ and the quotient group $G/N$, one can construct a presentation of $G$ itself~\cite[Proposition~2.55]{holt2005handbook}. This framework applies directly in our setting, since the Pauli group $\Pauli_n$ is a normal subgroup of the Clifford group $\Clifford_n$, and the quotient group $\Clifford_n / \Pauli_n$ is isomorphic to the symplectic group $\symplectic{2n,\Z_p}$. Since, up to global phases, $\Clifford_n$ is a semidirect product of $\symplectic{2n,\Z_p}$ and the Paulis (\cref{thm:semidirect}), lifting symplectic relations to unitary Clifford relations amounts to accounting for the associated Pauli relations, which leads to a particularly simple form of the general construction.

We first recall some notation from \cref{subsec:assemble}. We will use $A;B$ to denote the sequential composition of gates in diagrammatic order, meaning that $A$ is applied first and then $B$. 
For a Clifford circuit $C$, we write $\llbracket C\rrbracket_p$ for its projective unitary representation which ignores the global phase, and $\llbracket C\rrbracket_s$ for its interpretation as a symplectic matrix, i.e.~where we ignore Paulis.
In particular, in the symplectic representation, Paulis are all equal to the identity: $\llbracket P \rrbracket_s = \text{id}$. Recall that if two circuits are equal in the unitary interpretation, then they are also equal in the symplectic interpretation.

\begin{lemma}\label{lem:symplectic-Pauli-correction}
  Let $C_1$ and $C_2$ be two qudit Clifford circuits with equal symplectic representation: $\llbracket C_1\rrbracket_s = \llbracket C_2\rrbracket_s$. Then there exists a unique Pauli $P$ (unique up to a global phase) such that $\llbracket C_1\rrbracket_p = \llbracket C_2;P\rrbracket_p$.
\end{lemma}
In other words: given two Cliffords implementing the same symplectic map, we can find a unique ``Pauli correction'' that makes the Clifford unitaries equal up to a global phase.

\begin{definition}
  Let $\mathcal{R}_s := \{C_l^{(i)} \sim_{\mathcal{R}_s} C_r^{(i)}\}$ be a set of rewrite rules of Clifford circuits. We say it is \emph{symplectically complete} if it is \emph{sound} with respect to the symplectic interpretation, $\llbracket C_l^{(i)}\rrbracket_s = \llbracket C_r^{(i)}\rrbracket_s$ for all $i$, and \emph{complete} for this interpretation: given circuits $C_1$ and $C_2$ such that $\llbracket C_1\rrbracket_s = \llbracket C_2\rrbracket_s$, we can find a sequence of rewrites from $\mathcal{R}_s$ that rewrites $C_1$ to $C_2$.
\end{definition}

\begin{definition}
  Let $\mathcal{R}_p$ be a set of rewrite rules of Clifford circuits. We say it is \emph{Pauli complete} if it is sound for the projective interpretation, and it suffices to prove the commutation and cancellation identities for Paulis as well as how to push Paulis through all the Clifford generators. That is, $\mathcal{R}_p$ must prove the following equations:
  \begin{itemize}
      \item $Z;X \sim_{\mathcal{R}_p} X;Z$
      \item $Z^p \sim_{\mathcal{R}_p} X^p \sim_{\mathcal{R}_p} \text{id}$
      \item $Z;H \sim_{\mathcal{R}_p} H;X^{-1}$ and $X;H \sim_{\mathcal{R}_p} H;Z$
      \item $Z;S \sim_{\mathcal{R}_p} S;Z$ and $X;S \sim_{\mathcal{R}_p} S;Z;X$
      \item $(Z\otimes I);CZ \sim_{\mathcal{R}_p} CZ;(Z\otimes I)$
      \item $(X\otimes I);CZ \sim_{\mathcal{R}_p} CZ;(X\otimes Z)$
  \end{itemize}
  \label{def:pauli-relations}
\end{definition}

\begin{proposition}\label{prop:Pauli-complete}
    The rewrite rules of \cref{fig:rewriterules6} are Pauli complete.
\end{proposition}
\begin{proof}
    For visual clarity in this proof we will write `=' to denote rewriting using the rules of \cref{fig:rewriterules6}.
    There are only two nontrivial equations we need to show follow from \cref{fig:rewriterules6}, $Z;X = X;Z$ and $X;S = S;Z;X$. We first prove all the other ones.
    % $Z;X = X;Z$ follows directly from (C6).

    Define $S' := H^2;S;H^2$ so that $Z=S^{-1};S'$ and \eqref{eq:SH2-SH2-soundness} simplifies to $S';S = S; S'$. Note that combining \eqref{eq:S-soundness} and $H^4 = \text{id}$ we get $(S')^p = \text{id}$.

     $Z^p = 1$ follows from repeated application of \eqref{eq:SH2-SH2-soundness} to rewrite $Z^p = (S';S)^p = (S'^p);S^p = \text{id}$. Since we have $X=H;Z;H^3$ we then also get $X^p = \text{id}$, $Z;H = H;X^{-1}$ and $X;H = H;Z$. 

     $Z;S = S;Z$ follows directly from \eqref{eq:SH2-SH2-soundness}. To prove $(Z\otimes I);CZ = CZ;(Z\otimes I)$ first note that $(S'\otimes I);CZ = CZ;(S'\otimes I)$ by using \eqref{eq:H2-soundness} to convert the $H^2$ in $S'$ to $M_{-1}$, \eqref{eq:CZ-Mg-soundness} to push this through the CZ, and then \eqref{eq:CZ-S-soundness} to push the $S$, then pushing another $M_{-1}$ through CZ using \eqref{eq:CZ-Mg-soundness}, and finally converting these back to $H^2$ using \eqref{eq:H2-soundness}. Then $(Z\otimes I);CZ = CZ;(Z\otimes I)$ easily follows as well. That $(X\otimes I);CZ = CZ;(X\otimes Z)$ follows similarly, but using \eqref{eq:CX-CZ-soundness} instead.

     Next, we prove $X;S = S;Z;X$. From \eqref{eq:Mg-soundness} with $k=0$ we get $M_1 = \text{id}$. Combining this with definition \eqref{eq:def-Mg} for $a=1$ we get $\text{id} = H;S;H;S;H;S$ which we can rewrite to $H;S;H;S = S^{-1};H^3$. We can then write $X;S = H;S^{-1};H;(H;S;H;S) = H;S^{-1};H;S^{-1};H^3$. On the other hand, using $S;Z=S'$ we get 
     \begin{align*}
         S;Z;X \ &=\  S';X \ =\  H^2;S;H^2;H;S^{-1};H^2;S;H \\
         &= \ H^2;S;H;(S')^{-1};S;H \\ 
         &= \ H^2;S;H;S;(S')^{-1};H \\
         &= \ H;(H;S;H;S);H^2;S^{-1};H^3 \\
         &= \ H;S^{-1};H^3;H^2;S^{-1};H^3 \\
         &= \ H;S^{-1};H;S^{-1};H^3,
     \end{align*}
     which matches the expression we found for $X;S$.

     Finally, to prove $Z;X = X;Z$ we will actually prove $Z;X;Z^{-1};X^{-1} = \text{id}$ which suffices. Using \eqref{eq:def-Mg} with $a=-1$ and \eqref{eq:H2-soundness} gives us $H^2 = M_{-1} = H;S^{-1};H;S^{-1};H;S^{-1};X;Z^{-1}$ where for the Pauli we get $(1-a)/2 = 1$ and $(1-a)/(2a) = -1$. Bringing everything to the other side gives us $X;Z^{-1} = S;H^3;S;H^3;S;H$. Further using $Z=H^2;S;H^2;S^{-1}$ and $X^{-1} = H^3;S^{-1};H^2;S;H^3$ we calculate:
     \begin{align*}
         Z;(X;Z^{-1});X^{-1} \ &= \ H^2;S;H^2;S^{-1};S;H^3;S;H^3;S;H;H^3;S^{-1};H^2;S;H^3 \\
         &= \ H^2;S;H^2;H^3;S;H^3;S;S^{-1};H^2;S;H^3 \\
         &= \ H^2;S;H;S;H^3;H^2;S;H^3 \\
         &= \ H^2;S;H;S;H;S;H^3 \\
         &= \ H;(H;S;H;S;H;S);H^3 \\
         &= \ H;H^3 \\
         &= \ \text{id} \qedhere
     \end{align*}
\end{proof}

\begin{proposition}\label{prop:combine}
  Let $\mathcal{R}_s := \{C_l^{(i)} \sim_{\mathcal{R}_s} C_r^{(i)}\}$ be a symplectically complete set of rewrite rules and $\mathcal{R}_p$ a Pauli complete set of rewrite rules. Then we can efficiently find a modified set of rewrite rules $\mathcal{R}'_s := \{C_l^{(i)} \sim_{\mathcal{R}'_s} C_r^{(i)};P^{(i)}\}$ where the $P^{(i)}$ are uniquely defined Paulis such that $\mathcal{R}'_s$ is sound in the projective interpretation and $\mathcal{R}'_S \cup \mathcal{R}_p$ is sound and complete for the projective interpretation.
\end{proposition}
\begin{proof}
  By Lemma~\ref{lem:symplectic-Pauli-correction} we can find unique Paulis to make the equations of $\mathcal{R}_s$ sound in the projective interpretation by changing them to $\mathcal{R}'_s := \{C_l^{(i)} \sim_{\mathcal{R}'_s} C_r^{(i)};P^{(i)}\}$. Hence the equations $\mathcal{R}'_S \cup \mathcal{R}_p$ are indeed all sound for the projective interpretation. 
  Now, to prove completeness let us assume we have two circuits $C_1$ and $C_2$ that are equal in the projective interpretation: $\llbracket C_1\rrbracket_p = \llbracket C_2\rrbracket_p$. We need to find a set of rewrites from $\mathcal{R}'_S \cup \mathcal{R}_p$ that rewrites $C_1$ into $C_2$.
  
  Note that $C_1$ and $C_2$ are also necessarily equal in the symplectic interpretation. Hence, there is a series of rewrite rules from $\mathcal{R}_s$ that rewrites $C_1$ to $C_2$ in a symplectically sound way. Let us consider the first two rewrites in this sequence: $C_1 \rightsquigarrow C_1^{'} \rightsquigarrow C_1^{''}$ and suppose the rewrites used were $C_l^{(i)} \sim_{\mathcal{R}_s} C_r^{(i)}$ and $C_l^{(j)} \sim_{\mathcal{R}_s} C_r^{(j)}$. These rewrites were potentially applied locally in some larger context. 

  Now instead of $C_l^{(i)} \sim_{\mathcal{R}_s} C_r^{(i)}$ we must use the projectively sound modification $C_l^{(i)} \sim_{\mathcal{R}'_s} C_r^{(i)}; P^{(i)}$. The \LHS is still the same, and hence the rule still matches to $C_1$, but now we do not end up with $C_1^{'}$, but instead a modified circuit that has $P^{(i)}$ appearing at the end of the application of $C_r^{(i)}$. Using the rules from the Pauli-complete $\mathcal{R}_p$ we can now push $P^{(i)}$ further to the right in the modified $C_1^{'}$ past all the generators. This changes $P^{(i)}$ to a potentially different Pauli, but otherwise does not change any of the structure of $C_1^{'}$. At the end we hence end up with a circuit $C_1^{'};Q^{(i)}$ where the $Q^{(i)}$ is uniquely determined by $P^{(i)}$. We hence see that our rewrite rules prove $C_1 \rightsquigarrow C_1^{'};Q^{(i)}$.

  Now, the second rewrite rule $C_l^{(j)} \sim_{\mathcal{R}_s} C_r^{(j)}$ is modified to the projectively sound $C_l^{(j)} \sim_{\mathcal{R}'_s} C_r^{(j)};P^{(j)}$. The \LHS is still the same, and we still have $C_1^{'}$ (the only change is that it is composed with $Q^{(i)}$, but this does not change the fact whether the \LHS matches or not). Hence, we can apply this modified rewrite rule to $C_1^{'};Q^{(i)}$. This then results in a modified $C_1^{''}$ where there is the Pauli $P^{(j)}$ composed after the matching $C_r^{(j)}$. In the same way as above we can push this to the end of the circuit to end up with $C_1^{''};Q^{(j)}; Q^{(i)}$. 

  Repeating this procedure for all the rewrite rules from $\mathcal{R}_s$ that bring us from $C_1$ to $C_2$ gives us a set of now projectively sound rewrites from $\mathcal{R}'_S \cup \mathcal{R}_p$ that rewrites $C_1$ to $C_2; P$ where $P$ is some big collection of Paulis. Now, because we know that $C_1$ and $C_2$ are projectively equal, we must have $\llbracket P\rrbracket_u = \text{id}$, and so $\mathcal{R}_p$ must be able to rewrite $P$ to the identity (empty) circuit. This then gives us a rewrite from $C_1$ to $C_2; P$ and finally to $C_2$.
\end{proof}

This statement allows us to extend a symplectically complete set of rewrites, like that of \cref{prop:reduced-relations}, to a projectively complete set of rewrite rules. 

\begin{theorem}
    \completeness
    % The rewrite rules in \cref{fig:rewriterules3} are \emph{complete} for $n$-qudit Clifford circuits over any qudit dimension $p$ that is an odd prime. That is, given two $n$-qudit unitary Clifford circuits $C_1$ and $C_2$ implementing the same linear map, there is a sequence of rewrites from \cref{fig:rewriterules3} that proves $C_1$ and $C_2$ are equal.
    \label{thm:completeness}
\end{theorem}
\begin{proof}
Soundness is established by \cref{prop:soundness}. For completeness, \cref{prop:reduced-relations} shows the rules are symplectically complete and \cref{prop:Pauli-complete} shows they are Pauli complete. Hence, by \cref{prop:combine} they are projectively complete. As in the proof of \cref{prop:combine}, we can apply a similar procedure to boost this to a unitarily complete set of rewrite rules: We have $\widehat{\Clifford}_n \cong \Clifford_n/\langle-\omega\rangle$ where $\langle-\omega\rangle$ is a normal subgroup of $\Clifford_n$. Hence, it suffices to find a complete set of rewrite rules for the group $\langle-\omega\rangle$ and for ``pushing'' through the generators of $\Clifford_n$ and then we can boost a complete set of rewrite rules of $\widehat{\Clifford}_n$ to a complete set of rewrite rules of $\Clifford_n$. A complete set of rewrite rules for $\langle-\omega\rangle$ is easily found: we only need $(-\omega)^{2p} = 1$, which is \eqref{eq:omega-soundness}. Pushing $-\omega$ through generators is equally trivial since it commutes with everything.
% By direct computation, the rewrite rules in \cref{fig:rewriterules3} imply all Pauli relations listed in \cref{def:pauli-relations}. By \cref{prop:combine}, the rules in \cref{fig:rewriterules3} form a complete rewrite system for $n$-qudit Clifford circuits.
\end{proof}

%% file: scripts/5-conclusion.tex
\section{Conclusion and Open Problems}
\label{sec:conclusion}
We found a complete and finite equational theory for $n$-qudit Clifford circuits in every odd prime dimension, by presenting a small set of local rewrite rules that suffices to determine equality of Clifford unitaries. Our approach proceeds by introducing a symplectic normal form for Clifford circuits, deriving relations at the symplectic level, and then lifting these relations to circuit identities by accounting for Pauli corrections.

An immediate open question is to extend this completeness framework beyond odd prime dimensions, in particular to prime-power or arbitrary dimensions, where the arithmetic over $\Z_p$ and the behaviour of Clifford phases are more intricate. More broadly, it would be interesting to develop compact circuit-level equational theories for larger circuit fragments, such as Clifford–permutation–dihedral circuits. In the meantime, it will be fruitful to extend our software infrastructure as a verified backend for automated optimisation and synthesis of qudit circuit families.

\paragraph{Acknowledgements}
The authors would like to thank Yuan Feng, Michele Mosca, Maris Ozols, and Peter Selinger for enlightening discussions. We are especially thankful to Colin Blake for pointing out that three relations in an earlier version of our rule set, namely the symmetry of $\CZ$ under $\SWAP$, the interaction between $\CZ$ and $\CX$ on two qudits, and the commutation of two $\CZ$ gates sharing a qudit, are derivable from the remaining rules. This allowed us to reduce \cref{fig:rewriterules6} from $19$ to $16$ relations.

The circuit diagrams in this paper were typeset using TikZiT~\cite{tikz}. XB acknowledges support from the National Natural Science Foundation of China under Grant 92465202. JvdW acknowledges support from an NWO Veni grant. YZ~is supported by VILLUM FONDEN via QMATH Centre of Excellence grant number 10059 and Villum Young Investigator grant number 37532.

\paragraph{AI Disclosure}
The majority of the work that went into this paper was carried out by the authors in autumn 2025 and winter 2025--26, before the current era of AI-assisted mathematics. During this period, we generalised the relevant algebraic results on the qubit Clifford group to odd prime dimensions, identified the connection between the Clifford group and its symplectic representation, and designed the normal boxes that construct the symplectic and Pauli normal forms. We found a complete set of parametric box relations before December 2025, and reduced them to the symplectic relations in \cref{fig:rewriterules2} during the 2025 Christmas holiday. Finally, between January and March 2026, we showed that the lifted relations (i.e., the Pauli- and phase-corrected versions of the relations in \cref{fig:rewriterules2}) follow from \cref{fig:rewriterules6}. In April 2026, we explored a more symmetric design of the $B$ and $D$ boxes to make the box relations simpler. All novel theoretical results in this paper were developed by the authors without AI assistance.

The Agda formalisation is a minor contribution of this paper. The authors carried out its main verification steps without AI assistance and showed that the box relations are derivable from relations similar to \cref{fig:rewriterules6}. This is the most valuable part of the verification. Anthropic Claude Fable 5.1 and later versions were then used to connect all the steps into a formal proof of the final completeness theorem. This involved composing box relations to merge gates into normal forms, defining the symplectic group, and establishing the equivalence between the relations in the code and \cref{fig:rewriterules6}.

The manuscript was written entirely by the authors. OpenAI GPT-5.6 and GPT-6, accessed through ChatGPT from June 2026 onwards, were used for editing suggestions, and Claude was used to resolve \LaTeX{} compilation issues. We take full responsibility for the correctness of all results.

%% file: scripts/appendix/sectiontwoproofs.tex
% !TEX root = ../../qudit-quantum.tex
\section{Proofs from \texorpdfstring{\cref{sec:foundations}}{Section~2}}
\label{sec:sectiontwoproofs}

In this appendix, we provide proofs that were omitted from \cref{sec:foundations} and document algebraic results related to the qudit Clifford group. Recall the summation of a geometric sequence for $a_n = a_1q^{n-1}$ when $q\neq 1$ is 
\begin{equation}
    S_n= \frac{a_1(1-q^n)}{1-q}.
\label{eq:sum-geo-sequence}
\end{equation}

\begin{lemma}
\label{lem:summation}
    Let $k,n \in \Z$ and $\omega$ be the primitive $n$-th root of unity, $\omega = e^{2\pi i/n}$.
    \begin{equation}
        \sum_{\ell = 0}^{n-1}\omega^{k\ell} = n\delta_n(k),\quad\text{where}\quad \delta_n(k) = \begin{cases}
        0, & k\not\equiv 0 \pmod{n}\\
        1, & k \equiv 0 \pmod{n}.
    \end{cases}
    \label{eq:summation}
    \end{equation}
\end{lemma}

\begin{proof}
When $k\equiv 0 \pmod{n}$, $\sum_{\ell = 0}^{n-1}\omega^{k\ell} = \sum_{\ell = 0}^{n-1}\omega^0 = n$. When $k\not\equiv 0 \pmod{n}$, $\omega^k \neq 1$. By \eqref{eq:sum-geo-sequence},
\[
\sum_{\ell = 0}^{n-1}\omega^{k\ell} = \frac{1\left(1-(\omega^k)^{n}\right)}{1-\omega^k}=0.\qedhere
\]
\end{proof}

The orthogonality relation \eqref{eq:summation} in \cref{lem:summation} is the key algebraic identity underpinning the computations throughout this section. Whenever we evaluate the action of a generator on a computational basis state, the resulting path sum collapses via this identity: every term vanishes except those whose exponent is congruent to $0$ modulo $p$, isolating the single surviving basis state. We therefore invoke \cref{lem:summation} repeatedly in what follows, for instance in establishing the actions of $H^2$ and of the derived generators $X$, $\CX$, and $\SWAP$ on the computational basis.

\subsection{Properties of the Qudit Pauli Generators}
\label{subsec:pauli-properties}

In this subsection, we study the algebraic properties of the qudit Pauli generators $X$ and $Z$, working directly from their definition on the computational basis given in \cref{def:pauli}. Starting from the actions $X\ket{j} = \ket{j+1}$ and $Z\ket{j} = \omega^j\ket{j}$, we derive their orders, their commutation relation, and the induced conjugation identities that will be used throughout the remainder of this section.

\begin{lemma}
    $X^p = Z^p = 1.$
\label{lem:Pauli-order}
\end{lemma}    

\begin{proof}
For any $0\leq j \leq p-1$, we have $X^p\ket{j} = \ket{j+p} = \ket{j}$ and $Z^p\ket{j} = \omega^p\ket{j} = \ket{j}$.
\end{proof}

% Here is the proof that $ZXZ^\dagger X^\dagger =\omega$:
\begin{lemma}
    $ZXZ^\dagger X^\dagger =\omega$.
\label{lem:omega}
\end{lemma} 

\begin{proof}
For $0\leq j \leq p-1$, we have
\[
ZXZ^\dagger X^\dagger \ket{j} = ZXZ^\dagger \ket{j-1} = ZX\omega^{-(j-1)}\ket{j-1} = Z\omega^{-(j-1)}\ket{j} = \omega^{-(j-1)+j}\ket{j} = \omega\ket{j}.
\]
\end{proof}

\begin{lemma}
For all $k \in \Z$, $X^k\ket{j} = \ket{j+k}$ and $Z^k\ket{j} = \omega^{kj}\ket{j}$, where addition and multiplication are taken modulo $p$.
\label{lem:X-Z-copy}
\end{lemma}    

\begin{proof}
Since $Z = \diag(1,\omega,\ldots,\omega^{p-1})$, $Z^k = \diag(1,\omega^k,\ldots,\omega^{k(p-1)}) = \sum_{j=0}^{p-1} \omega^{kj} \ket{j}\bra{j}$. Hence, $Z^k\ket{j} = \omega^{kj}\ket{j}$. For $X^k\ket{j} = \ket{j+k}$, we prove this by induction on $k$. The base case $k=1$ is true by definition. Assume that $X^k\ket{j} = \ket{j+k}$ for some $k \geq 1$. Then
\[
X^{k+1}\ket{j} = X\left(X^k\ket{j}\right) = X\left(\ket{j+k}\right) = \ket{j+k+1}.
\]
This completes the proof.
\end{proof}    

\begin{lemma}
For all $a,b \in \Z_p$, $Z^aX^b(Z^\dagger)^a = \omega^{ab}X^b$ and $X^bZ^a(X^\dagger)^b = \omega^{-ab}Z^a$.
\label{lem:anticommutation}
\end{lemma}

\begin{proof}
We first show the soundness of \eqref{eq:ZX-1} and \eqref{eq:ZX-2}. Then we will combine both to show \eqref{eq:ZX-3}. 

\begin{align}
Z^aX(Z^\dagger)^a &= \omega^{a}X\label{eq:ZX-1}\\
ZX^bZ^\dagger &= \omega^{b}X^b\label{eq:ZX-2}\\
Z^aX^b(Z^\dagger)^a &= \omega^{ab}X^b\label{eq:ZX-3}
\end{align}

For \eqref{eq:ZX-1}, we proceed by induction on $a$. The base case $a=1$ is true by \eqref{eq:XZ}. Assume that \eqref{eq:ZX-1} holds for some $a\geq 1$. Then 

\[
Z^{a+1}X(Z^\dagger)^{a+1} = Z^aZXZ^\dagger(Z^\dagger)^a = Z^a\omega X(Z^\dagger)^a \xlongequal[\IH]{\eqref{eq:ZX-1}} \omega Z^aX(Z^\dagger)^a = \omega \cdot \omega^a X = \omega^{a+1}X.
\]

For \eqref{eq:ZX-2}, we proceed by induction on $b$. The base case $b=1$ is true by \eqref{eq:XZ}. Assume that \eqref{eq:ZX-2} holds for some $b\geq 1$. Then 

\[
ZX^{b+1}Z^\dagger = \left(ZX^bZ^\dagger\right) \left(ZXZ^\dagger\right) \xlongequal[\IH]{\eqref{eq:ZX-2},\eqref{eq:XZ}} \left(\omega^b X^b\right) \left(\omega X\right) =  \omega^{b+1} X^{b+1}.
\]

For \eqref{eq:ZX-3}, we proceed by induction on $a$. The base case $a=1$ is true by \eqref{eq:ZX-2}. Assume that \eqref{eq:ZX-3} holds for some $a\geq 1$. Then 

\[
Z^{a+1}X^b(Z^\dagger)^{a+1} = Z^{a}\left(Z X^b Z^\dagger\right) (Z^\dagger)^{a} \xlongequal[]{\eqref{eq:ZX-2}} Z^a\left(\omega^b X^b\right)(Z^\dagger)^a \xlongequal[\IH]{\eqref{eq:ZX-3}} \omega^b \cdot \omega^{ab} X^b = \omega^{(a+1)b}X^b.
\]

Right-multiplying both sides of \eqref{eq:ZX-3} by $Z^a$ yields
\begin{equation}
  Z^aX^b = \omega^{ab}X^bZ^a.
  \label{eq:ZX-4}
\end{equation}  

Right-multiplying both sides of \eqref{eq:ZX-4} by $(X^\dagger)^b$ yields \eqref{eq:ZX-5}. This completes the proof.
\begin{equation}
  Z^a = \omega^{ab}X^bZ^a(X^\dagger)^b \iff \omega^{-ab}Z^a = X^bZ^a(X^\dagger)^b.
  \label{eq:ZX-5}
\end{equation}  
\end{proof}  

For convenience, we visualise \cref{lem:anticommutation} as follows.

\begin{equation}
    \scalebox{1}{\input{figures/Preliminaries/pauli-pushing.tikz}}
    \label{eq:pauli-pushing}
\end{equation}

\subsection{Properties of the Qudit Clifford Generators}
\label{subsec:clifford-properties}

In this subsection, we study the algebraic properties of the Clifford generators $H$, $S$, and $\CZ$, together with the normalisation phase $\lambda_p$. We establish, among other identities, the order and determinant of $H$, the action of $H^2$ on the computational basis, and the way $H$ and $S$ conjugate the Pauli generators. These properties form the algebraic backbone on which the soundness arguments of the later sections rest: when we verify that the derived generators and the rewrite relations act as intended, we repeatedly appeal to the identities collected here.

Recall the definition of these basic generators $-\omega$, $H$, $S$, and $\CZ$:
\begin{align}
(-\omega) \ket{\cdot} &=  -\omega\ket{\cdot}\label{eq:omega}\\
H  \ket{j} &= \frac{1}{\lambda_p\sqrt{p}}\sum_{\ell=0}^{p-1}\omega^{j\ell}\ket{\ell}\label{eq:H}\\
S \ket{j} &=  \omega^{\frac{j(j-1)}{2}} \ket{j}\label{eq:S}\\
\CZ \ket{j,\ell} &=  \omega^{j \ell} \ket{j,\ell}\label{eq:CZ}
\end{align}  

\begin{T13}
\lambdainring
\end{T13}    

\begin{proof}
This follows from the results in \cite{prakash2021normal}. By the quadratic Gauss sum, $\sum_{j}\omega^{j^2}=\left(\frac{2}{p}\right)_L\lambda_p^{-1}\sqrt{p}$. Here, $\left(\frac{2}{p}\right)_L$ is the Legendre symbol, which is $1$ if $p\equiv 1$ or $7 \pmod{8}$, and $-1$ if $p\equiv 3$ or $5 \pmod{8}$. Since $\left(\frac{2}{p}\right)_L^2 = 1$,

\[
\frac{1}{\lambda_p\sqrt{p}}=\frac{\lambda_p^{-1}\sqrt{p}}{p}=\frac{\sum_{j}\omega^{j^2}}{\left(\frac{2}{p}\right)_L\cdot p}=\frac{\sum_{j}\omega^{j^2} \cdot \left(\frac{2}{p}\right)_L}{p}\in\Z[1/p,\omega].
\]
\end{proof}

\begin{T8}
\rootsofunity
\end{T8}    

\begin{proof}
    By Theorem 5.3 in \cite{extensions}, 
    the roots of unity in $\Q(\omega)$ are exactly the complex numbers of the form $\pm \omega^t$, $t \in \Z_p$. Since $\Z[1/p,\omega] \subset \Q(\omega)$, this completes the proof.
\end{proof}

\begin{lemma}
Let $p$ be an odd prime and $\omega = e^{2\pi i/p}$.
    \begin{align}
        -1 & = (-\omega)^p \reltag{D}{eq:neg-one-omega-power}\\
        \omega & = (-\omega)^{p+1} \reltag{D}{eq:omega-neg-omega-power}
    \end{align}
\label{lem:gen-scalar}
\end{lemma}

\begin{proof}
\begin{align*}
    \eqref{eq:neg-one-omega-power}.\RHS &= (-\omega)^p = -\omega^p = -1=\eqref{eq:neg-one-omega-power}.\LHS.\\
    \eqref{eq:omega-neg-omega-power}.\RHS &= (-\omega)^{p+1}  = (-\omega)^p (-\omega) \xlongequal[]{\eqref{eq:neg-one-omega-power}} (-1)(-\omega)=\omega=\eqref{eq:omega-neg-omega-power}.\LHS. \qedhere
\end{align*}
\end{proof}

\begin{T9}
\Hsquare
\end{T9}    

\begin{proof}
We have
\begin{align*}
H^2 \ket{j} &= H\left( \frac{1}{\lambda_p\sqrt{p}} \sum_{k=0}^{p-1} \omega^{kj} \ket{k} \right) \\
                 &= \frac{1}{\lambda_p\sqrt{p}} \sum_{k=0}^{p-1} \omega^{kj} H\ket{k} \\
                 &= \frac{1}{\lambda_p\sqrt{p}} \sum_{k=0}^{p-1} \omega^{kj} \left( \frac{1}{\lambda_p\sqrt{p}} \sum_{\ell=0}^{p-1} \omega^{\ell k} \ket{\ell} \right) \\
                 &= \frac{1}{\lambda_p^{2}p} \sum_{\ell=0}^{p-1} \left(\sum_{k=0}^{p-1} \omega^{k(j+\ell)}  \right) \ket{\ell}.
\end{align*}
Now consider the sum $\sum_{k=0}^{p-1} \omega^{k(j+\ell)}$. When $\ell \equiv -j \pmod{p}$, then $\omega^{k(j+\ell)}=1$, so that the sum is equal to $p$. When $\ell \not\equiv -j \pmod{p}$, then $\omega^{j+\ell}$ is a primitive $p$-th root of unity, so that the sum is equal to 0. Hence,
\[
H^2\ket{j} = \frac{1}{\lambda_p^{2}p} \sum_{\ell=0}^{p-1} \left(\sum_{k=0}^{p-1} \omega^{k(j+\ell)}  \right) \ket{\ell} = \frac{1}{\lambda_p^2}\ket{-j},
\]
as desired.
\end{proof}

% \begin{definition}
% Let $k\in \Z_p$ and $f:\Z_p \to \Z_p$ be an invertible function. A \emph{generalised permutation matrix} is a matrix that maps $\ket{k}$ to $a_k\ket{f(k)}$, for some $a_k\in \C$ and $\lvert a_k\rvert = 1$.
% \label{def:generalised-perm}
% \end{definition}    

% \begin{corollary}
% $H^2$ is a generalised permutation matrix.
% \label{lem:H2-perm}
% \end{corollary} 

\begin{lemma}
$\lambda_p^4 = 1$.
\label{lem:order-lambda-p}
\end{lemma}  

\begin{proof}
  By \cref{def:lambda}, $\lambda_p  := e^{(p-1)\pi i/4}$. Then since $p$ is odd,
\[
    \lambda_p^4 = e^{4\cdot (p-1)\pi i/4}=e^{(p-1)\pi i} = (-1)^{p-1} = 1.
\]
\end{proof}  

% \begin{T15}
%     \Horder
% \end{T15}   
\begin{corollary}
  \label{cor:orderh}
  \Horder
\end{corollary}

\begin{proof}
Let $\ket{j}$ be an arbitrary computational basis state. It suffices to show that $H^4\ket{j} = \ket{j}$. By \cref{prop:hsquared,lem:order-lambda-p}, $H^4\ket{j} =H^2\left(\lambda_p^{-2}\ket{-j}\right)=\lambda_p^{-4}\ket{j}=\ket{j}$.
\end{proof}

\iffalse
\begin{proof}
By \cref{def:H-and-S}, $H=\frac{1}{\lambda_p\sqrt{p}}\sum_{j,\ell=0}^{p-1}\omega^{j\ell}\ket{j}\bra{\ell}$. It follows that

\begin{align*}
H^2 &= \frac{1}{\lambda_p^2\sqrt{p}^2}\left(\sum_{j,\ell=0}^{p-1}\omega^{j\ell}\ket{j}\bra{\ell}\right)\left(\sum_{j,\ell=0}^{p-1}\omega^{j\ell}\ket{j}\bra{\ell}\right)\\
&= \frac{1}{\lambda_p^2 \cdot p}\left(\sum_{j,\ell,m,n=0}^{p-1}\omega^{j\ell}\omega^{mn}\ket{j}\braket{\ell}{m}\bra{n}\right)\\
&= \frac{1}{\lambda_p^2 \cdot p}\left(\sum_{j,\ell,n=0}^{p-1}\omega^{\ell(j+n)}\ket{j}\bra{n}\right)\\
&= \frac{1}{\lambda_p^2 \cdot p}\left(\sum_{j,n=0}^{p-1}\ket{j}\bra{n}\left(\sum_{\ell=0}^{p-1}\omega^{\ell(j+n)}\right)\right)\\
&= \frac{1}{\lambda_p^2 \cdot p}\left(\sum_{j+n \equiv 0 \pmod{p}}^{p-1}\ket{j}\bra{n}\left(\sum_{\ell=0}^{p-1}\omega^{\ell(j+n)}\right) + \sum_{j+n \not\equiv 0 \pmod{p}}^{p-1}\ket{j}\bra{n}\left(\sum_{\ell=0}^{p-1}\omega^{\ell(j+n)}\right)\right)\\
&\xlongequal[]{\cref{lem:summation}}\frac{1}{\lambda_p^2 \cdot p}\left(\sum_{j+n \equiv 0 \pmod{p}}^{p-1}p\cdot\ket{j}\bra{n} + \sum_{j+n \not\equiv 0 \pmod{p}}^{p-1}0\cdot\ket{j}\bra{n}\right)\\
&=\frac{p}{\lambda_p^2 \cdot p}\sum_{j=0}^{p-1}\ket{j}\bra{p-j}=\frac{1}{\lambda_p^2}\sum_{j=0}^{p-1}\ket{j}\bra{p-j}.
\end{align*}

By \cref{def:generalised-perm}, $H^2$ is a generalised permutation matrix since for all $k\in\Z_p$,
\[
H^2\ket{k} = \frac{1}{\lambda_p^2}\sum_{j=0}^{p-1}\ket{j}\braket{p-j}{k}=\frac{1}{\lambda_p^2}\ket{-k},\quad \lambda_p^2 = 1 \;\text{or} \; -1.\qedhere
\]
\end{proof}    
\fi
\begin{lemma}
$\det(H)=1$.
\label{lem:det-H}
\end{lemma}    
Recall that $H=\frac{1}{\lambda_p\sqrt{p}}F$, where $F=\sum_{j,\ell=0}^{p-1}\omega^{j\ell}\ket{j}\bra{\ell}$ is a discrete Fourier transform. The determinant of $F$ is a standard computation (see e.g., \cite{MP72}), and the phase $\lambda_p$ is chosen so that $\det(H)=1$ (see also \cite{prakash2021normal}). For completeness, we give a short elementary proof.
\begin{proof}
    We first compute
\begin{align*}
\det(F)
&=\prod_{0\leq \ell<j\leq p-1}(\omega^j-\omega^\ell)
=\prod_{0\leq \ell<j\leq p-1}\omega^\ell(\omega^{j-\ell}-1)\\
&=\left(\prod_{t=1}^{p-1}\prod_{\ell=1}^{p-1-t}\omega^\ell\right)
\cdot\left(\prod_{t=1}^{p-1}(\omega^t-1)^{p-t}\right)\\
&=\omega^{(p-1)p(p-2)/6}
\left(\prod_{t=1}^{p-1}(\omega^t-1)^{p-t}\right).
\end{align*}
Let $\xi:=e^{\pi i/p}$. Since
\begin{equation*}
    \omega^t-1=2i\xi^t\sin(\pi t/p)
\end{equation*}
for all $t$, we obtain
\begin{align*}
\prod_{t=1}^{p-1}(\omega^t-1)^{p-t}
&=(2i)^{\sum_{t=1}^{p-1}(p-t)}
\cdot\xi^{\sum_{t=1}^{p-1}t(p-t)}
\cdot\prod_{t=1}^{p-1}\left(\sin(\pi t/p)\right)^{p-t}\\
&=(2i)^{p(p-1)/2}
\cdot\xi^{(p-1)p(p+1)/6}
\cdot\left(\prod_{t=1}^{(p-1)/2}\sin(\pi t/p)\right)^p,
\end{align*}
where the last equality pairs $(\sin(\pi t/p))^{p-t}$ with
$(\sin(\pi(p-t)/p))^t$. Differentiating the cyclotomic factorisation
\begin{equation*}
    x^p-1=\prod_{t=0}^{p-1}(x-\omega^t)
\end{equation*}
and evaluating at $x=1$ yields
\begin{equation*}
   \prod_{t=1}^{p-1}(1-\omega^t)=p. 
\end{equation*}
Note that $\vert 1-e^{i\theta}\vert=2\sin(\theta/2)$ for all
$\theta\in(0,2\pi)$. Taking absolute values therefore gives
\begin{equation*}
    \prod_{t=1}^{p-1}2\sin(\pi t/p)=p.
\end{equation*}
Since $\sin(\pi t/p)>0$ for all $1\leq t\leq p-1$, by pairing $t$ with $p-t$, we see that
\begin{equation*}
    \prod_{t=1}^{(p-1)/2}\sin(\pi t/p)
    =
    \left(\prod_{t=1}^{p-1}\sin(\pi t/p)\right)^{1/2}
    =
    2^{-(p-1)/2}p^{1/2}.
\end{equation*}
Hence
\begin{align*}
    \prod_{t=1}^{p-1}(\omega^t-1)^{p-t}
    =
    i^{p(p-1)/2}
    \cdot\xi^{(p-1)p(p+1)/6}
    \cdot p^{p/2}.
\end{align*}
Combining this with the prefactor
\begin{equation*}
    \omega^{(p-1)p(p-2)/6}
    =
    \xi^{(p-1)p(p-2)/3}
\end{equation*}
and using
\begin{equation*}
    \frac{(p-1)p(p-2)}{3}
    +
    \frac{(p-1)p(p+1)}{6}
    =
    \frac{p(p-1)^2}{2},
\end{equation*}
we obtain
\begin{equation*}
    \xi^{p(p-1)^2/2}
    =
    e^{\pi i(p-1)^2/2}
    =
    1,
\end{equation*}
since $(p-1)^2/2$ is an even integer for odd $p$. It follows that
\begin{align*}
    \det(F)=p^{p/2} e^{p(p-1)\pi i/4}=p^{p/2}\lambda_p^p.
\end{align*}
We conclude that $ \det(H)
    =
    \frac{\det(F)}{\lambda_p^p p^{p/2}}
    =
    1$.
\end{proof}

% \begin{T10}
% \determinant
% \end{T10}

\begin{lemma}
  \label{lem:Clifford-syllable-determinant}
  $\det(\CZ)=1$.
  $\det(S)=\omega$ when $p = 3$, and $\det(S)=1$ when $p \geq 5$.
\end{lemma} 

\begin{proof}
    Using straightforward properties of sums of integers, we have
    \begin{align}
        \sum_{j=0}^{p-1}j^2 \ &=\  \left(\sum_{j=1}^{p}j^2\right) - p^2 \ = \  \frac{p(p+1)(2p+1)}{6}-p^2.\label{eq:sum-psquare}\\
        \sum_{j=0}^{p-1}j \ &=\ \frac{\left(0+(p-1)\right)p}{2}\ =\ \frac{p(p-1)}{2}.\label{eq:sum-p}\\
        \sum_{j,\ell=0}^{p-1}j\ell \ &=\  \sum_{j=0}^{p-1}j\left(\sum_{\ell=0}^{p-1}\ell\right)\ \xlongequal[]{\eqref{eq:sum-p}}\ \sum_{j=0}^{p-1}j\left(\frac{p(p-1)}{2}\right) \ \xlongequal[]{\eqref{eq:sum-p}}\ \left(\frac{p(p-1)}{2}\right)^2\ = \ \frac{p^2(p-1)^2}{4}.\label{eq:sum-pl}
    \end{align}
    Then we have
    \begin{align}
        \sum_{j=0}^{p-1}j^2 - \sum_{j=0}^{p-1}j&\xlongequal[\eqref{eq:sum-p}]{\eqref{eq:sum-psquare}}\frac{p(p+1)(2p+1)}{6}-p^2 -\frac{p(p-1)}{2}\ = \ \frac{p(p^2-3p+2)}{3}.
        \label{eq:det-S-mid}\\
        \sum_{j=0}^{p-1}j^2 + \sum_{j=0}^{p-1}j&\xlongequal[\eqref{eq:sum-p}]{\eqref{eq:sum-psquare}}\frac{p(p+1)(2p+1)}{6}-p^2 +\frac{p(p-1)}{2}\ = \ \frac{p(p^2-1)}{3}.
        \label{eq:det-S-prime-mid}
    \end{align}   

    Since $p$ is an odd prime, $p = 2m + 1$ for some integer $m \geq 1$. Hence 
    \begin{equation}
    \frac{p^2(p-1)^2}{4} = \frac{(2m+1)^2(2m)^2}{4}=m^2(2m+1)^2=p^2m^2.
    \label{eq:sum-pl-mod}
    \end{equation}

    It follows that

    \begin{align*}
        \det(\CZ) &\xlongequal[]{} \prod_{j,\ell=0}^{p-1}\omega^{j\ell}\ =\ \omega^{\sum_{j,\ell=0}^{p-1}j\ell}\xlongequal[]{\eqref{eq:sum-pl}} \omega^{p^2(p-1)^2/4}\xlongequal[]{\eqref{eq:sum-pl-mod}}\omega^{p^2 m^2}= \left(\omega^p\right)^{p m^2} = 1.
    \end{align*}  
    
    To calculate $\det(S)$, we proceed by case distinctions on $p$. Firstly, note that

    \begin{align}
        \det(S) \ &\xlongequal[]{\eqref{eq:S}}\  \prod_{j=0}^{p-1}\omega^{\frac{j(j-1)}{2}} \ =\ \omega^{\sum_{j=0}^{p-1}\frac{j(j-1)}{2}} \ =\ \omega^{\frac{1}{2}\left(\sum_{j=0}^{p-1}j^2 - \sum_{j=0}^{p-1}j\right)} \notag\\
        \ &\xlongequal[]{\eqref{eq:det-S-mid}}\  \omega^{\frac{1}{2}\left(\frac{p(p^2-3p+2)}{3}\right)}\ =\ \omega^{p(p^2-3p+2)/6}=\omega^{p(p-1)(p-2)/6}.
        \label{eq:det-S-final}
    \end{align}

    When $p = 3$, $p(p-1)(p-2)/6 = 3\cdot 2 \cdot 1 / 6 = 1$, so $\det(S) = \omega$. 
    
    When $p\geq 5$, we need to show that $6\vert (p-1)(p-2)$. Since $p$ is odd, $p-1$ is even so $2\vert (p-1)$. Since $p > 3$ and $p$ is prime, $3 \nmid p$. Hence, $p = 3m + 1$ or $p = 3m + 2$, for some integer $m$. Then $p - 1 = 3m$ or $p - 2 = 3m$, so $3 \vert (p-1)(p-2)$. Therefore, $6 \vert (p-1)(p-2)$. Putting everything together, we have

    \begin{align*}
        \det(S) &\xlongequal[]{\eqref{eq:det-S-final}}  \omega^{p(p-1)(p-2)/6} = \left(\omega^p\right)^{(p-1)(p-2)/6} =1.
    \end{align*}
    
\end{proof}  

\begin{lemma}
For every integer $g \geq 1$, $S^g\ket{j} =\omega^{\frac{j(j-1)}{2}\cdot g}\ket{j}$, where multiplication is taken modulo $p$.
\label{lem:multiple-s}
\end{lemma}  

\begin{proof}
Let $\ket{j}$ be a computational basis state, and we proceed by induction on $g$. When $g=1$, the base case holds by \eqref{eq:S}. Suppose the statement holds when $g\geq 1$. Then
\[
S^{g+1}\ket{j} = S\left(S^g\ket{j}\right)\xlongequal[]{\IH}\omega^{\frac{j(j-1)}{2}\cdot g}\left(S\ket{j}\right)\xlongequal[]{\eqref{eq:S}}\omega^{\frac{j(j-1)}{2}\cdot g}\omega^{\frac{j(j-1)}{2}}\ket{j}=\omega^{\frac{j(j-1)}{2}\cdot (g+1)}\ket{j}.
\]
This completes the proof.
\end{proof}  

\begin{lemma}
For every integer $g \geq 1$, $\CZ^g\ket{j,\ell} =\omega^{gj\ell}\ket{j,\ell}$, where multiplication is taken modulo $p$.
\label{lem:multiple-cz}
\end{lemma}  

\begin{proof}
Let $\ket{j,\ell}$ be a computational basis state, and we proceed by induction on $g$. When $g=1$, the base case holds by \eqref{eq:CZ}. Suppose the statement holds when $g\geq 1$. Then
\[
\CZ^{g+1}\ket{j,\ell} = \CZ\left(\CZ^g\ket{j,\ell}\right)\xlongequal[]{\IH}\omega^{gj\ell}\left(\CZ\ket{j,\ell}\right)\xlongequal[]{\eqref{eq:CZ}}\omega^{gj\ell}\omega^{j\ell}\ket{j,\ell}=\omega^{(g+1)j\ell}\ket{j,\ell}.
\]
This completes the proof.
\end{proof}

% \begin{T14}
%     \Sorder
% \end{T14}   

\begin{corollary}
    \Sorder
    \label{cor:S-order}
\end{corollary}

\begin{proof}
Let $\ket{j}$ be an arbitrary computational basis state. It suffices to show that $S^p\ket{j} = \ket{j}$.
By \cref{lem:multiple-s}, $S^p\ket{j} = \left(\omega^p\right)^{\frac{j(j-1)}{2}}\ket{j}=\ket{j}$.
\end{proof}    

% \begin{T16}
%     \CZorder
% \end{T16}   

\begin{corollary}
    \CZorder
    \label{cor:CZ-order}
\end{corollary}

\begin{proof}
Let $\ket{j,\ell}$ be an arbitrary computational basis state. It suffices to show that $\CZ^p\ket{j,\ell} = \ket{j,\ell}$.
By \cref{lem:multiple-cz}, $\CZ^p\ket{j,\ell} = \left(\omega^p\right)^{j\ell}\ket{j,\ell}=\ket{j,\ell}$.
\end{proof}      
% W show that $\det(H)=1$ in the next appendix.

% \begin{proof}
%     When $1 \leq i < j \leq p-1$, $-(p-1) \leq -j < -i \leq -1$. It follows that $1 \leq p-j < p-i \leq p-1$. Since $-j \equiv p-j \pmod{p}$ and $-i \equiv p-i \pmod{p}$, we have

%     \begin{align}
%         \prod_{1 \leq i < j \leq p-1}\left(\omega^{-j} - \omega^{-i}\right)&=\prod_{1 \leq p-j < p-i \leq p-1}\left(\omega^{p-j} - \omega^{p-i}\right) =\prod_{1 \leq j' < i' \leq p-1}\left(\omega^{j'} - \omega^{i'}\right)\nonumber\\
%         &=\prod_{1 \leq j' < i' \leq p-1}-\left(\omega^{i'} - \omega^{j'}\right)=\prod_{1 \leq i < j \leq p-1}-\left(\omega^{j} - \omega^{i}\right)=(-1)^{\frac{(p-1)(p-2)}{2}}\prod_{1 \leq i < j \leq p-1}\left(\omega^{j} - \omega^{i}\right).\label{eq:sum}
%     \end{align}    
%     It follows that

%     \begin{align*}
%         \prod_{0 \leq i < j \leq p-1}\left(\omega^{-j} - \omega^{-i}\right) &=\left(\prod_{1 \leq j \leq p-1}\left(\omega^{-j} - \omega^{0}\right)\right)\left(\prod_{1 \leq i < j \leq p-1}\left(\omega^{-j} - \omega^{-i}\right)\right)\nonumber\\
%         &\xlongequal[]{\eqref{eq:sum}} (-1)^{\frac{(p-1)(p-2)}{2}}\left(\prod_{1 \leq k \leq p-1}\left(\omega^{k} - \omega^{0}\right)\right)\left(\prod_{1 \leq i < j \leq p-1}\left(\omega^{j} - \omega^{i}\right)\right)\\
%         &=(-1)^{\frac{(p-1)(p-2)}{2}} \prod_{0 \leq i < j \leq p-1}\left(\omega^{j} - \omega^{i}\right).\qedhere
%     \end{align*}  
% \end{proof}    

% Here is the proof of \cref{lem:automorphism}.
\begin{T11}
\HSaction
\end{T11}    

\begin{proof}
Let $\ket{j}$ be an arbitrary computational basis state. It suffices to show that $HXH^\dagger\ket{j} = Z\ket{j}$, $HZH^\dagger\ket{j} = X^\dagger\ket{j}$, $SXS^\dagger\ket{j} = XZ\ket{j}$, and $SZS^\dagger\ket{j} = Z\ket{j}$. One can check that $\lambda_p\lambda_p^\dagger = 1$. By \cref{def:H-and-S,def:pauli}, we have
\begin{align}
HXH^\dagger \ket{j} &\xlongequal[]{\eqref{eq:H}}   HX\left( \frac{1}{\lambda_p^\dagger\sqrt{p}} \sum_{k=0}^{p-1} \omega^{-kj} \ket{k} \right) \notag\\
                &=\  H\left( \frac{1}{\lambda_p^\dagger\sqrt{p}} \sum_{k=0}^{p-1} \omega^{-kj} \ket{k+1} \right) \notag\\      
                 &\xlongequal[]{\eqref{eq:H}}  \frac{1}{\lambda_p^\dagger\sqrt{p}} \sum_{k=0}^{p-1} \omega^{-kj} \left( 
                 \frac{1}{\lambda_p\sqrt{p}} \sum_{\ell=0}^{p-1} \omega^{(k+1)\ell} \ket{\ell}
                 \right)\notag\\ 
                 &=\  \frac{1}{p} \sum_{\ell=0}^{p-1} \omega^\ell \left( \sum_{k=0}^{p-1} \omega^{k(\ell -j)} \right)\ket{\ell}.
                 \label{eq:HXHdagger}
\end{align}
Now consider the sum $\sum_{k=0}^{p-1} \omega^{k(\ell -j)}$. When $\ell \equiv j \pmod{p}$, then $\omega^{k(\ell -j)}=1$, so that the sum is equal to $p$. When $\ell \not\equiv j \pmod{p}$, then $\omega^{\ell-j}$ is a primitive $p$-th root of unity, so that the sum is equal to 0. Hence,
\begin{equation}
\eqref{eq:HXHdagger}\ =\   \omega^j \ket{j} + \sum_{\ell\not\equiv j \pmod{p}} \omega^\ell \left(\frac{1}{p} \left( \sum_{k=0}^{p-1} \omega^{k(\ell -j)} \right)\right)\ket{\ell} \ =\  \omega^j \ket{j}  \ =\  Z\ket{j}.
\label{eq:HXHdagger-final}
\end{equation}

It follows that $HZH^\dagger \xlongequal[]{\eqref{eq:HXHdagger-final}} H^2 X (H^\dagger)^2 \xlongequal[]{\cref{cor:orderh}} H^2 X H^2$, and
% \[
% HZH^\dagger\ket{j}\ = \  H^2X H^2 \ket{j} \ = \  H^2 X \ket{-j} \ = \  H^2\ket{-j+1} \ = \  \ket{-(-j+1)} \ = \  \ket{j-1} \ = \  X^\dagger \ket{j}.
% \]

\begin{align*}
  HZH^\dagger\ket{j} &= H^2X H^2 \ket{j} \xlongequal[]{\cref{prop:hsquared}}  H^2 X \ket{-j}\cdot \lambda_p^{-2}\notag\\
  &= H^2\ket{-j+1}\cdot \lambda_p^{-2} \xlongequal[]{\cref{prop:hsquared}} \ket{-(-j+1)}\cdot \lambda_p^{-4} \notag\\
  &= \ket{j-1}\cdot \lambda_p^{-4} \xlongequal[]{\cref{lem:order-lambda-p}}\ket{j-1}=  X^\dagger \ket{j}.
\end{align*}  

Next, we have
\begin{align*}
    SX{S}^\dagger \ket{j} &\xlongequal[]{\eqref{eq:S}}  SX \ket{j}\cdot \omega^{-\frac{j(j-1)}{2}} =  S\ket{j+1}\cdot \omega^{-\frac{j(j-1)}{2}} \\ 
    &\xlongequal[]{\eqref{eq:S}}  \omega^{\frac{j(j+1)}{2}- \frac{j(j-1)}{2}} \ket{j+1} \ = \  \omega^{j}\ket{j+1} \ = \  XZ\ket{j}.
\end{align*}

Finally, we have
\begin{align*}
    SZS^\dagger \ket{j} &\xlongequal[]{\eqref{eq:S}}  SZ\ket{j}\cdot \omega^{-\frac{j(j-1)}{2}} =  S\ket{j}\cdot \omega^{-\frac{j(j-1)}{2} + j} \\ 
    &\xlongequal[]{\eqref{eq:S}} \omega^{-\frac{j(j-1)}{2} + j + \frac{j(j-1)}{2}}\ket{j} \ = \  \omega^j\ket{j} \ = \  Z\ket{j}.
\end{align*}

% and
% \begin{align*}
%     S'Z{S'}^\dagger \ket{j} \ &= \  SZ\omega{-\frac{j(j+1)}{2}}\ket{j} \ = \  S\omega{-\frac{j(j+1)}{2} + 1}\ket{j} \\
%      &=\  \omega{-\frac{j(j+1)}{2} + 1 + \frac{j(j+1)}{2}}\ket{j} \ = \  \omega\ket{j} \ = \  Z\ket{j}. \qedhere
% \end{align*}
\end{proof}

The proof of \cref{lem:automorphism-CZ} is similar to the proof of \cref{lem:automorphism}, and is omitted here.

\subsection{Soundness of the Derived Generators}
\label{subsec:soundness-derived}

In this subsection, we verify that each derived generator, built from the Clifford generators, acts on the computational basis exactly as intended. Recall the definition of the derived generators $X$, $Z$, $\CX$, $\XC$, $\SWAP$, and $\CIZ$ given in \cref{fig:derived-generators2}.

\begin{align}
  Z &= H^2 S H^2 S^{-1}\label{eq:Z}\\
  X &= HSH^2S^{-1}H\label{eq:X}\\
  M_a &= Z^{(1-a)/(2a)}\, X^{(1-a)/2}\, S^{a^{-1}} H S^{a} H S^{a^{-1}} H\cdot \left(\frac{a}{p}\right)_L\omega^{\frac{-a^2+4a-2}{8a}}\label{eq:Mg}\\
  \CX &= (I\otimes H^3)\CZ(I\otimes H)\label{eq:CX}\\
  \SWAP &= \CZ(H\otimes H)\CZ(H\otimes H)\CZ(H\otimes H)\cdot \lambda_p^2\label{eq:SWAP}\\
  \XC &= (\SWAP)\CX(\SWAP)\label{eq:XC}\\
  \CIZ &= (I\otimes \SWAP)(\CZ\otimes I)(I\otimes \SWAP)\label{eq:CIZ}\\
  \CIX &= (I\otimes \SWAP)(\CX\otimes I)(I\otimes \SWAP)\label{eq:CIX}\\
  \XIC &= (\SWAP\otimes I)(I\otimes\XC)(\SWAP\otimes I)\label{eq:XIC}
\end{align}

Next, we show that the derived generators given in \eqref{eq:Z}--\eqref{eq:XIC} have the desired actions on the computational basis states, as shown in \eqref{eq:Z-action}--\eqref{eq:XIC-action}.

\begin{align}
Z\ket{j} &= \omega^j\ket{j}\label{eq:Z-action}\\
X\ket{j}&= \ket{j+1}\label{eq:X-action}\\
M_a\ket{j}&= \ket{aj}\label{eq:Mg-action}\\
\CX\ket{j,\ell} &= \ket{j,j+\ell}\label{eq:CX-action}\\
\SWAP\ket{j,\ell} &= \ket{\ell,j}\label{eq:SWAP-action}\\
\XC\ket{j,\ell} &= \ket{j+\ell,\ell}\label{eq:XC-action}\\
\CIZ\ket{j,\ell,k} &= \omega^{jk}\ket{j,\ell,k}\label{eq:CIZ-action}\\
\CIX\ket{j,\ell,k} &= \ket{j,\ell,j+k}\label{eq:CIX-action}\\
\XIC\ket{j,\ell,k} &= \ket{j+k,\ell,k}\label{eq:XIC-action}
\end{align}  

% we prove two lemmas. The first one provides a simplified path-sum expression for $H^2$. The second one provides a path-sum interpretation of the derived generators above.

\begin{lemma}
  For every $j\in\Z_p$, $Z$ defined below satisfies $Z\ket{j}=\omega^j\ket{j}$.
  \[
    \scalebox{1}{\input{figures/Preliminaries/Z-soundness.tikz}}
  \]
\label{lem:Z-soundness}
\end{lemma}

% \begin{proposition}
% \label{prop:d2}
% For $0\leq j \leq p-1$, $H^2 \ket{j} = \lambda_p^{-2}\ket{-j}$. Here the negation $-j$ is performed modulo $p$.
% \end{proposition}

\begin{proof}
Let $\ket{j}$ be a computational basis state. 
\begin{align}
\RHS\ket{j}=H^2SH^2S^{-1}\ket{j} 
&\xlongequal[\cref{lem:multiple-s}]{\eqref{eq:S}}  H^2SH^2\left(\omega^{-\frac{j(j-1)}{2}}\right)\ket{j} \notag\\
&\xlongequal[]{\cref{prop:hsquared}} \left(\omega^{-\frac{j(j-1)}{2}}\lambda_p^{-2}\right)H^2S\ket{-j} \notag\\
&\xlongequal[]{\eqref{eq:S}} \left(\omega^{-\frac{j(j-1)}{2}+\frac{(-j)(-j-1)}{2}}\lambda_p^{-2}\right)H^2\ket{-j} \notag\\
&\xlongequal[]{\cref{prop:hsquared}} \left(\omega^{-\frac{j(j-1)}{2}+\frac{(-j)(-j-1)}{2}}\lambda_p^{-4}\right)\ket{j} \notag\\
&\xlongequal[]{\cref{lem:order-lambda-p}} \left(\omega^{-j\frac{(j-1) + (-j-1)}{2}}\right)\ket{j} = \omega^j\ket{j}. \label{eq:Z-int}
\end{align}
\end{proof}  

\begin{lemma}
  For every $j\in\Z_p$, $X$ defined below satisfies $X\ket{j}=\ket{j+1}$.
  \[
    \scalebox{1}{\input{figures/Preliminaries/X-soundness.tikz}}
  \]
\label{lem:X-soundness}
\end{lemma}

\begin{proof}
Let $\ket{j}$ be a computational basis state.
\begin{align}
  \RHS\ket{j} &= HSH^2S^{-1}H\ket{j}\notag\\
  &\xlongequal[]{\eqref{eq:H}}HSH^2S^{-1}\left(\frac{1}{\lambda_p\sqrt{p}}\sum_{\ell=0}^{p-1}\omega^{j\ell}\ket{\ell}\right)\notag\\
  &\xlongequal[]{\cref{lem:multiple-s}}HSH^2\left(\frac{1}{\lambda_p\sqrt{p}}\sum_{\ell=0}^{p-1}\omega^{j\ell}\omega^{\frac{-\ell(\ell-1)}{2}}\ket{\ell}\right)\notag\\
  &\xlongequal[]{\cref{prop:hsquared}}HS\left(\frac{1}{\lambda_p\sqrt{p}}\sum_{\ell=0}^{p-1}\omega^{j\ell+\frac{-\ell(\ell-1)}{2}}\lambda_p^{-2}\ket{-\ell}\right)\notag\\
  &\xlongequal[]{\eqref{eq:S}}H\left(\frac{1}{\lambda_p^3\sqrt{p}}\sum_{\ell=0}^{p-1}\omega^{j\ell+\frac{-\ell(\ell-1)}{2}}\omega^{\frac{-\ell(-\ell-1)}{2}}\ket{-\ell}\right)\notag\\
  &\xlongequal[]{\eqref{eq:H}}\left(\frac{1}{\lambda_p^3\sqrt{p}}\sum_{\ell=0}^{p-1}\omega^{j\ell+\frac{-\ell(\ell-1)}{2}+\frac{-\ell(-\ell-1)}{2}}\left(\frac{1}{\lambda_p\sqrt{p}}\sum_{k=0}^{p-1}\omega^{-k\ell}\ket{k}\right)\right)\notag\\
  &=\left(\frac{1}{\lambda_p^4\cdot p}\sum_{\ell,k=0}^{p-1}\omega^{j\ell+\frac{-\ell(\ell-1-\ell-1)}{2}-k\ell}\ket{k}\right)\notag\\
  &\xlongequal[]{\cref{lem:order-lambda-p}}\frac{1}{p}\sum_{\ell,k=0}^{p-1}\omega^{j\ell+\ell-k\ell}\ket{k}=\sum_{k=0}^{p-1}\left(\frac{1}{p}\sum_{\ell=0}^{p-1}\omega^{\ell(j+1-k)}\right)\ket{k}\label{eq:X-int}
\end{align}

Hence, $\RHS\ket{j} = \sum_{k=0}^{p-1}C_k\ket{k}$, where

\begin{equation}
C_k = \frac{1}{p}\sum_{\ell=0}^{p-1}\omega^{\ell(j+1-k)} = \frac{1}{p}\sum_{\ell=0}^{p-1}\omega^a,\quad a = j+1-k.
\label{eq:X-int2}
\end{equation}  

By \cref{lem:summation}, we proceed by case distinctions. When $a \equiv 0 \pmod{p}$, $C_k = 1$. When $a \not\equiv 0 \pmod{p}$, $C_k = 0$. That is, each term $\ket{k}$ has coefficient $0$, except the one with $k \equiv j+1 \pmod{p}$. Therefore, $\RHS\ket{j} \xlongequal[]{\eqref{eq:X-int2}} \ket{j+1}$.
\end{proof}  

% \begin{lemma}
%   $M_a\ket{j} = \ket{aj}$, where multiplication is taken modulo $p$ and
%   \[
%     \scalebox{.9}{\tikzfig{figures/Preliminaries/Mg-soundness}}
%   \]
% \label{lem:Mg-soundness}
% \end{lemma}  

\begin{T12}
\multiplier
\end{T12}    

\begin{proof}
For every $a\in\Z_p^*$, by tracking the Pauli computation (and ignoring the global phase) we can check that 

\begin{align}
M_aZM_a^\dagger &= Z^{a^{-1}} \iff M_aZ = Z^{a^{-1}}M_a\label{eq:Ma-Z}\\
M_aXM_a^\dagger &= X^a \iff M_aX = X^aM_a.\label{eq:Ma-X}
\end{align}

% $M_aZM_a^\dagger=Z^{a^{-1}}$ and $M_aXM_a^\dagger=X^a$. 

\[
     \scalebox{.7}{\input{figures/Preliminaries/multiplier-soundness.tikz}}
\]

We must then have $M_a\ket{0} = M_a Z^a\ket{0} \xlongequal[]{\eqref{eq:Ma-Z}} Z^{a^{-1}a}M_a \ket{0} = ZM_a\ket{0}$, so that $M_a\ket{0}$ is a $+1$ eigenstate of $Z$ and hence must be proportional to $\ket{0}$. That is, 

\begin{equation}
M_a\ket{0} = c_a\ket{0}, \quad \text{where } c_a\in \U(1).
\label{eq:Ma-0}
\end{equation}  

For all $y \in \Z_p$, $M_a\ket{y} = M_a X^y\ket{0}  \xlongequal[]{\eqref{eq:Ma-X}} X^{ay} M_a\ket{0} \xlongequal[]{\eqref{eq:Ma-0}} c_a X^{ay}\ket{0} = c_a\ket{ay}$. By using the path-sum formalism and equations \eqref{eq:H}, \eqref{eq:S}, \eqref{eq:Z-action}, and \eqref{eq:X-action}, we have
\begin{equation*}
M_a\ket{0}
=
\frac{1}{\lambda_p^3p^{3/2}}
\omega^{\frac{-a^2+4a-2}{8a}}
\left(\frac{a}{p}\right)_L
\sum_h
\omega^{t_h}\mu_h
\ket{h+\frac{1-a}{2}},
\end{equation*}
where
\begin{equation*}
t_h
=
\frac{h(h-1)}{2a}
+
\frac{1-a}{2a}
\left(
h+\frac{1-a}{2}
\right)
\end{equation*}
and
\begin{equation*}
\mu_h
=
\sum_{k,j}
\omega^{
\frac{j(j-1)}{2a}
+kj
+\frac{ak(k-1)}{2}
+kh
}.
\end{equation*}
As we also have $M_a\ket{0}=c_a\ket{0}$ we must then have $\mu_h=0$ for all $h\neq\frac{a-1}{2}$ so that
\begin{equation*}
\sum_h\mu_h=\mu_{\frac{a-1}{2}}.
\end{equation*}
It follows that
\begin{equation*}
\begin{aligned}
\mu_{\frac{a-1}{2}}
&=
\sum_h\mu_h\\
&=
\sum_{k,j}
\omega^{
\frac{j(j-1)}{2a}
+kj
+\frac{ak(k-1)}{2}
}
\sum_h\omega^{kh}\\
&=
p\left(
\sum_j\omega^{\frac{j(j-1)}{2a}}
\right)\\
&=
\lambda_p^{-1}p^{3/2}
\left(\frac{a}{p}\right)_L
\omega^{-\frac{1}{8a}}.
\end{aligned}
\end{equation*}
The second last equality uses that
\begin{equation*}
\sum_h\omega^{kh}
=
\begin{cases}
p & \text{if } k=0,\\
0 & \text{otherwise}.
\end{cases}
\end{equation*}
The last equality follows from the property of the quadratic Gauss sum
\begin{equation*}
\sum_j\omega^{\frac{j^2}{2a}}
=
\lambda_p^{-1}\sqrt{p}
\left(\frac{a}{p}\right)_L,
\end{equation*}
where $\left(\frac{\cdot}{\cdot}\right)_L$ is the Legendre symbol. Finally,
\begin{equation*}
t_{\frac{a-1}{2}}
=
\frac{(a-1)(a-3)}{8a}.
\end{equation*}
It follows that
\begin{equation*}
\begin{aligned}
c_a
&=
\left(\frac{a}{p}\right)_L
\omega^{\frac{-a^2+4a-2}{8a}}
\frac{1}{\lambda_p^3p^{3/2}}
\omega^{t_{\frac{a-1}{2}}}
\mu_{\frac{a-1}{2}}\\
&=
\lambda_p^{-4}
\left(\frac{a}{p}\right)_L^2
\omega^{
\frac{-a^2+4a-2}{8a}
+
\frac{(a-1)(a-3)}{8a}
-
\frac{1}{8a}
}\\
&=1,
\end{aligned}
\end{equation*}
where we used $\lambda_p^4=1$, $\left(\frac{a}{p}\right)_L^2=1$, and
\begin{equation*}
(-a^2+4a-2)+(a-1)(a-3)-1=0.
\end{equation*}
We conclude that $M_a\ket{x}=\ket{ax}$ for all $x\in\mathbb{Z}_p$.
\end{proof}

\begin{corollary}
  For every $a\in\Z_p^*$, $M_a$ is a unitary operator.
  \label{cor:Mg-unitary}
\end{corollary}

\begin{proof}
By \cref{lem:multiplier}, $M_a\ket{j} = \ket{aj}$ for all $j\in\Z_p$. Write $M_a = \sum_{j=0}^{p-1} \ket{aj}\bra{j}$. Since $p$ is prime, every non-zero element of $\Z_p$ is invertible. Let $k = aj$, then $j = a^{-1}k$ for some $k \in \Z_p$. Then $M_a^\dagger = \sum_{j=0}^{p-1} \ket{j}\bra{aj}=\sum_{k=0}^{p-1} \ket{a^{-1}k}\bra{k}$. It follows that $M_a M_a^\dagger = M_a^\dagger M_a=1$, since

\begin{align*}
M_a M_a^\dagger&=\left(\sum_{j=0}^{p-1} \ket{aj}\bra{j}\right)\left(\sum_{k=0}^{p-1} \ket{a^{-1}k}\bra{k}\right)=\sum_{j,k=0}^{p-1} \ket{aj}\braket{j}{a^{-1}k}\bra{k}\\
&\xlongequal[]{j=a^{-1}k}\sum_{k=0}^{p-1} \ket{aa^{-1}k}\bra{k}=\sum_{k=0}^{p-1} \ket{k}\bra{k}=1.
\end{align*}

\begin{align*}
M_a^\dagger M_a&=\left(\sum_{k=0}^{p-1} \ket{a^{-1}k}\bra{k}\right)\left(\sum_{j=0}^{p-1} \ket{aj}\bra{j}\right)=\sum_{k,j=0}^{p-1} \ket{a^{-1}k}\braket{k}{aj}\bra{j}\\
&\xlongequal[]{k=aj}\sum_{j=0}^{p-1} \ket{a^{-1}aj}\bra{j}=\sum_{j=0}^{p-1} \ket{j}\bra{j}=1.
\end{align*}

Therefore, $M_a$ is unitary.
\end{proof}

\begin{lemma}
  For every $j,\ell\in\Z_p$, $\CX$ defined below satisfies $\CX\ket{j,\ell}=\ket{j,j+\ell}$, and addition is taken modulo $p$.
  \[
    \scalebox{1}{\input{figures/Preliminaries/CX-soundness.tikz}}
  \]
\label{lem:CNOT-soundness}
\end{lemma}  

\begin{proof}
  Let $\ket{j}$ and $\ket{\ell}$ be arbitrary computational basis states.
  \begin{align}
  \RHS\ket{j,\ell} &= (I \otimes H^3)\CZ (I\otimes H)\ket{j,\ell}\notag\\
  &\xlongequal[]{\eqref{eq:H}} (I \otimes H^3)\CZ\left(\frac{1}{\lambda_p\sqrt{p}}\sum_{k=0}^{p-1}\omega^{k\ell}\ket{j,k}\right)\notag\\
   &\xlongequal[]{\eqref{eq:CZ}} (I \otimes H^3)\left(\frac{1}{\lambda_p\sqrt{p}}\sum_{k=0}^{p-1}\omega^{k\ell}\omega^{jk}\ket{j,k}\right)\notag\\
   &= \frac{1}{\lambda_p\sqrt{p}}(I\otimes H)\left(\sum_{k=0}^{p-1}\omega^{k(j+\ell)}\ket{j} \left(H^2\ket{k}\right)\right)\notag\\
   &\xlongequal[]{\cref{prop:hsquared}}\frac{1}{\lambda_p\sqrt{p}}(I\otimes H)\left(\sum_{k=0}^{p-1}\omega^{k(j+\ell)}\ket{j}\left(\lambda_p^{-2}\ket{-k}\right)\right)\notag\\
   &=\frac{1}{\lambda_p^3\sqrt{p}}(I\otimes H)\left(\sum_{k=0}^{p-1}\omega^{k(j+\ell)}\ket{j,-k}\right)\notag\\
   &\xlongequal[]{\eqref{eq:H}} \frac{1}{\lambda_p^3\sqrt{p}}\left(\sum_{k=0}^{p-1}\omega^{k(j+\ell)}\ket{j}\left(\frac{1}{\lambda_p\sqrt{p}}\sum_{t=0}^{p-1}\omega^{-kt}\ket{t}\right)\right)\notag\\
   &=\frac{1}{\lambda_p^4 \cdot p}\left(\sum_{k,t=0}^{p-1}\omega^{k(j+\ell-t)}\ket{j,t}\right)\label{eq:CX-int}.
  \end{align} 

  By \cref{lem:order-lambda-p}, we can further simplify $\RHS\ket{j,\ell}$ as follows.

  \begin{equation}
    \eqref{eq:CX-int} = \frac{1}{p}\left(\sum_{k,t=0}^{p-1}\omega^{k(j+\ell-t)}\ket{j,t}\right)=\sum_{t=0}^{p-1}\left(\frac{1}{p}\sum_{k=0}^{p-1}\omega^{k(j + \ell - t)}\right)\ket{j,t}=\sum_{t=0}^{p-1}C_t\ket{j,t},\;\text{where}
    \label{eq:CX-int2}
  \end{equation}

  \begin{equation}
    C_t = \frac{1}{p}\sum_{k=0}^{p-1}\omega^{k(j + \ell - t)}=\frac{1}{p}\sum_{k=0}^{p-1}\omega^{ka},\quad a = j + \ell - t.
    \label{eq:CX-int3}
  \end{equation}

  By \cref{lem:summation}, we proceed by case distinctions. When $a \equiv 0 \pmod{p}$, $C_t = 1$. When $a \not\equiv 0 \pmod{p}$, $C_t = 0$. That is, each term $\ket{j,t}$ has coefficient $0$, except the one with $t \equiv j+\ell \pmod{p}$. Therefore, $\RHS\ket{j,\ell} \xlongequal[]{\eqref{eq:CX-int2}} \ket{j,j+\ell}$.
\end{proof}

\begin{lemma}
For all $g \in \Z_p^*$, $\CX^g\ket{j,\ell} =\ket{j,\ell+gj}$, where multiplication and addition are taken modulo $p$.
\label{lem:multiple-cx}
\end{lemma}  

\begin{proof}
Let $\ket{j,\ell}$ be a computational basis state, and we proceed by induction on $g$. When $g=1$, the base case holds by \eqref{eq:CX-action}. Suppose the statement holds when $g\geq 1$. Then
\[
\CX^{g+1}\ket{j,\ell} = \CX\left(\CX^g\ket{j,\ell}\right)\xlongequal[]{\IH}\CX\ket{j,\ell+gj}\xlongequal[]{\eqref{eq:CX-action}}\ket{j,j+\ell+gj}=\ket{j,\ell+(g+1)j}.
\]
This completes the proof.
\end{proof} 

\begin{lemma}
  For every $j,\ell\in\Z_p$, $\SWAP$ defined below satisfies $\SWAP\ket{j,\ell}=\ket{\ell,j}$.
    \[
    \scalebox{1}{\input{figures/Preliminaries/SWAP-soundness.tikz}}
  \]
\label{lem:SWAP-soundness}
\end{lemma}  

\begin{proof}
Let $\ket{j}$ and $\ket{\ell}$ be arbitrary computational basis states.
\begin{align}
  \RHS\ket{j,\ell} &= \lambda_p^2 \cdot \CZ(H\otimes H)\CZ(H\otimes H)\CZ(H\otimes H)\ket{j,\ell}\notag\\
  &\xlongequal[]{\eqref{eq:H}} \lambda_p^2 \cdot \CZ(H\otimes H)\CZ(H\otimes H)\CZ\left(\frac{1}{\lambda_p^2 p}\cdot\sum_{k,t=0}^{p-1}\omega^{kj+t\ell}\ket{k,t}\right)\notag\\
  &\xlongequal[]{\eqref{eq:CZ}} \lambda_p^2 \cdot \frac{1}{\lambda_p^2 p}\cdot \CZ(H\otimes H)\CZ(H\otimes H)\left(\sum_{k,t=0}^{p-1}\omega^{kj+t\ell+kt}\ket{k,t}\right)\notag\\
  &\xlongequal[\eqref{eq:CZ}]{\eqref{eq:H}} \frac{1}{p}\cdot\CZ(H\otimes H)\left(\frac{1}{\lambda_p^2 p}\cdot\sum_{k,t=0}^{p-1}\omega^{kj+t\ell+kt}\sum_{a,b=0}^{p-1}\omega^{ak+bt+ab}\ket{a,b}\right)\notag\\
  &= \frac{1}{\lambda_p^2 p^2}\cdot\CZ(H\otimes H)\left(\sum_{a,b,k,t=0}^{p-1}\omega^{kj+t\ell+kt + ak+bt+ab}\ket{a,b}\right)\notag\\
  &\xlongequal[\eqref{eq:CZ}]{\eqref{eq:H}} \frac{1}{\lambda_p^4 p^3}\cdot\left(\sum_{a,b,k,t=0}^{p-1}\omega^{kj+t\ell+kt + ak+bt+ab}\sum_{c,r = 0}^{p-1}\omega^{ca + br + cr}\ket{c,r}\right)\notag\\
  &\xlongequal[]{\cref{lem:order-lambda-p}}\frac{1}{p^3}\cdot\left(\sum_{a,b,c,r,k,t=0}^{p-1}\omega^{kj+t\ell+kt+ak+bt+ab+ca+br+cr}\ket{c,r}\right)\notag\\
  &=\frac{1}{p^3}\cdot\left(\sum_{b,c,r,k,t=0}^{p-1}\omega^{kj+t\ell+kt+bt+br+cr}\left(\sum_{a=0}^{p-1}\omega^{a(b+c+k)}\right)\ket{c,r}\right)\notag\\
  &\xlongequal[]{\eqref{eq:summation}}\frac{p}{p^3}\cdot\left(\sum_{b,c,r,k,t=0}^{p-1}\omega^{kj+t\ell+kt+bt+br+cr}\delta_p(b+c+k)\ket{c,r}\right)\notag\\
  &=\frac{1}{p^2}\cdot\left(\sum_{c,r,k,t=0}^{p-1}\omega^{kj+t\ell+kt+(-k-c)t+(-k-c)r+cr}\ket{c,r}\right)\label{eq:SWAP-int}
\end{align}

Note that $kj+t\ell+kt+(-k-c)t+(-k-c)r+cr=kj+t\ell+kt-kt-tc-kr-cr+cr=kj+t\ell-tc-kr = k(j-r)+t(\ell-c)$. Hence, \eqref{eq:SWAP-int} can be further simplified as follows.
\begin{align}
    \RHS\ket{j,\ell} &= \frac{1}{p^2}\cdot\left(\sum_{c,r,k,t=0}^{p-1}\omega^{k(j-r)}\cdot\omega^{t(\ell-c)}\ket{c,r}\right)\notag\\
    &= \frac{1}{p^2}\cdot\left(\sum_{c,r=0}^{p-1}\left(\sum_{k=0}^{p-1}\omega^{k(j-r)}\right)\cdot\left(\sum_{t=0}^{p-1}\omega^{t(\ell-c)}\right)\ket{c,r}\right)\notag\\
    &\xlongequal[]{\eqref{eq:summation}} \frac{1}{p^2}\cdot\left(\sum_{c,r=0}^{p-1}p^2\cdot\delta_p(j-r)\cdot \delta_p(\ell-c)\ket{c,r}\right)\notag\\
    &= \sum_{c,r=0}^{p-1}\delta_p(j-r)\cdot \delta_p(\ell-c)\ket{c,r} = \ket{\ell,j}.
\end{align}
\end{proof}  

\begin{lemma}
  For every $j,\ell\in\Z_p$, $\XC$ defined below satisfies $\XC\ket{j,\ell}=\ket{j+\ell,\ell}$, and addition is taken modulo $p$.
    \[
    \scalebox{1}{\input{figures/Preliminaries/XC-soundness.tikz}}
  \]
\label{lem:XC-soundness}
\end{lemma}  

\begin{proof}
  Let $\ket{j}$ and $\ket{\ell}$ be arbitrary computational basis states.
\[
  \RHS\ket{j,\ell} = (\SWAP)\CX(\SWAP)\ket{j,\ell} \xlongequal[]{\eqref{eq:SWAP-action}}(\SWAP)\CX\ket{\ell,j} \xlongequal[]{\eqref{eq:CX-action}} \SWAP\ket{\ell,j+\ell} \xlongequal[]{\eqref{eq:SWAP-action}}\ket{j+\ell,\ell}.
\]
\end{proof}

\begin{lemma}
  For every $j,\ell,k\in\Z_p$, $\CIZ$ defined below satisfies $\CIZ\ket{j,\ell,k}=\omega^{jk}\ket{j,\ell,k}$, and multiplication is taken modulo $p$.
    \[
    \scalebox{1}{\input{figures/Preliminaries/CIZ-soundness.tikz}}
  \]
\label{lem:CIZ-soundness}
\end{lemma}

\begin{proof}
Let $\ket{j}$, $\ket{\ell}$, and $\ket{k}$ be arbitrary computational basis states.
\begin{align*}
\CIZ\ket{j,\ell,k}
&= (I\otimes \SWAP)(\CZ\otimes I)(I\otimes \SWAP)\ket{j,\ell,k} \xlongequal[]{\eqref{eq:SWAP-action}} (I\otimes \SWAP)(\CZ\otimes I)\ket{j,k,\ell}\\
&\xlongequal[]{\eqref{eq:CZ}} (I\otimes \SWAP)\omega^{jk}\ket{j,k,\ell}\xlongequal[]{\eqref{eq:SWAP-action}} \omega^{jk}\ket{j,\ell,k}.
\end{align*}
\end{proof}  

Reasoning analogously to the proofs of \cref{lem:CNOT-soundness,lem:SWAP-soundness,lem:XC-soundness,lem:CIZ-soundness}, we can show that the derived generators $\CIX$ and $\XIC$ have the desired actions on the computational basis states.

\begin{corollary}
  For every $j,\ell,k\in\Z_p$, $\CIX$ defined below satisfies $\CIX\ket{j,\ell,k}=\ket{j,\ell,j+k}$, and addition is taken modulo $p$.
    \[
    \scalebox{1}{\input{figures/Preliminaries/CIX-soundness.tikz}}
  \]
\label{cor:CIX-soundness}
\end{corollary}  

\begin{corollary}
  For every $j,\ell,k\in\Z_p$, $\XIC$ defined below satisfies $\XIC\ket{j,\ell,k}=\ket{j+k,\ell,k}$, and addition is taken modulo $p$.
    \[
    \scalebox{1}{\input{figures/Preliminaries/XIC-soundness.tikz}}
  \]
\label{cor:XIC-soundness}
\end{corollary}

\subsection{Soundness of the Reduced Relations}
\label{subsec:soundness}

Finally, we show that the relations of \cref{fig:rewriterules6} are sound. That is, we show that, for each relation in \cref{fig:rewriterules6}, the circuit on the left-hand side of the relation and the circuit on the right-hand side of the relation correspond to the same linear map. To this end, we use path-sums~\cite{amy2019towards,koh2017computing} to show that the left- and right-hand side circuits act identically on an arbitrary computational basis state, and therefore on all computational basis states. Because linear maps are completely determined by their action on a basis, this suffices to establish the soundness of the relation. An advantage of path-sums for this task is that it enables us to conveniently reason for an unspecified odd prime $p$. As a result, we can establish the soundness of the relations in \cref{fig:rewriterules6} for all odd primes in one fell swoop.

\paragraph{The Soundness of the Single-Qudit Clifford Relations} 

We first show that the single-qudit Clifford relations hold on the nose. That is, we show that the left- and right-hand side circuits of each relation in \cref{fig:rewriterules6} are equal as linear maps, not just up to a global phase.

\begin{lemma}
Let $g \in \Z_p^*$ and $k \in \Z$. $Z$ and $M_g$ are defined in \eqref{eq:Z} and \eqref{eq:Mg}, respectively. The following equations are sound.
\begin{align}
  (-\omega)^{2p} &= 1 \reltag{C}{eq:omega-soundness}\\
  S^p &= 1 \reltag{C}{eq:S-soundness}\\
  H^2 &= M_{-1}\cdot (-1)^{\frac{p-1}{2}} \reltag{C}{eq:H2-soundness}\\
  (M_g)^k &= M_{g^k} \reltag{C}{eq:Mg-soundness}\\
  M_gS &= Z^{(1-g)/(2g^2)}S^{g^{-2}}M_{g} \reltag{C}{eq:Mg-S-soundness}\\
  SH^2SH^2 &= H^2SH^2S \reltag{C}{eq:SH2-SH2-soundness}
\end{align}
\label{lem:single-qudit-soundness}
\end{lemma}  

\begin{proof}
By \eqref{eq:neg-one-omega-power}, $(-\omega)^{2p} = \left((-\omega)^p\right)^2 = (-1)^2 = 1$, so \eqref{eq:omega-soundness} holds. By \cref{cor:S-order}, \eqref{eq:S-soundness} holds. Now consider an arbitrary computational basis state $\ket{j}$. To show \eqref{eq:H2-soundness} holds, note that since $(p-1)/2$ is an integer, \cref{def:lambda} gives $\lambda_p^{-2} = e^{-(p-1)\pi i/2} = (-1)^{\frac{p-1}{2}}$. Hence,
\[
(-1)^{\frac{p-1}{2}}\cdot M_{-1}\ket{j}\xlongequal[]{\eqref{eq:Mg-action}} (-1)^{\frac{p-1}{2}}\ket{-j} = \lambda_p^{-2}\ket{-j}\xlongequal[]{\cref{prop:hsquared}}H^2\ket{j}.
\]

To show \eqref{eq:Mg-soundness} holds, we have

\[
(M_g)^k\ket{j} =\ket{g^kj} \xlongequal[]{\eqref{eq:Mg-action}} M_{g^k}\ket{j}.
\]

To show \eqref{eq:Mg-S-soundness} holds, we have
\begin{equation}
\LHS\ket{j} = M_gS\ket{j} \xlongequal[]{\eqref{eq:S}} \omega^{\frac{j(j-1)}{2}}M_g\ket{j} \xlongequal[]{\eqref{eq:Mg-action}}\omega^{\frac{j(j-1)}{2}}\ket{gj}
\label{eq:Mg-S-soundness-1}
\end{equation}
\begin{align}
\RHS\ket{j} &= Z^{(1-g)/(2g^2)}S^{g^{-2}}M_g\ket{j} \xlongequal[]{\eqref{eq:Mg-action}} Z^{(1-g)/(2g^2)}S^{g^{-2}}\ket{gj}  \notag\\
&\xlongequal[]{\cref{lem:multiple-s}} Z^{(1-g)/(2g^2)}\omega^{\frac{gj(gj-1)}{2g^2}}\ket{gj} \notag\\
&\xlongequal[]{\eqref{eq:Z-action}} \omega^{\frac{(1-g)gj}{2g^2}}\omega^{\frac{gj(gj-1)}{2g^2}}\ket{gj}.
\label{eq:Mg-S-soundness-2}
\end{align}

We can simplify the scalar of \eqref{eq:Mg-S-soundness-2}, since 
\[
\frac{(1-g)gj}{2g^2}+\frac{gj(gj-1)}{2g^2}= \frac{gj(1-g+gj-1)}{2g^2} = \frac{gj(gj-g)}{2g^2} = \frac{j(j-1)}{2}.
\]

Combining \eqref{eq:Mg-S-soundness-1} and \eqref{eq:Mg-S-soundness-2}, $\LHS \ket{j} = \RHS \ket{j}$ for all $j\in \Z_p$, which establishes the soundness of \eqref{eq:Mg-S-soundness}.

To show \eqref{eq:SH2-SH2-soundness} holds, we have
\begin{align}
\LHS\ket{j} &= SH^2SH^2\ket{j} \xlongequal[]{\cref{prop:hsquared}} \lambda_p^{-2}\cdot SH^2S\ket{-j} \xlongequal[]{\eqref{eq:S}} \lambda_p^{-2}\omega^{\frac{(-j)(-j-1)}{2}} SH^2\ket{-j} \notag\\
&\xlongequal[]{\cref{prop:hsquared}} \lambda_p^{-4}\omega^{\frac{(-j)(-j-1)}{2}}S\ket{j}\xlongequal[]{\eqref{eq:S}} \lambda_p^{-4}\omega^{\frac{(-j)(-j-1)}{2}}\omega^{\frac{j(j-1)}{2}}\ket{j}\notag\\ 
&\xlongequal[]{\cref{lem:order-lambda-p}} \omega^{\frac{(-j)(-j-1)+j(j-1)}{2}}\ket{j} = \omega^{j^2}\ket{j}.
\label{eq:SH2-SH2-soundness-1}
\end{align}
% We can simplify the scalar of \eqref{eq:SH2-SH2-soundness-1}, since 
% \[
% \frac{(-j)(-j-1)+j(j-1)}{2}= \frac{j(j+1)+j(j-1)}{2} = \frac{2j^2}{2} = j^2.
% \]

\begin{align}
\RHS\ket{j} &= H^2SH^2S\ket{j} \xlongequal[]{\eqref{eq:S}} \omega^{\frac{j(j-1)}{2}}H^2SH^2\ket{j} \xlongequal[]{\cref{prop:hsquared}} \lambda_p^{-2}\omega^{\frac{j(j-1)}{2}}H^2S\ket{-j} \notag\\
&\xlongequal[]{\eqref{eq:S}} \lambda_p^{-2}\omega^{\frac{j(j-1)}{2}}\omega^{\frac{(-j)(-j-1)}{2}}H^2\ket{-j} \xlongequal[]{\cref{prop:hsquared}} \lambda_p^{-4}\omega^{\frac{j(j-1)}{2}}\omega^{\frac{(-j)(-j-1)}{2}}\ket{j} \notag\\
&\xlongequal[]{\cref{lem:order-lambda-p}} \omega^{\frac{j(j-1)+(-j)(-j-1)}{2}}\ket{j} = \omega^{j^2}\ket{j}.
\label{eq:SH2-SH2-soundness-2}
\end{align}

Combining \eqref{eq:SH2-SH2-soundness-1} and \eqref{eq:SH2-SH2-soundness-2}, $\LHS \ket{j} = \RHS \ket{j}$ for all $j\in \Z_p$, which establishes the soundness of \eqref{eq:SH2-SH2-soundness}.
\end{proof}  

\begin{lemma}
Let $g \in \Z_p^*$. $M_g$, $\CX$, and $\SWAP$ are defined in \eqref{eq:Mg}, \eqref{eq:CX}, and \eqref{eq:SWAP}, respectively. The following equations are sound.
\begin{align}
  \CZ^p &= 1 \reltag{C}{eq:CZ-soundness}\\
  \SWAP^2 &= 1 \reltag{C}{eq:SWAP-soundness}\\
  \CZ(S\otimes I) &= (S\otimes I)\CZ \reltag{C}{eq:CZ-S-soundness}\\
  \CZ(M_g\otimes I) &= (M_g\otimes I)\CZ^g \reltag{C}{eq:CZ-Mg-soundness}\\
  \SWAP(S\otimes I) &= (I\otimes S)\SWAP \reltag{C}{eq:SWAP-S-soundness}\\
  \SWAP(H\otimes I) &= (I\otimes H)\SWAP \reltag{C}{eq:SWAP-H-soundness}\\
  (S^{-1}\otimes S^{-1})\CX^{-1}(I\otimes S)\CX &= \CZ \reltag{C}{eq:CX-CZ-soundness}
\end{align}
\label{lem:two-qudit-soundness}
\end{lemma} 

\begin{proof}
By \cref{cor:CZ-order}, \eqref{eq:CZ-soundness} holds. To show \eqref{eq:SWAP-soundness} holds, it suffices to show that $\SWAP^2\ket{j,\ell} = \ket{j,\ell}$ for all $j,\ell\in \Z_p$. By \cref{lem:SWAP-soundness}, we have $\SWAP^2\ket{j,\ell} = \SWAP\ket{\ell,j} = \ket{j,\ell}$, which establishes the soundness of \eqref{eq:SWAP-soundness}. To show \eqref{eq:CZ-S-soundness} holds, we have
\begin{align*}
\CZ(S\otimes I)\ket{j,\ell} &\xlongequal[]{\eqref{eq:S}} \omega^{\frac{j(j-1)}{2}}\CZ\ket{j,\ell} \xlongequal[]{\eqref{eq:CZ}} \omega^{\frac{j(j-1)}{2}}\omega^{j\ell}\ket{j,\ell} \\
&\xlongequal[]{\eqref{eq:S}} \omega^{j\ell}(S\otimes I)\ket{j,\ell} \xlongequal[]{\eqref{eq:CZ}} (S\otimes I)\CZ\ket{j,\ell}.
\end{align*}

To show \eqref{eq:CZ-Mg-soundness} holds, we have
\begin{align*}
\CZ(M_g\otimes I)\ket{j,\ell} &\xlongequal[]{\eqref{eq:Mg-action}} \CZ\ket{gj,\ell} \xlongequal[]{\eqref{eq:CZ}} \omega^{gj\ell}\ket{gj,\ell} \\
&\xlongequal[]{\eqref{eq:Mg-action}} \omega^{gj\ell}(M_g\otimes I)\ket{j,\ell} \xlongequal[\cref{lem:multiple-cz}]{\eqref{eq:CZ}} (M_g\otimes I)\CZ^g\ket{j,\ell}.
\end{align*}

To show \eqref{eq:SWAP-S-soundness} holds, we have
\begin{align*}
\SWAP(S\otimes I)\ket{j,\ell}&\xlongequal[]{\eqref{eq:S}} \omega^{\frac{j(j-1)}{2}}\SWAP\ket{j}\ket{\ell} \xlongequal[]{\eqref{eq:SWAP-action}} \omega^{\frac{j(j-1)}{2}}\ket{\ell}\ket{j} \\
&\xlongequal[]{\eqref{eq:S}} (I\otimes S)\ket{\ell}\ket{j}\xlongequal[]{\eqref{eq:SWAP-action}}(I\otimes S)\SWAP\ket{j}\ket{\ell}.
\end{align*}

To show \eqref{eq:SWAP-H-soundness} holds, we have
\begin{align*}
\SWAP(H\otimes I)\ket{j,\ell}&\xlongequal[]{\eqref{eq:H}} \frac{1}{\lambda_p\sqrt{p}}\cdot\SWAP\left( \sum_{k=0}^{p-1}\omega^{kj}\ket{k,\ell}\right) \xlongequal[]{\eqref{eq:SWAP-action}} \frac{1}{\lambda_p\sqrt{p}}\cdot\sum_{k=0}^{p-1}\omega^{kj}\ket{\ell,k} \\
&= \ket{\ell}\left(\frac{1}{\lambda_p\sqrt{p}}\cdot\sum_{k=0}^{p-1}\omega^{kj}\ket{k}\right) \xlongequal[]{\eqref{eq:H}} (I\otimes H) \ket{\ell,j} \xlongequal[]{\eqref{eq:SWAP-action}} (I\otimes H)\SWAP \ket{j,\ell}.
\end{align*}

To show \eqref{eq:CX-CZ-soundness} holds, we have
\begin{align}
(S^{-1}\otimes S^{-1})\CX^{-1}(I\otimes S)\CX\ket{j,\ell} &\xlongequal[]{\eqref{eq:CX-action}} (S^{-1}\otimes S^{-1})\CX^{-1}(I\otimes S)\ket{j,j+\ell}\notag\\
&\xlongequal[]{\eqref{eq:S}} \omega^{\frac{(j+\ell)(j+\ell-1)}{2}}(S^{-1}\otimes S^{-1})\CX^{-1}\ket{j,j+\ell} \notag\\
&\xlongequal[]{\cref{lem:multiple-cx}} \omega^{\frac{(j+\ell)(j+\ell-1)}{2}}(S^{-1}\otimes S^{-1})\ket{j,\ell}\notag\\
&\xlongequal[]{\cref{lem:multiple-s}} \omega^{\frac{(j+\ell)(j+\ell-1)}{2}}\omega^{\frac{-j(j-1)}{2}}\omega^{\frac{-\ell(\ell-1)}{2}}\ket{j,\ell}\label{eq:CX-CZ-soundness-1}
\end{align}
The scalar of \eqref{eq:CX-CZ-soundness-1} can be simplified as $\omega^{j\ell}$, since
\begin{align*}
\frac{(j+\ell)(j+\ell-1)}{2}-\frac{j(j-1)}{2}-\frac{\ell(\ell-1)}{2} &= \frac{j^2+2j\ell+\ell^2-(j+\ell)-j^2+j-\ell^2 +\ell}{2} = j\ell.
\end{align*}

By \eqref{eq:CZ}, we have $\omega^{j\ell}\ket{j,\ell} = \CZ\ket{j,\ell}$, which establishes the soundness of \eqref{eq:CX-CZ-soundness}.
\end{proof}  

\begin{lemma}
$\CX$, $\SWAP$, and $\CIZ$ are defined in \eqref{eq:CX}, \eqref{eq:SWAP}, and \eqref{eq:CIZ}, respectively. The following equations are sound.
\begin{align}
  (\SWAP\otimes I)(I\otimes \SWAP)(\SWAP\otimes I) &= (I\otimes \SWAP)(\SWAP\otimes I)(I\otimes \SWAP)  \reltag{C}{eq:SWAP-SWAP-soundness}\\
  (I\otimes \SWAP)(\SWAP\otimes I)(I\otimes \CZ) &= (\CZ\otimes I)(I\otimes \SWAP)(\SWAP\otimes I) \reltag{C}{eq:SWAP-CZ-soundness}\\
  (I\otimes \CZ)(\CX\otimes I) &= \CIZ(\CX\otimes I)(I\otimes \CZ) \reltag{C}{eq:CZ-CX-soundness}
\end{align}
\label{lem:three-qudit-soundness}
\end{lemma} 

\begin{proof}
To show \eqref{eq:SWAP-SWAP-soundness} holds, it suffices to show that $(\SWAP\otimes I)(I\otimes \SWAP)(\SWAP\otimes I)\ket{j,\ell,k} = (I\otimes \SWAP)(\SWAP\otimes I)(I\otimes \SWAP)\ket{j,\ell,k}$ for all $j,\ell,k \in \Z_p$. Then,

\begin{align*}
(\SWAP\otimes I)(I\otimes \SWAP)(\SWAP\otimes I)\ket{j,\ell,k} 
&\xlongequal[]{\eqref{eq:SWAP-action}} (\SWAP\otimes I)(I\otimes \SWAP)\ket{\ell,j,k}  \\
&\xlongequal[]{\eqref{eq:SWAP-action}} (\SWAP\otimes I)\ket{\ell,k,j} \\
&\xlongequal[]{\eqref{eq:SWAP-action}} \ket{k,\ell,j} \\
&\xlongequal[]{\eqref{eq:SWAP-action}} (I\otimes \SWAP)\ket{k,j,\ell} \\
&\xlongequal[]{\eqref{eq:SWAP-action}} (I\otimes \SWAP)(\SWAP\otimes I)\ket{j,k,\ell} \\
&\xlongequal[]{\eqref{eq:SWAP-action}} (I\otimes \SWAP)(\SWAP\otimes I)(I\otimes \SWAP)\ket{j,\ell,k}.
\end{align*}

To show \eqref{eq:SWAP-CZ-soundness} holds, we have
\begin{align*}
(I\otimes \SWAP)(\SWAP\otimes I)(I\otimes \CZ)\ket{j,\ell,k} 
&\xlongequal[]{\eqref{eq:CZ}} \omega^{\ell k}(I\otimes \SWAP)(\SWAP\otimes I) \ket{j,\ell,k} \\
&\xlongequal[]{\eqref{eq:SWAP-action}} \omega^{\ell k}(I\otimes \SWAP) \ket{\ell,j,k} \\
&\xlongequal[]{\eqref{eq:SWAP-action}} \omega^{\ell k}\ket{\ell,k,j} \\
&\xlongequal[]{\eqref{eq:CZ}} (\CZ\otimes I)\ket{\ell,k,j} \\
&\xlongequal[]{\eqref{eq:SWAP-action}} (\CZ\otimes I)(I\otimes \SWAP)\ket{\ell,j,k} \\
&\xlongequal[]{\eqref{eq:SWAP-action}} (\CZ\otimes I)(I\otimes \SWAP)(\SWAP\otimes I)\ket{j,\ell,k}.
\end{align*}

To show \eqref{eq:CZ-CX-soundness} holds, we have

\begin{align*}
(I\otimes \CZ)(\CX\otimes I)\ket{j,\ell,k} 
&\xlongequal[]{\eqref{eq:CX-action}} (I\otimes \CZ)\ket{j,j+\ell,k} \\
&\xlongequal[]{\eqref{eq:CZ}} \omega^{(j+\ell)k}\ket{j,j+\ell,k} \\
&= \omega^{jk}\omega^{\ell k}\ket{j,j+\ell,k} \\
&\xlongequal[]{\eqref{eq:CIZ-action}} \omega^{\ell k}\CIZ\ket{j,j+\ell,k} \\
&\xlongequal[]{\eqref{eq:CX-action}} \omega^{\ell k}\CIZ(\CX\otimes I)\ket{j,\ell,k}\\
&\xlongequal[]{\eqref{eq:CZ}} \CIZ(\CX\otimes I)(I\otimes \CZ)\ket{j,\ell,k}.
\end{align*}
\end{proof}
\begin{proposition}
The relations in \cref{fig:rewriterules6} are sound.
\label{prop:soundness}
\end{proposition}

\begin{proof}
By \cref{lem:single-qudit-soundness,lem:two-qudit-soundness,lem:three-qudit-soundness}, the relations in \cref{fig:rewriterules6} are sound.
\end{proof}
% \iffalse
% \fi

%% file: scripts/appendix/clifford.tex
\section{Clifford Conjugation}
\label{sec:clifford}

\begin{lemma}
\label{lem:H-auto-var1}
    For all $a \in \Z_p$,
    \begin{equation}
    \input{figures/Preliminaries/single-qupit-Clifford-auto-var1.tikz}
    \label{eq:H-auto-var1}
\end{equation}
\end{lemma}

\begin{proof}
    Based on \cref{lem:automorphism}, we proceed by repeatedly applying the action of $H$ on $a$ copies of Pauli $X$ and $Z$.
\begin{align*}
    HX^aH^\dagger &= (HXH^\dagger)^a \xlongequal[]{\eqref{eq:H-S-auto}}Z^a \\ 
     HZ^aH^\dagger &= (HZH^\dagger)^a \xlongequal[]{\eqref{eq:H-S-auto}}\left(X^{-1}\right)^a = X^{-a}\\ 
     H^\dagger X^aH &= (H^\dagger XH)^a \xlongequal[]{\eqref{eq:H-S-auto}}\left(Z^{-1}\right)^a = Z^{-a}\qedhere
\end{align*}    
\end{proof}  

\begin{lemma}
\label{lem:H-auto-var2}
    \begin{equation}
    \input{figures/Preliminaries/single-qupit-Clifford-auto-var2.tikz}
    \label{eq:H-auto-var2}
\end{equation}
\end{lemma}    

\begin{proof}
    Based on \cref{lem:automorphism}, we proceed by repeatedly applying the action of $H$ on Pauli $X$ and $Z$.
    \[
        \scalebox{.9}{\input{figures/Preliminaries/single-qupit-Clifford-auto-var2-proof.tikz}}%\qedhere
    \]
\end{proof}

\begin{lemma}
    \begin{align}
        &\scalebox{.9}{\input{figures/Preliminaries/CNOT-action.tikz}}\label{eq:CNOT-action}\\[1 em]
        &\scalebox{.9}{\input{figures/Preliminaries/NOTC-action.tikz}}\label{eq:NOTC-action}
    \end{align}
\label{lem:soundness-CNOT}
\end{lemma}   

\begin{proof}
Recall that
\begin{align}
        &\scalebox{.8}{\input{figures/Preliminaries/CNOT.tikz}}\reltag{D}{eq:CNOT}\\[1 em]
        &\scalebox{.8}{\input{figures/Preliminaries/NOTC.tikz}}\reltag{D}{eq:NOTC}
    \end{align}

Up to symmetry, it suffices to show \eqref{eq:CNOT-action}. By \eqref{eq:H-S-auto}, \eqref{eq:CZ-auto}, and \eqref{eq:H-auto-var2}, we therefore have
\[
\scalebox{.7}{\input{figures/Preliminaries/CNOT-action-proof.tikz}}
\]
\end{proof}

\begin{lemma}
\label{lem:H-auto-var3}
    For all $a \in \Z_p$,
    \begin{equation}
    \scalebox{.9}{\input{figures/Preliminaries/single-qupit-Clifford-auto-var3.tikz}}
    \label{eq:H-auto-var3}
\end{equation}
\end{lemma}    

\begin{proof}
    By \cref{lem:automorphism}, $SZ^aS^\dagger = \left(SZS^\dagger\right)^a = Z^a$. For the actions of $S$ on $X^a$, we proceed by induction on $a$.
   \begin{description}
        \item[Base Case:]When $a=1$, the statement holds by \cref{lem:automorphism}.
        \item[Induction Hypothesis (\IH):] Suppose the statement holds for $a\geq 1$. That is,
            \begin{equation}
                SX^aS^\dagger = \omega^{\frac{a(a-1)}{2}}X^aZ^a,\qquad S^\dagger X^aS = \omega^{-\frac{a(a-1)}{2}}X^aZ^{-a}.
                \label{eq:base-case-var3}
            \end{equation}    
        \item[Induction Step:] We want to show that the statement holds for $a+1$.
            \begin{align*}
                SX^{a+1}S^\dagger &= S\left(X^{a}X\right)S^\dagger = \left(SX^aS^\dagger\right)\left(SXS^\dagger\right)\xlongequal[\eqref{eq:H-S-auto}]{\IH}\left(\omega^{\frac{a(a-1)}{2}}X^aZ^a\right)\left(XZ\right)\\
                &\xlongequal[]{\eqref{eq:XZ}}\omega^{\frac{a(a-1)}{2}}\omega^a X^{a+1}Z^{a+1}=\omega^{\frac{(a+1)a}{2}}X^{a+1}Z^{a+1}.
            \end{align*}  
            \begin{align*}
                S^\dagger X^{a+1} S&= S^\dagger\left(X^{a}X\right)S = \left(S^\dagger X^aS\right)\left(S^\dagger XS\right)\xlongequal[\eqref{eq:H-S-auto}]{\IH}\left(\omega^{-\frac{a(a-1)}{2}}X^aZ^{-a}\right)\left(XZ^{-1}\right)\\
                &\xlongequal[]{\eqref{eq:XZ}}\omega^{-\frac{a(a-1)}{2}}\omega^{-a} X^{a+1}Z^{-(a+1)}=\omega^{-\frac{(a+1)a}{2}}X^{a+1}Z^{-(a+1)}.\qedhere
            \end{align*}    
   \end{description}
\end{proof}  

\begin{lemma}
\label{lem:H-auto-var4}
    For all $a \in \Z_p$,
    \begin{equation}
    \scalebox{.9}{\input{figures/Preliminaries/single-qupit-Clifford-auto-var4.tikz}}
    \label{eq:H-auto-var4}
\end{equation}
\end{lemma}   

\begin{proof}
    By \cref{lem:automorphism}, $S^aZ\left(S^\dagger\right)^a = Z$. For the actions of $S^a$ on $X$, we proceed by induction on $a$.
   \begin{description}
        \item[Base Case:]When $a=1$, the statement holds by \cref{lem:automorphism}.
        \item[Induction Hypothesis (\IH):] Suppose the statement holds for $a\geq 1$. That is,
            \begin{equation}
                S^aX\left(S^\dagger\right)^a = XZ^a,\qquad \left(S^\dagger\right)^a XS^a = XZ^{-a}.
                \label{eq:base-case-var4}
            \end{equation}    
        \item[Induction Step:] We want to show that the statement holds for $a+1$.
            \begin{align*}
                S^{a+1}X\left(S^\dagger\right)^{a+1} \ &=\  S\left(S^aX\left(S^\dagger\right)^a\right)S^\dagger \ \xlongequal[]{\IH}\ S\left(XZ^a\right)S^\dagger \\
                &=\ \left(SXS^\dagger\right) \left(S Z^a S^\dagger\right)\ \xlongequal[]{\eqref{eq:H-S-auto}}\ XZZ^a \ =\  XZ^{a+1}\\
                \left(S^\dagger\right)^{a+1}XS^{a+1} \ &=\  S^\dagger\left(\left(S^\dagger \right)^a XS^a\right)S \ \xlongequal[]{\IH}\ S^\dagger\left(XZ^{-a}\right)S \ \\
                &=\ \left(S^\dagger XS\right) \left(S^\dagger Z^{-a} S\right)
                \ \xlongequal[]{\eqref{eq:H-S-auto}}\ XZ^{-1}Z^{-a} \ =\  XZ^{-(a+1)}.
            \end{align*}    
   \end{description}
\end{proof} 

\begin{lemma}
\label{lem:H-auto-var5}
    For all $a,b \in \Z_p$,
    \begin{equation}
        \scalebox{.9}{\input{figures/Preliminaries/single-qupit-Clifford-auto-var5.tikz}}
        \label{eq:H-auto-var5}
    \end{equation}
\end{lemma} 
\begin{proof}
    By \cref{lem:automorphism}, $S^bZ^a\left(S^\dagger\right)^b = Z^a$. For the actions of $S^b$ on $X^a$, we proceed by induction on $b$.
   \begin{description}
        \item[Base Case:]When $b=1$, the statement holds by \cref{lem:H-auto-var3}.
        \[
            \input{figures/Preliminaries/single-qupit-Clifford-auto-var3-X.tikz}
        \] 
        \item[Induction Hypothesis (\IH):] Suppose the statement holds for $b\geq 1$. That is,
            \begin{equation}
                S^bX^a\left(S^\dagger\right)^b = \omega^{\frac{ab(a-1)}{2}}X^aZ^{ab},\qquad \left(S^\dagger\right)^b X^aS^b = \omega^{-\frac{ab(a-1)}{2}}X^aZ^{-ab}.
                \label{eq:base-case-var5}
            \end{equation}    
        \item[Induction Step:] We want to show that the statement holds for $b+1$.
            \begin{align*}
                S^{b+1}X^a\left(S^\dagger\right)^{b+1} \ &=\  S\left(S^bX^a\left(S^\dagger\right)^b\right)S^\dagger \ \xlongequal[]{\IH}\ S\left(\omega^{\frac{ab(a-1)}{2}}X^aZ^{ab}\right)S^\dagger \\
                &=\ \omega^{\frac{ab(a-1)}{2}}\left(SX^aS^\dagger\right)\left(SZ^{ab}S^\dagger\right)\ \xlongequal[]{\eqref{eq:H-auto-var3}}\ \omega^{\frac{ab(a-1)}{2}}\left(\omega^{\frac{a(a-1)}{2}}X^aZ^a\right)Z^{ab} \\
                &=\ \omega^{\frac{a(b+1)(a-1)}{2}}X^aZ^{a(b+1)}.\\
                \left(S^\dagger\right)^{b+1}X^a S^{b+1} \ &=\ S^\dagger\left(\left(S^\dagger\right)^bX^aS^b\right)S \ \xlongequal[]{\IH}\ S^\dagger\left(\omega^{-\frac{ab(a-1)}{2}}X^aZ^{-ab}\right)S\\
                &= \ \omega^{-\frac{ab(a-1)}{2}}\left(S^\dagger X^aS\right)\left(S^\dagger Z^{-ab}S\right)\\
                &\xlongequal[]{\eqref{eq:H-auto-var3}}\ \omega^{-\frac{ab(a-1)}{2}}\left(\omega^{-\frac{a(a-1)}{2}}X^aZ^{-a}\right)Z^{-ab}\\
                &=\ \omega^{-\frac{a(b+1)(a-1)}{2}}X^aZ^{-a(b+1)}.\qedhere
            \end{align*}    
   \end{description}
\end{proof} 

\begin{lemma}
\label{lem:H-auto-var6}
    For all $a,b \in \Z_p$,
    \begin{equation}
    \left(XZ^{-a}\right)^b = \omega^{-\frac{ab(b-1)}{2}}X^bZ^{-ab}.
    \label{eq:H-auto-var6}
\end{equation}
\end{lemma} 
\begin{proof}
    We proceed by induction on $b$.
   \begin{description}
        \item[Base Case:]When $b=1$, the statement holds trivially.
        \item[Induction Hypothesis (\IH):] Suppose the statement holds for $b\geq 1$. That is,
            \begin{equation}
                \left(XZ^{-a}\right)^b = \omega^{-\frac{ab(b-1)}{2}}X^bZ^{-ab}.
                \label{eq:base-case-var6}
            \end{equation}    
        \item[Induction Step:] We want to show that the statement holds for $b+1$.
            \begin{align*}
                \left(XZ^{-a}\right)^{b+1} &= \left(XZ^{-a}\right)^{b}\left(XZ^{-a}\right)\xlongequal[]{\IH}\left(\omega^{-\frac{ab(b-1)}{2}}X^bZ^{-ab}\right)\left(XZ^{-a}\right)\\
                &\xlongequal[]{\eqref{eq:XZ}}\omega^{-\frac{ab(b-1)}{2}}\omega^{-ab}X^{b+1}Z^{-ab}Z^{-a}=\omega^{-\frac{a(b+1)b}{2}}X^{b+1}Z^{-a(b+1)}.\qedhere
            \end{align*}    
   \end{description}
\end{proof}

\begin{lemma}
    The $X$ gate, defined in \eqref{eq:def-X}, acts on the Pauli generators as follows.
    \[
        \input{figures/Preliminaries/X-action.tikz}
    \]
\label{lem:X-action}
\end{lemma}

\begin{proof}
    Based on \eqref{eq:H-S-auto} and \cref{lem:H-auto-var1,lem:H-auto-var2,lem:H-auto-var3,lem:H-auto-var4,lem:H-auto-var5}, we have
    \[
        \scalebox{.8}{\input{figures/Preliminaries/X-implementation-proof.tikz}}
    \]
\end{proof}    

\begin{lemma}
    The $Z$ gate, defined in \eqref{eq:def-Z}, acts on the Pauli generators as follows.
    \[
        \input{figures/Preliminaries/Z-action.tikz}
    \]
\label{lem:Z-action}
\end{lemma}

\begin{proof}
    Let $S'=H^2SH^2$. Based on \eqref{eq:H-S-auto} and \cref{lem:H-auto-var4}, we have

    \[
        \input{figures/Preliminaries/Z-implementation-proof.tikz}\qedhere
    \]
\end{proof}   

\begin{lemma}
\label{lem:SWAP-action}
    The $\SWAP$ gate, defined in \eqref{eq:def-SWAP}, acts on the Pauli generators as follows.
    \begin{equation}
        \scalebox{.9}{\input{figures/Preliminaries/SWAP-action.tikz}}
        \label{eq:SWAP-Pauli-action}
    \end{equation}
\end{lemma}

\begin{proof}
The global phase $\lambda_p^2$ in \eqref{eq:def-SWAP} does not affect conjugation. Up to symmetry, it suffices to show the action on $X\otimes I$ and $Z\otimes I$. By \eqref{eq:H-S-auto} and \eqref{eq:CZ-auto}, we can show that
\[
\input{figures/Preliminaries/SWAP-XI.tikz}
\]

\[
\input{figures/Preliminaries/SWAP-ZI.tikz}
\]
\end{proof}    

%% file: scripts/appendix/comp-def.tex
\newpage

\section{Actions of the Normal Boxes}
\label{sec:comp-def}

% \begin{figure}[!htb]
%     \centering
% % \scalebox{0.8}{\tikzfig{figures/NormalForm/normalBoxesMacroSimp-v5}}
% \scalebox{0.8}{\tikzfig{figures/NormalForm/normalBoxesMacroSimp-v4}}
% % \centering \scalebox{0.8}{\tikzfig{figures/agda-box-defs-rels/box_def_c}}
%     \caption{The concrete implementations of the Z- and X-normal boxes. Here $a,b\in \Z_p$ and $a\neq 0$.}
%     \label{fig:Z-and-X-Normal-Boxes-imp-alt}
% \end{figure}

\Cref{fig:A-Boxes-auto-phase-free,fig:B-Boxes-auto-phase-free,fig:D-Boxes-auto-phase-free,fig:E-Boxes-auto} list the actions of all the normal boxes on all Pauli generators up to phase.

\begin{figure}[!htb]
    \centering
    \scalebox{1}{\input{figures/PhaseFreeNormalForm/ABox-phasefree-auto-v3.tikz}}
    \caption{The action of $A$ boxes on the Pauli generators up to phase.}
    \label{fig:A-Boxes-auto-phase-free}
\end{figure}

\begin{figure}[!htb]
    \centering
    \scalebox{.8}{\input{figures/PhaseFreeNormalForm/BBox-phasefree-auto.tikz}}
    \caption{The action of $B$ boxes on the Pauli generators up to phase.}
    \label{fig:B-Boxes-auto-phase-free}
\end{figure}

\begin{figure}[!htb]
    \centering
    \scalebox{.8}{\input{figures/PhaseFreeNormalForm/DBox-phasefree-auto.tikz}}
    \caption{The action of $D$ boxes on the Pauli generators up to phase.}
    \label{fig:D-Boxes-auto-phase-free}
\end{figure}

\begin{figure}[!htb]
    \centering
    \scalebox{1}{\input{figures/PhaseFreeNormalForm/E-Boxes-auto.tikz}}
    \caption{The action of $E$ boxes on the Pauli generators up to phase.}
    \label{fig:E-Boxes-auto}
\end{figure}

\clearpage

%% file: scripts/appendix/push-normal.tex
\section{Pushing Through a Normal Box}
\label{sec:push-normal}

The goal of this section is to establish what we call the proof-of-principle box relations as shown in \cref{fig:two-qubit-box-relations}. That is, the dirty gates resulting from pushing the generators through the boxes that occur in the normal form have the shape that we want them to have.

\begin{figure}[!htbp]
    \[
    \scalebox{.65}{\input{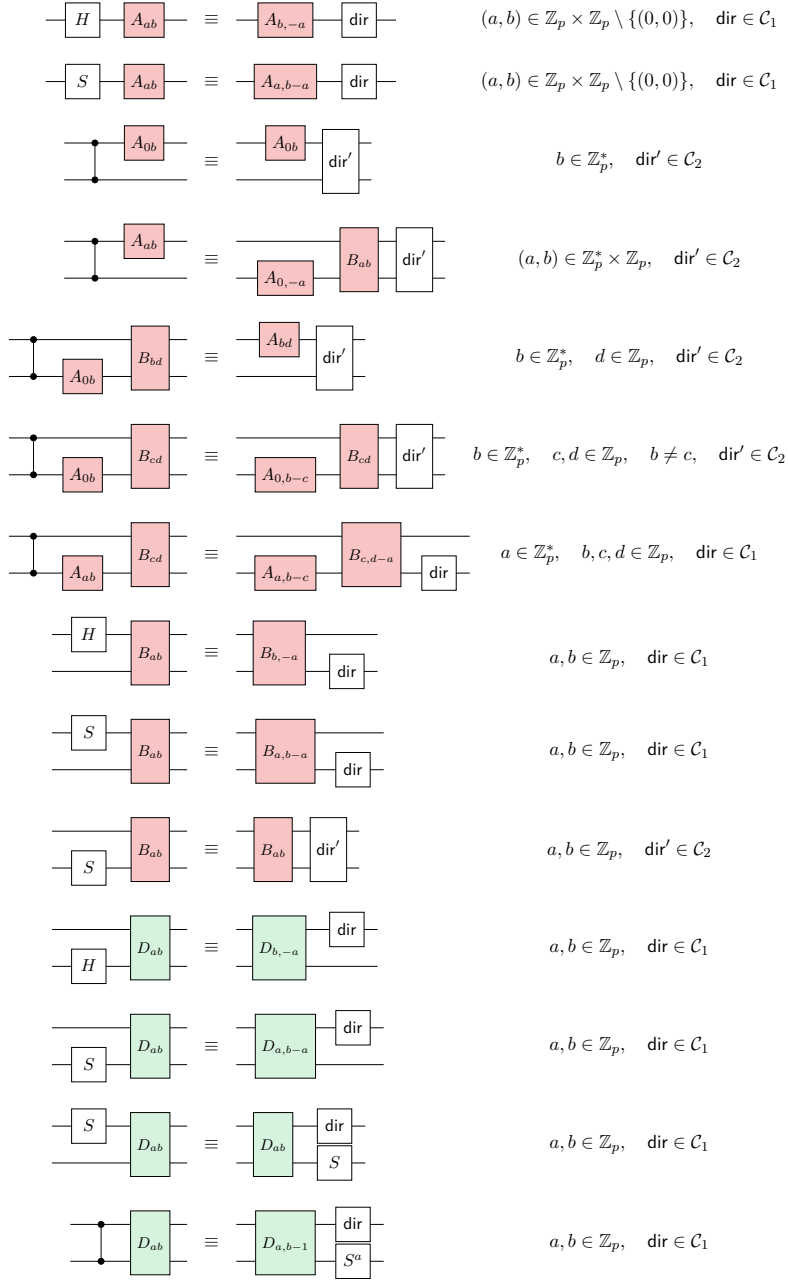}}
    \]
    \caption{Proof-of-principle single- and two-qudit box relations. They demonstrate the structure of residual dirty gates in all scenarios when we normalise a two-qudit Clifford operator through ``pushing through a normal box''.}
    \label{fig:two-qubit-box-relations}
\end{figure}    

Recall that we use $\mathcal{C}_n$ to denote the $n$-qudit Clifford group, where $p$ is an odd prime, and we use $\mathcal{B}_n$ to denote the set of $n$-qudit Pauli generators, which together with $-\omega$ generate the $n$-qudit Pauli group $\mathcal{P}_n$.
\begin{proposition} \label{prop:untangled}
     Let $k_1, k_2\in \N$, and $C\in \mathcal{C}_{k_1+k_2}$. Suppose that for every $B \in \mathcal{B}_{k_1}$ there exists a $B' \in \mathcal{P}_{k_1}$ such that
     \begin{equation}
         C(B \otimes I)C^\dagger = B' \otimes I. \label{eq:CBC}
     \end{equation}
     Then there exist $C_1 \in \Clifford_{k_1}$ and $C_2 \in \Clifford_{k_2}$ such that $C = C_1\otimes C_2$, up to scalar. Here, $I$ is the identity map on $\C^{p^{k_2}}$.
\end{proposition}
\begin{proof}
    By \Cref{eq:CBC}, the mapping $\Phi:\mathcal{B}_{k_1}\rightarrow \mathcal{M}_{p^{k_1}}(\C)$ sending $B\mapsto \text{tr}_2\left( C(B\otimes I)C^{\dagger} \right)$ extends to an automorphism $\Phi$ of $\mathcal{P}_{k_1}$ such that $C(B\otimes I)C^{\dagger}=\Phi(B)\otimes I$, where $\text{tr}_2$ is the normalised partial trace $\frac{1}{p^{k_2}} id \otimes \tr$ that traces out the second system $\mathcal{M}_{p^{k_2}}(\C)$. Hence, there exists a $C_1\in \mathcal{C}_{k_1}$ such that $C_1 B C_1^\dagger=\Phi(B)$. The rest follows from \cite[Lemma E.6]{qutrit2024}.
\end{proof}

Firstly, note that \eqref{eq:B-auto} and \eqref{eq:D-auto} follow from \cref{fig:B-Boxes-auto-phase-free,fig:D-Boxes-auto-phase-free}. Based on this, we can derive proof-of-principle box relations as shown in \cref{fig:two-qubit-box-relations}.
\begin{align}
    \input{figures/NormalForm/B-auto-LHS.tikz}\quad &\quad\input{figures/NormalForm/B-auto-RHS.tikz}\label{eq:B-auto}\\[2 em]
    \input{figures/NormalForm/D-auto-LHS.tikz}\quad&\quad\input{figures/NormalForm/D-auto-RHS.tikz}\label{eq:D-auto}
\end{align}    

\begin{lemma}
    For all $(a,b)\in\Z_p\times \Z_p\setminus\{(0,0)\}$, there exist $\dir,\;\dir' \in \circuitClifford_1$ such that
    \begin{equation}
        \input{figures/NormalForm/HSA.tikz}
        \label{eq:HSA}
    \end{equation}
    \label{lem:HSA}
\end{lemma}

\begin{proof}
    Computing the preimage of $Z$ in the \LHS of \eqref{eq:HSA} yields
\[
\input{figures/NormalForm/HSA-preimage.tikz}
\]

By the uniqueness of the symplectic normal form, $X^bZ^{-a}$ and $X^aZ^{b-a}$ imply that the updated normal boxes should be those in the \RHS of \eqref{eq:HSA}.
    \iffalse
    By the uniqueness of the normal form and the action of the $A$ boxes specified in \cref{fig:normal-box-auto}, it suffices to find the preimage of $X^aZ^b$ under the action of $H$ and $S$ gates.
    \[
    H^\dagger (X^aZ^b) H = (H^\dagger X^a H)(H^\dagger Z^b H)\xlongequal[]{\eqref{eq:H-S-auto}}Z^{-a}X^b\xlongequal[]{\eqref{eq:XZ}}\omega^{-ab}X^bZ^{-a}.
    \]

    Hence, after pushing $H$ gate through an $A_{ab}$ box, the updated $A$ box must be of the form $A_{b,-a}$.

    \[
    S^\dagger (X^aZ^b) S = (S^\dagger X^a S)(S^\dagger Z^b S)\xlongequal[]{\eqref{eq:H-S-auto}}\left((XZ^{-1})\cdots(XZ^{-1})\right)Z^b\xlongequal[]{\eqref{eq:XZ}}\omega^{\frac{-a(a-1)}{2}}X^aZ^{-a}Z^b=\omega^{\frac{-a(a-1)}{2}}X^aZ^{b-a}.
    \]

    Hence, after pushing an $S$ gate through an $A_{ab}$ box, the updated $A$ box must be of the form $A_{a,b-a}$.\qedhere
    \fi
\end{proof}

\begin{lemma}
    For all $(a,b)\in\Z_p\times \Z_p\setminus\{(0,0)\}$, there exist $\dir,\;\dir' \in \circuitClifford_2$ such that
    \begin{equation}
        \input{figures/NormalForm/CZ-A.tikz}
        \label{eq:CZ-A}
    \end{equation}    
    \label{lem:CZ-A}
\end{lemma}

\begin{proof}
Computing the preimage of $Z \otimes I$ in the \LHS of \eqref{eq:CZ-A} yields
\[
\input{figures/NormalForm/CZ-A-preimage.tikz}
\]

By the uniqueness of the symplectic normal form, $Z^b \otimes I$ and $X^aZ^b \otimes Z^{-a}$ imply that the updated normal boxes should be those in the \RHS of \eqref{eq:CZ-A}.
\end{proof}  

\begin{lemma}
    For all $(a,b)\in\Z_p\times \Z_p\setminus\{(0,0)\}$ and $(c,d)\in\Z_p\times \Z_p$, there exist $\dir,\;\dir'\in \circuitClifford_2$ and $\dir''\in \circuitClifford_1$ such that
    \begin{equation}
        \input{figures/NormalForm/CZ-AB-v2.tikz}
        \label{eq:CZ-A-B}
    \end{equation}    
    \label{lem:CZ-A-B}
\end{lemma}

\begin{proof}
Computing the preimage of $Z \otimes I$ in the \LHS of \eqref{eq:CZ-A-B} yields
\[
\input{figures/NormalForm/CZ-AB-preimage.tikz}
\]

By the uniqueness of the symplectic normal form, $X^bZ^d \otimes I$ and $X^cZ^{d-a} \otimes X^aZ^{b-c}$ imply that the updated normal boxes should be those in the \RHS of \eqref{eq:CZ-A-B}.

Next, we consider the case when $a\neq 0$. Suppose that for some $\dir'' \in \circuitClifford_2$,
\[
\input{figures/NormalForm/CZ-AB.tikz}
\]

By \cref{prop:untangled}, it remains to show that
\[
\input{figures/NormalForm/CZ-AB-dir-auto.tikz}
\]

Equivalently, we can compute the preimage of $Z\otimes I$ and $X\otimes I$ in the \LHS of \eqref{eq:CZ-AB-int}.

\begin{equation}
\input{figures/NormalForm/CZ-AB-int.tikz}
\label{eq:CZ-AB-int}
\end{equation}    

% By \cref{fig:B-Boxes-auto-phase-free}, for all $a,b\in \Z_p$, we have
% \[
% \tikzfig{figures/NormalForm/B-auto}
% \]

By \eqref{eq:B-auto}, we have
\[
\input{figures/NormalForm/CZ-AB-int2.tikz}
\]

This completes the proof.
\end{proof}  

\begin{lemma}
    For all $a,b\in\Z_p$, there exists $\dir\in \circuitClifford_1$ such that
    \begin{equation}
        \input{figures/NormalForm/HI-B-v2.tikz}
        \label{eq:HI-B}
    \end{equation}    
    \label{lem:HI-B}
\end{lemma}

\begin{proof}
Computing the preimage of $Z \otimes I$ in the \LHS of \eqref{eq:HI-B} yields
\[
\input{figures/NormalForm/HI-B-preimage.tikz}
\]

By the uniqueness of the symplectic normal form, $X^bZ^{-a} \otimes Z$ implies that the updated normal box should be the one in the \RHS of \eqref{eq:HI-B}. Suppose that for some $\dir \in \circuitClifford_2$,
\[
\input{figures/NormalForm/HI-B.tikz}
\]

By \cref{prop:untangled}, it remains to show that
\[
\input{figures/NormalForm/HI-B-dir-auto.tikz}
\]

Equivalently, we can compute the preimage of $Z\otimes I$ and $X\otimes I$ in the \LHS of \eqref{eq:HI-B-int}.

\begin{equation}
\input{figures/NormalForm/HI-B-int.tikz}
\label{eq:HI-B-int}
\end{equation}    

% By \cref{fig:B-Boxes-auto-phase-free}, for all $a,b\in \Z_p$, we have
% \[
% \tikzfig{figures/NormalForm/B-auto}
% \]

% It follows that
By \eqref{eq:B-auto}, we have

\[
\input{figures/NormalForm/HI-B-int2.tikz}
\]

This completes the proof.
\end{proof}  

\begin{lemma}
    For all $a,b\in\Z_p$, there exists $\dir\in \circuitClifford_1$ such that
    \begin{equation}
        \input{figures/NormalForm/SI-B-v2.tikz}
        \label{eq:SI-B}
    \end{equation}    
    \label{lem:SI-B}
\end{lemma}

\begin{proof}
Computing the preimage of $Z \otimes I$ in the \LHS of \eqref{eq:SI-B} yields
\[
\input{figures/NormalForm/SI-B-preimage.tikz}
\]

By the uniqueness of the symplectic normal form, $X^aZ^{b-a} \otimes Z$ implies that the updated normal box should be the one in the \RHS of \eqref{eq:SI-B}. Suppose that for some $\dir \in \circuitClifford_2$,
\[
\input{figures/NormalForm/SI-B.tikz}
\]

By \cref{prop:untangled}, it remains to show that
\[
\input{figures/NormalForm/SI-B-dir-auto.tikz}
\]

Equivalently, we can compute the preimage of $Z\otimes I$ and $X\otimes I$ in the \LHS of \eqref{eq:SI-B-int}.

\begin{equation}
\input{figures/NormalForm/SI-B-int.tikz}
\label{eq:SI-B-int}
\end{equation}    

% By \cref{fig:B-Boxes-auto-phase-free}, for all $a,b\in \Z_p$, we have
% \[
% \tikzfig{figures/NormalForm/B-auto}
% \]

% It follows that
By \eqref{eq:B-auto}, we have
\[
\input{figures/NormalForm/SI-B-int2.tikz}
\]

This completes the proof.
\end{proof} 

\begin{lemma}
    For all $a,b\in\Z_p$, there exists $\dir\in \circuitClifford_2$ such that
    \begin{equation}
        \input{figures/NormalForm/IS-B.tikz}
        \label{eq:IS-B}
    \end{equation}    
    \label{lem:IS-B}
\end{lemma}

\begin{proof}
Computing the preimage of $Z \otimes I$ in the \LHS of \eqref{eq:IS-B} yields
\[
\input{figures/NormalForm/IS-B-preimage.tikz}
\]

By the uniqueness of the symplectic normal form, $X^aZ^b \otimes Z$ implies that the updated normal box should be the one in the \RHS of \eqref{eq:IS-B}.
\end{proof} 

\begin{lemma}
    For all $a,b\in\Z_p$, there exists $\dir\in \circuitClifford_1$ such that
    \begin{equation}
        \input{figures/NormalForm/IH-D-v2.tikz}
        \label{eq:IH-D}
    \end{equation}    
    \label{lem:IH-D}
\end{lemma}

\begin{proof}
For any $j\in \Z_p$, computing the preimage of $I \otimes XZ^j$ in the \LHS of \eqref{eq:IH-D} yields
\[
\input{figures/NormalForm/IH-D-preimage.tikz}
\]

By the uniqueness of the symplectic normal form, $XZ^j \otimes X^bZ^{-a}$ implies that the updated normal box should be the one in the \RHS of \eqref{eq:IH-D}. Suppose that for some $\dir \in \circuitClifford_2$,
\[
\input{figures/NormalForm/IH-D.tikz}
\]

By \cref{prop:untangled}, it remains to show that
\[
\input{figures/NormalForm/IH-D-dir-auto.tikz}
\]

Equivalently, we can compute the preimage of $I\otimes X$ and $I\otimes Z$ in the \LHS of \eqref{eq:IH-D-int}.

\begin{equation}
\input{figures/NormalForm/IH-D-int.tikz}
\label{eq:IH-D-int}
\end{equation}    

% By \cref{fig:D-Boxes-auto-phase-free}, for all $a,b\in \Z_p$, we have
% \[
% \tikzfig{figures/NormalForm/D-auto}
% \]

% It follows that

By \eqref{eq:D-auto}, we have

\[
\input{figures/NormalForm/IH-D-int2.tikz}
\]

This completes the proof.
\end{proof} 

\begin{lemma}
    For all $a,b\in\Z_p$, there exists $\dir\in \circuitClifford_1$ such that
    \begin{equation}
        \input{figures/NormalForm/IS-D-v2.tikz}
        \label{eq:IS-D}
    \end{equation}    
    \label{lem:IS-D}
\end{lemma}

\begin{proof}
For any $j\in \Z_p$, computing the preimage of $I \otimes XZ^j$ in the \LHS of \eqref{eq:IS-D} yields
\[
\input{figures/NormalForm/IS-D-preimage.tikz}
\]

By the uniqueness of the symplectic normal form, $XZ^j \otimes X^aZ^{b-a}$ implies that the updated normal box should be the one in the \RHS of \eqref{eq:IS-D}.  Suppose that for some $\dir \in \circuitClifford_2$,
\[
\input{figures/NormalForm/IS-D.tikz}
\]

By \cref{prop:untangled}, it remains to show that
\[
\input{figures/NormalForm/IS-D-dir-auto.tikz}
\]

Equivalently, we can compute the preimage of $I\otimes X$ and $I\otimes Z$ in the \LHS of \eqref{eq:IS-D-int}.

\begin{equation}
\input{figures/NormalForm/IS-D-int.tikz}
\label{eq:IS-D-int}
\end{equation}    

% By \cref{fig:D-Boxes-auto-phase-free}, for all $a,b\in \Z_p$, we have
% \[
% \tikzfig{figures/NormalForm/D-auto}
% \]

% It follows that

By \eqref{eq:D-auto}, we have

\[
\input{figures/NormalForm/IS-D-int2.tikz}
\]

This completes the proof.
\end{proof} 

\begin{lemma}
    For all $a,b\in\Z_p$, there exists $\dir\in \circuitClifford_1$ such that
    \begin{equation}
        \input{figures/NormalForm/SI-D-v2.tikz}
        \label{eq:SI-D}
    \end{equation}    
    \label{lem:SI-D}
\end{lemma}

\begin{proof}
For any $j\in \Z_p$, computing the preimage of $I \otimes XZ^j$ in the \LHS of \eqref{eq:SI-D} yields
\[
\input{figures/NormalForm/SI-D-preimage.tikz}
\]

By the uniqueness of the symplectic normal form, $XZ^{j-1} \otimes X^aZ^b$ implies that the updated normal box should be the one in the \RHS of \eqref{eq:SI-D}.  Suppose that for some $\dir \in \circuitClifford_2$,
\[
\input{figures/NormalForm/SI-D.tikz}
\]

By \cref{prop:untangled}, it remains to show that
\[
\input{figures/NormalForm/SI-D-dir-auto.tikz}
\]

Equivalently, we can compute the preimage of $I\otimes X$ and $I\otimes Z$ in the \LHS of \eqref{eq:SI-D-int}.

\begin{equation}
\input{figures/NormalForm/SI-D-int.tikz}
\label{eq:SI-D-int}
\end{equation}

% By \cref{fig:D-Boxes-auto-phase-free}, for all $a,b\in \Z_p$, we have
% \[
% \tikzfig{figures/NormalForm/D-auto}
% \]

% It follows that
By \eqref{eq:D-auto}, we have
\[
\input{figures/NormalForm/SI-D-int2.tikz}
\]

Since $S$ sends $XZ^{-1}$ to $X$ and $Z$ to $Z$, this completes the proof.
\end{proof} 

\begin{lemma}
    For all $a,b\in\Z_p$, there exists $\dir\in \circuitClifford_1$ such that
    \begin{equation}
        \input{figures/NormalForm/CZ-D-v2.tikz}
        \label{eq:CZ-D}
    \end{equation}    
    \label{lem:CZ-D}
\end{lemma}

\begin{proof}
For any $j\in \Z_p$, computing the preimage of $I \otimes XZ^j$ in the \LHS of \eqref{eq:CZ-D} yields
\[
\input{figures/NormalForm/CZ-D-preimage.tikz}
\]

By the uniqueness of the symplectic normal form, $XZ^{j-a} \otimes X^aZ^{b-1}$ implies that the updated normal box should be the one in the \RHS of \eqref{eq:CZ-D}. Suppose that for some $\dir \in \circuitClifford_2$,
\[
\input{figures/NormalForm/CZ-D.tikz}
\]

By \cref{prop:untangled}, it remains to show that
\[
\input{figures/NormalForm/CZ-D-dir-auto.tikz}
\]

Equivalently, we can compute the preimage of $I\otimes X$ and $I\otimes Z$ in the \LHS of \eqref{eq:CZ-D-int}.

\begin{equation}
\input{figures/NormalForm/CZ-D-int.tikz}
\label{eq:CZ-D-int}
\end{equation}    

% By \cref{fig:D-Boxes-auto-phase-free}, for all $a,b\in \Z_p$, we have
% \[
% \tikzfig{figures/NormalForm/D-auto}
% \]

% It follows that

By \eqref{eq:D-auto}, we have

\[
\input{figures/NormalForm/CZ-D-int2.tikz}
\]

Since $S^a$ sends $XZ^{-a}$ to $X$ and $Z$ to $Z$, this completes the proof.
\end{proof} 

\begin{lemma}
    For all $a,b\in\Z_p$, there exists $\dir\in \circuitClifford_2$ such that
    \begin{equation}
        \input{figures/NormalForm/CZ-BB.tikz}
        \label{eq:CZ-BB}
    \end{equation}    
    \label{lem:CZ-BB}
\end{lemma}

\begin{proof}
Computing the preimage of $Z\otimes I \otimes I$ in the \LHS of \eqref{eq:CZ-BB} yields

\[
\input{figures/NormalForm/CZ-BB-preimage-v2.tikz}
\]

By the uniqueness of the symplectic normal form, $X^cZ^{d-a}\otimes X^aZ^{b-c} \otimes Z$ implies that the updated normal boxes should be the ones in the \RHS of \eqref{eq:CZ-BB}. Suppose that for some $\dir\in \circuitClifford_3$,

\[
\input{figures/NormalForm/CZ-BB-int.tikz}
\]

By \cref{prop:untangled}, it remains to show that
\[
\input{figures/NormalForm/CZ-BB-dir-auto.tikz}
\]

Equivalently, we can compute the preimage of $Z\otimes I \otimes I$ and $X\otimes I \otimes I$ in the \LHS of \eqref{eq:CZ-BB-int2}.

\begin{equation}
\input{figures/NormalForm/CZ-BB-int2.tikz}
\label{eq:CZ-BB-int2}
\end{equation}    

According to \cref{fig:B-Boxes-auto-phase-free}, 
\[
\input{figures/NormalForm/B-auto.tikz}
\]

we have

\[
\input{figures/NormalForm/CZ-BB-int3.tikz}
\]

This completes the proof.
\end{proof}  

\begin{lemma}
    For all $a,b\in\Z_p$, there exists $\dir\in \circuitClifford_3$ such that
    \begin{equation}
        \input{figures/NormalForm/ICZ-B.tikz}
        \label{eq:ICZ-B}
    \end{equation}    
    \label{lem:ICZ-B}
\end{lemma}

\begin{proof}
Computing the preimage of $Z\otimes I \otimes I$ in the \LHS of \eqref{eq:ICZ-B} yields

\[
\input{figures/NormalForm/ICZ-B-preimage.tikz}
\]

By the uniqueness of the symplectic normal form, $X^aZ^b\otimes Z \otimes I$ implies that the updated normal box should be the one in the \RHS of \eqref{eq:ICZ-B}. 
\end{proof}

\begin{lemma}
    For all $a,b\in\Z_p$, there exists $\dir\in \circuitClifford_2$ such that
    \begin{equation}
        \input{figures/NormalForm/CZ-DD.tikz}
        \label{eq:CZ-DD}
    \end{equation}    
    \label{lem:CZ-DD}
\end{lemma}

\begin{proof}
For any $j\in \Z_p$, computing the preimage of $I\otimes I \otimes XZ^j$ in the \LHS of \eqref{eq:CZ-DD} yields
\[
\input{figures/NormalForm/CZ-DD-preimage.tikz}
\]

By the uniqueness of the symplectic normal form, $XZ^j \otimes X^aZ^{b-c}\otimes X^cZ^{d-a}$ implies that the updated normal boxes should be the ones in the \RHS of \eqref{eq:CZ-DD}. Suppose that for some $\dir\in \circuitClifford_3$,

\[
\input{figures/NormalForm/CZ-DD-int.tikz}
\]

By \cref{prop:untangled}, it remains to show that 
\[
\input{figures/NormalForm/CZ-DD-dir-auto.tikz}
\]

Equivalently, we can compute the preimage of $I\otimes I \otimes Z$ and $I\otimes I \otimes X$ in the \LHS of \eqref{eq:CZ-DD-int2}.

\begin{equation}
\input{figures/NormalForm/CZ-DD-int2.tikz}
\label{eq:CZ-DD-int2}
\end{equation}  

According to \cref{fig:D-Boxes-auto-phase-free}, 
\[
\input{figures/NormalForm/D-auto.tikz}
\]

we have

\[
\input{figures/NormalForm/CZ-DD-int3.tikz}
\]

This completes the proof.
\end{proof}  

%% file: scripts/appendix/normal-form.tex
\section{Proofs from \texorpdfstring{\cref{sec:phase-free-normal-form}}{Section~3.1}}
\label{sec:sectionthreeproofs}

\begin{T1}
\lemznormal
\end{T1}
% \begin{lemma}
  
%     For every $P \in \Pauli_n\setminus \{\omega^tI;\;t \in \Z_p\}$ be a non-scalar $n$-qudit Pauli operator, there exists a unique $Z$-normal circuit $W_Z$ such that $\llbracket W_Z\rrbracket_u\bullet P\equiv_p Z \otimes I \otimes \cdots \otimes I$.
% \label{lem:Z-normal-reduction}
% \end{lemma}

\begin{proof}
There are unique $a_j, b_j \in \Z_p, 1\leq j\leq n$ such that 
$P \equiv_p \bigotimes_{j=1}^{n}X^{a_j}Z^{b_j}$. Since $P$ is not a scalar, there exists $j$ such that $X^{a_j}Z^{b_j}$ is not an identity. Let $m$ be the largest such index, then we can write
\begin{equation}
    P \equiv_p \left(X^{a_1}Z^{b_1}\right)\otimes \cdots \otimes \left(X^{a_m}Z^{b_m}\right) \otimes I \otimes \cdots \otimes I.
    \label{eq:P-decomposition}
\end{equation}
Next, we show how to construct a unique $Z$-normal circuit $W_Z$ based on \eqref{eq:P-decomposition} and the unique normal box actions specified in the second column of \cref{fig:normal-box-auto}. Starting from the qudit wire $m$, $A_{a_m,b_m}$ is uniquely determined by $X^{a_m}Z^{b_m}$. As a result, the unitary implementation of this $A$ box maps $X^{a_m}Z^{b_m}$ to $Z$ up to a phase.%\footnote{Recall that the Pauli propagation is tracked up to a phase}. 
Consequently, $B_{a_{m-1},b_{m-1}}$ is uniquely determined by $X^{a_{m-1}}Z^{b_{m-1}} \otimes Z$. Then, this $B$ box maps $X^{a_{m-1}}Z^{b_{m-1}} \otimes Z$ to $Z \otimes I$. Continue this process by concatenating $B$ normal boxes upward, until reaching the top two qudit wires, where $B_{a_1,b_1}$ is uniquely determined by $X^{a_{1}}Z^{b_{1}}\otimes Z$. As a result, $X^{a_{1}}Z^{b_{1}}\otimes Z$ is mapped to $Z \otimes I$. This gives us a $Z$-normal circuit $W_Z$ that is displayed below, which is uniquely determined by $P$. Since all the Pauli propagations are tracked up to phases, we conclude that $\llbracket W_Z\rrbracket_u\bullet P \equiv_p Z \otimes I \otimes \cdots \otimes I$.

\[
\scalebox{.75}{\input{figures/PhaseFreeNormalForm/ZNormalReduction.tikz}}\qedhere
\] 
\end{proof}

Note that from this proof it is clear that given a $Z$-normal circuit $W_Z$, we can read off its preimage of $Z \otimes I \otimes \cdots \otimes I$ by reading the indices of each normal box in $W_Z$ from right to left.

\begin{T2}
\lemxnormal
\end{T2}

% \begin{lemma}
% For every $Q\in \Pauli_n$ such that $(Z \otimes I \otimes \cdots \otimes I) Q = \omega Q(Z \otimes I \otimes \cdots \otimes I)$, there exists a unique $X$-normal circuit $W_X$ such that $\llbracket W_X\rrbracket_u\bullet Q \equiv_p I \otimes \cdots \otimes I \otimes X$.
% \label{lem:X-normal-reduction}
% \end{lemma}

\begin{proof}
There are unique $a_j, b_j\in \Z_p, 1\leq j\leq n$ such that  $Q \equiv_p \bigotimes_{j=1}^{n} X^{a_j}Z^{b_j}$. Since $(Z \otimes I \otimes \cdots \otimes I)Q = \omega Q (Z \otimes I \otimes \cdots \otimes I)$, we must have
\begin{equation}
    Q\equiv_p\left(XZ^i\right) \otimes \left(X^{a_2}Z^{b_2}\right) \otimes \cdots \otimes \left(X^{a_n}Z^{b_n}\right).
    \label{eq:Q-decomposition}
\end{equation}

Next, we show how to construct a unique $X$-normal circuit $W_X$ based on \eqref{eq:Q-decomposition} and the unique normal box actions specified in the second column of \cref{fig:normal-box-auto}. Starting from qudit wires $1$ and $2$, $D_{a_2,b_2}$ is uniquely determined by $ (XZ^i) \otimes (X^{a_2}Z^{b_2})$. As a result, this $D$ box maps $(XZ^i) \otimes (X^{a_2}Z^{b_2})$ to $I \otimes (XZ^i)$. Consequently, $D_{a_{3},b_{3}}$ is uniquely determined by $  (XZ^i) \otimes (X^{a_3}Z^{b_3})$. Again, this $D$ box maps $(XZ^i) \otimes (X^{a_3}Z^{b_3})$ to $I \otimes (XZ^i)$. Continue this process by concatenating $D$ normal boxes downward, until reaching the bottom two qudit wires, where $D_{a_n,b_n}$ is uniquely determined by $ (XZ^i) \otimes (X^{a_n}Z^{b_n})$. Thus, $ (XZ^i) \otimes (X^{a_n}Z^{b_n})$ is mapped to $I\otimes (XZ^i)$ by this last $D$ box.  Then, $E_i$ is uniquely determined by $XZ^i$, and it maps $XZ^i$ to $X$. This gives us an $X$-normal circuit $W_X$ that is displayed below. It is uniquely determined by $Q$ and $\llbracket W_X\rrbracket_u\bullet Q \equiv_p I \otimes \cdots \otimes I \otimes X$. 

\[
\scalebox{0.75}{\input{figures/PhaseFreeNormalForm/XNormalReduction.tikz}}\qedhere
\]
\end{proof}

Given an $X$-normal circuit $W_X$, we can read off its preimage of $I \otimes \cdots \otimes I \otimes X$ by reading the indices of each normal box in $W_X$ from right to left.

\begin{T3}
\lemxnormalpropagateZ
\end{T3}

% \begin{lemma}
% Every $X$-normal circuit $W_X$ satisfies $\llbracket W_X\rrbracket_u\bullet(Z \otimes I \otimes \cdots \otimes I \otimes I)\equiv_p I \otimes I \otimes \cdots \otimes I \otimes Z$.
% \label{lem:X-normal-propagate-Z}
% \end{lemma}

\begin{proof}
This follows from the additional actions of $D$ and $E$ boxes on certain Pauli operators (see the third column of \cref{fig:normal-box-auto}). In particular, every $D$ box maps $Z \otimes I$ to $I \otimes Z$, every $E$ box maps $Z$ to $Z$. It is visualised below.

    \[
    \scalebox{0.75}{\input{figures/PhaseFreeNormalForm/XNormalZCommute.tikz}}\qedhere
    \]
\end{proof}

\begin{T4}
    \lemnormalcircuitcompo
    % Let $P,Q\in \Pauli_n$ such that $PQ = \omega QP$. Then there exist unique normal circuits $W_Z$ and $W_X$ such that $\llbracket W_XW_Z\rrbracket_u\bullet P  \equiv_p I \otimes \cdots \otimes I \otimes Z$ and $\llbracket W_XW_Z\rrbracket_u\bullet Q  \equiv_p I \otimes \cdots \otimes I \otimes X$.
    % \label{lem:normal-circuit-compositon}
\end{T4}

\begin{proof}
    Since $PQ = \omega QP$, we must have $P \notin \langle -\omega \rangle$. By \cref{lem:Z-normal-reduction}, there exists a unique $Z$-normal circuit $W_Z$ such that $\llbracket W_Z\rrbracket_u \bullet P \equiv_p  Z \otimes I \otimes \cdots \otimes I$. Moreover, $(\llbracket W_Z\rrbracket_u\bullet P) (\llbracket W_Z\rrbracket_u\bullet Q) = \llbracket W_Z\rrbracket_u\bullet (PQ) = \llbracket W_Z\rrbracket_u\bullet (\omega QP) = \omega (\llbracket W_Z\rrbracket_u\bullet Q)  (\llbracket W_Z\rrbracket_u\bullet P)$. Hence $(Z \otimes I \otimes \cdots \otimes I) (\llbracket W_Z\rrbracket_u\bullet Q) = \omega (\llbracket W_Z\rrbracket_u\bullet Q)  (Z \otimes I \otimes \cdots \otimes I)$. By \cref{lem:X-normal-reduction}, there exists a unique $X$-normal circuit $W_X$ such that $\llbracket W_XW_Z\rrbracket_u\bullet Q=\llbracket W_X\rrbracket_u \bullet (\llbracket W_Z\rrbracket_u\bullet Q) \equiv_p  I \otimes \cdots \otimes I \otimes X$. 
    
    By \cref{lem:X-normal-propagate-Z}, $\llbracket W_XW_Z\rrbracket_u\bullet P=\llbracket W_X\rrbracket_u\bullet(\llbracket W_Z\rrbracket_u\bullet P) \equiv_p  \llbracket W_X\rrbracket_u\bullet(Z \otimes I \otimes \cdots \otimes I)\equiv_p I \otimes \cdots \otimes I \otimes Z$. It follows that
    \begin{equation}
        \llbracket W_XW_Z\rrbracket_u\bullet P \equiv_p  I \otimes \cdots \otimes I \otimes Z, \qquad  \llbracket W_XW_Z\rrbracket_u\bullet Q \equiv_p  I \otimes \cdots \otimes I \otimes X.
        \label{eq:condition}
    \end{equation}

    Diagrammatically, we can express this process:
    \[
    \scalebox{0.8}{\input{figures/PhaseFreeNormalForm/decompose.tikz}}
    \]
    
    This proves the existence of the desired $Z$-normal circuit $W_Z$ and $X$-normal circuit $W_X$. Now suppose towards contradiction that there exist another $Z$-normal circuit $W_Z'$ and $X$-normal circuit $W_X'$ satisfying  \eqref{eq:condition}. Since $\llbracket W_X'\rrbracket_u\bullet (\llbracket W_Z'\rrbracket_u\bullet P) \equiv_p I \otimes \cdots \otimes I \otimes Z$, \cref{lem:X-normal-propagate-Z} implies that $\llbracket W_Z'\rrbracket_u\bullet P \equiv_p Z\otimes I\otimes \cdots \otimes I$. By the uniqueness of the $Z$-normal circuit in \cref{lem:Z-normal-reduction}, $W_Z\equiv W_Z'$. Then 
    \[
        \llbracket W_X'\rrbracket_u\bullet(\llbracket W_Z'\rrbracket_u\bullet Q) \equiv_p \llbracket W_X'\rrbracket_u\bullet(\llbracket W_Z\rrbracket_u\bullet Q)\equiv_pI \otimes \cdots \otimes I \otimes X.
    \]
    By the uniqueness of the $X$-normal circuit in \cref{lem:X-normal-reduction}, $W_X'\equiv W_X$.
\end{proof}

\begin{T5}
\propnormalform
% Let $\phi: \Pauli_n \rightarrow \Pauli_n$ be an automorphism of the Pauli group such that $\phi$ fixes scalars. There exists a unique Clifford circuit $C\in\circuitClifford_n$ in symplectic normal form such that for all $P\in \Pauli_n$, $\llbracket C \rrbracket_u\bullet P \equiv_p \phi(P)$.
%     \label{prop:normal-form}
\end{T5}

\begin{proof}
We proceed by induction on $n$. When $n=0$, the Pauli operators are scalars of the form $(-\omega)^t$, $t \in \Z_{2p}$. In this case, $\phi$ is the identity. So $C$ must be the empty circuit. Now suppose that our claim is true for $n-1$ and consider the case of $n$. 
%In what follows, we carry out discussions up to Pauli correction on the Clifford normal form. 
First we prove that the symplectic normal circuit $W^{(n)}$ on the $n$-th layer exists. 

\[
    \scalebox{1}{\input{figures/PhaseFreeNormalForm/Normal-induction.tikz}}
    \]
Since $\phi$ is bijective, there exist $P, Q \in \Pauli_n$ such that $\phi(P) = I \otimes \cdots \otimes I \otimes Z$ and $\phi(Q) = I \otimes \cdots \otimes I \otimes X$. That is, $P = \phi^{-1}(I \otimes \cdots \otimes I \otimes Z)$ and $Q =\phi^{-1}(I \otimes \cdots \otimes I \otimes X)$. Then $PQ =\phi^{-1}(I \otimes \ldots \otimes I \otimes ZX)$. Since $(I \otimes \ldots \otimes I \otimes Z)(I \otimes \ldots \otimes I \otimes X) = \omega (I \otimes \ldots \otimes I \otimes X)(I \otimes \ldots \otimes I \otimes Z)$, $PQ=\omega QP$. By \cref{lem:normal-circuit-compositon}, there exist a unique $X$-normal circuit $W_X$ and a unique $Z$-normal circuit $W_Z$ such that $\llbracket W_XW_Z\rrbracket_u \bullet P \equiv_p I \otimes \ldots \otimes I \otimes Z = \phi(P)$ and $\llbracket W_XW_Z\rrbracket_u \bullet Q \equiv_p I \otimes \ldots \otimes I \otimes X = \phi(Q)$.

\[
    \scalebox{0.9}{\input{figures/PhaseFreeNormalForm/W-decompose.tikz}}
    \]

Let $\phi ': \Pauli_{n} \rightarrow \Pauli_{n}$ be the new automorphism defined as
\begin{equation}
    \phi '(U) = \phi\left(\llbracket W_XW_Z\rrbracket_u^{-1} \bullet U\right).
    \label{eq:new-auto}
\end{equation}
Then $I \otimes \cdots \otimes I \otimes Z$ and $I \otimes \cdots \otimes I \otimes X$ are fixed points of $\phi'$, since for $P\in \{X,Z\}$,

\begin{align*}
\phi '(I \otimes \dots \otimes I \otimes P) &= \phi\left(\llbracket W_XW_Z\rrbracket_u^{-1} \bullet (I \otimes \dots \otimes I \otimes P)\right) \\
&= \phi\left(\llbracket W_XW_Z\rrbracket_u^{-1} \bullet \llbracket W_XW_Z\rrbracket_u \bullet P\right) \equiv_p  \phi(P) = I \otimes \dots \otimes I \otimes P.
\end{align*}

For any $R \in \Pauli_{n-1}$, since $R \otimes I$ commutes with $I \otimes \dots \otimes I \otimes Z$ and $I \otimes \dots \otimes I \otimes X$, $\phi '(R \otimes I)$ commutes with $\phi '(I \otimes \dots \otimes I \otimes Z)= I \otimes \dots \otimes I \otimes Z$ and $\phi '(I \otimes \dots \otimes I \otimes X)= I \otimes \dots \otimes I \otimes X$. This implies that $\phi '(R\otimes I) = S \otimes I$ for some $S \in \Pauli_{n-1}$. Then there exists an automorphism $\phi '': \Pauli_{n-1} \rightarrow \Pauli_{n-1}$ such that $\phi''(R) = S$. Since $\phi '$ fixes $I \otimes \dots \otimes I \otimes Z$ and $I \otimes \dots \otimes I \otimes X$, 

\begin{equation}
    \phi ' = \phi '' \otimes I.
    \label{eq:equality}
\end{equation}

By the induction hypothesis, there exists $C' \in \circuitClifford_{n-1}$ in normal form such that, for all $R \in \pauli{n-1}$,

\begin{equation}
    \llbracket C'\rrbracket_u \bullet R \equiv_p  \phi ''(R).
    \label{eq:IH}
\end{equation}

Then $C =  (C'\otimes I)  W_XW_Z$ is in normal form. Next, we show that $\llbracket C\rrbracket_u \bullet U \equiv_p \phi(U)$ for all $U \in \Pauli_n$. By  \eqref{eq:new-auto},

\begin{equation}
    \llbracket W_XW_Z\rrbracket_u^{-1}\bullet U = \phi^{-1}  \phi'(U).
    \label{eq:map1}
\end{equation}
It follows that 
\begin{align*}
    \llbracket C\rrbracket_u \bullet U &=\llbracket(C' \otimes I)  W_XW_Z\rrbracket_u \bullet U \xlongequal[]{\eqref{eq:map1}} \llbracket C' \otimes I\rrbracket_u\bullet\left(\phi'^{-1}\left(\phi(U)\right)\right)\\
    &\xlongequal[]{\eqref{eq:equality}}
    \llbracket C' \otimes I\rrbracket_u\bullet (\phi''^{-1}\otimes I)\left(\phi(U)\right)\overset{\eqref{eq:IH}}{\equiv_p} \phi(U).
\end{align*}

To prove uniqueness, suppose towards contradiction that $D \in \circuitClifford_n$ is another Clifford circuit in symplectic normal form such that $\llbracket D\rrbracket_u \bullet U \equiv_p \phi(U)$ for all $U \in \pauli{n}$.  By \cref{def:normal}, $D = (D' \otimes I)  \left(W_X'W_Z'\right)$, where $W'_X$ is an $X$-normal circuit, $W'_Z$ is a $Z$-normal circuit, and $D'$ is a symplectic normal Clifford circuit on $n-1$ qudits. Since $\llbracket D \rrbracket_u \bullet P \equiv_p \phi(P) = I \otimes \dots \otimes I \otimes Z$, $\left(\llbracket D' \otimes I\rrbracket_u  \llbracket W_X'W_Z'\rrbracket_u\right)\bullet P \equiv_p I \otimes \dots \otimes I \otimes Z$. Then

\[
  \llbracket W_X'W_Z'\rrbracket_u \bullet P \equiv_p \llbracket D' \otimes I\rrbracket_u^{-1}\bullet(I \otimes \dots \otimes I \otimes Z) = (\llbracket D'\rrbracket_u^{-1} \otimes I) \bullet(I \otimes \dots \otimes I \otimes Z)=I \otimes \dots \otimes I \otimes Z.
\]

Similarly, $\llbracket D\rrbracket_u\bullet Q \equiv_p I \otimes \dots \otimes I \otimes X$ implies that $\llbracket W_X'W_Z'\rrbracket_u \bullet Q \equiv_p I \otimes \dots \otimes I \otimes X$. By the uniqueness of the $X$- and $Z$-normal circuits in \cref{lem:normal-circuit-compositon}, $W'_X\equiv W_X$ and $W'_Z\equiv W_Z$. Since $C$ and $D$ both act as $\phi$, it follows that $\llbracket D'\rrbracket_u \bullet R \equiv_p \llbracket C'\rrbracket_u \bullet R \equiv_p \phi''(R)$ for all $R \in \pauli{n-1}$. By the induction hypothesis, $C'\equiv D'$. Thus, the same is true for $C$ and $D$. This completes the proof.
\end{proof}

\begin{T6}
\corcardinality
% $\left\lvert \symplectic{2n, \Z_p} \right\rvert = \prod_{k=1}^n \left(p^{2k}-1\right) \cdot \left(p^{2k-1}\right).$
    % \label{lem:cardinality}
\end{T6}

\begin{proof}
Since $\symplectic{2n, \Z_p} \cong \Clifford_n/\Pauli_n$, by \cref{prop:normal-form}, it is sufficient to count the total number of distinct symplectic normal forms. Recall the symplectic normal form of an arbitrary $n$-qudit Clifford operator, which is unique up to Pauli correction.

 \[
    \input{figures/PhaseFreeNormalForm/cardinality2.tikz}
    \]

By induction, our problem is reduced to counting different ways of constructing $W_Z^{(n)}$ and $W_X^{(n)}$. 
\begin{itemize}
    \item For $W_Z^{(n)}$, according to \cref{fig:Z-and-X-Normal-Boxes-imp}, there are $p^2-1$ choices for an $A$ box. For $B$ boxes, we can start concatenating them upward from any qudit wire. According to \cref{def:Z-and-X-normal}, we proceed by cases.
    \begin{description}
        \item[In one extreme case: ] The ladder of $B$ boxes starts from the bottom qudit wire, as shown below. There are $n-1$ $B$ boxes in $W_Z^{(n)}$. According to \cref{fig:Z-and-X-Normal-Boxes-imp}, there are $p^2$ choices for a $B$ box. Since each choice of a normal box is independent of each other, there are $(p^2-1) \cdot p^2{^{(n-1)}}= p^{2n} - p^{2(n-1)}$ distinct $W_Z^{(n)}$ when it expands over all $n$ qudits.
        \[
            \input{figures/PhaseFreeNormalForm/ZNormalExtreme.tikz}
        \]
        \item[In the other extreme case: ]The ladder of $B$ boxes starts from the top qudit wire, as shown below. There is no $B$ box in $W_Z^{(n)}$. Since each choice of a normal box is independent of each other, there are $p^2-1$ distinct $W_Z^{(n)}$ when it expands over $1$ qudit.
        \[
            \input{figures/PhaseFreeNormalForm/ZNormalExtreme3.tikz}
        \]
        \item[In all other remaining cases: ] The ladder of $B$ boxes starts from qudit wire $k$, $1 < k < n$. In the illustration below, there are $(k-1)$ $B$ boxes in $W_Z^{(n)}$. Reasoning analogously as before, there are $(p^2-1) \cdot p^{2(k-1)} = p^{2k} - p^{2(k-1)}$ distinct $W_Z^{(n)}$ when it expands over $k$ qudits.
        \[
            \input{figures/PhaseFreeNormalForm/ZNormalExtreme2.tikz}
        \]
    \end{description}

    Considering all cases of concatenating B boxes in $W_Z^{(n)}$, the number of distinct $W_Z^{(n)}$ is
    \[
    \sum_{k = 1}^{n}\left(p^{2k} - p^{2(k-1)}\right) = \left(\sum_{k = 1}^{n}p^{2k}\right) - \left(\sum_{k = 1}^{n}p^{2(k-1)}\right)=\frac{\left(p^2-1\right)\left(1-p^{2n}\right)}{1-p^2}=p^{2n}-1.
    \]
    \item For $W_X^{(n)}$, we consider all possible ways of concatenating $D$ and $E$ boxes. According to \cref{fig:Z-and-X-Normal-Boxes-imp}, there are $p$ choices for an $E$ box. For each $D$ box, there are $p^2$ choices. According to \cref{def:Z-and-X-normal}, there are $n-1$ D boxes in $W_X^{(n)}$. Since each choice of a normal box is independent of each other, there are $p \cdot \left(p^2\right)^{(n-1)} = p^{2n-1}$ distinct $W_X^{(n)}$.
\end{itemize}
Since the choice of $W_Z^{(n)}$ and $W_X^{(n)}$ is independent of each other, there are $\left(p^{2n}-1\right) \cdot \left(p^{2n-1}\right)$ distinct $W_X^{(n)}W_Z^{(n)}$. According to the symplectic normal form outlined in \cref{fig:macro-normal-circuit}, we can calculate the total number of distinct $N_S^{(n)}$:
\[
\prod_{k=1}^n\left(p^{2k}-1\right) \cdot \left(p^{2k-1}\right).
\]
Since they are in one-to-one correspondence with the elements in $\Clifford_n/\Pauli_n$, this completes the proof.
\end{proof}

%% file: scripts/appendix/boxrelations.tex
\clearpage
\newgeometry{margin=2cm}
\section{A Complete Set of Box Relations}
\label{app:relations}

% \subsection{Phase-Free Box Relations}
% \label{subsec:box-relations-phase-free}
% In this section, we show how to achieve single-qupit Clifford completeness up to global phases and Pauli correction.
% \subsubsection{Single-Qupit Phase-Free Box Relations}
% \label{subsubsec:box-relations-phase-free-single}

\begin{figure}[H]
    \centering
    \[
    \scalebox{1}{\input{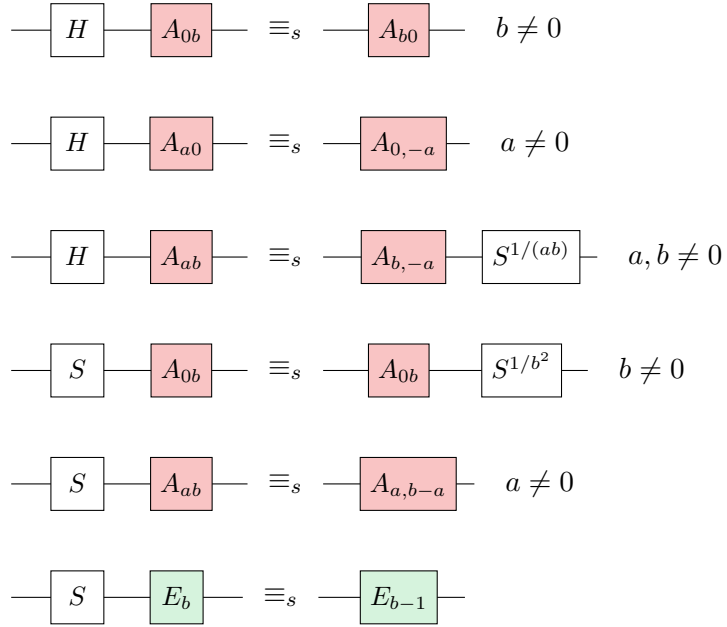}}
    \]
    \caption{Single-qudit box relations: Up to Pauli correction, pushing $H$ (or $S$) through an $A$ box, pushing $S$ through an $E$ box.}
    \label{fig:single-qudit-relations}
\end{figure}

\iffalse
\begin{figure}[H]
    \centering
    \[
    \scalebox{.8}{\input{figures/BoxRelations/H.A-v2.tikz}}
    \]
    \caption{Pushing $H$ through an $A$ box, up to Pauli correction.}
    \label{fig:H-A-relations}
\end{figure}

\begin{figure}[H]
    \centering
    \[
    \scalebox{.8}{\input{figures/BoxRelations/S.A-v2.tikz}}
    \]
    \caption{Pushing $S$ through an $A$ box, up to Pauli correction.}
    \label{fig:S-A-relations}
\end{figure}

\begin{figure}[H]
    \centering
    \[
    \scalebox{.8}{\input{figures/BoxRelations/S.E.tikz}}
    \]
    \caption{Pushing $S$ through an $E$ box.}
    \label{fig:S-E-relations}
\end{figure}
\fi

\begin{figure}[H]
    \centering
    \[
    \scalebox{1}{\input{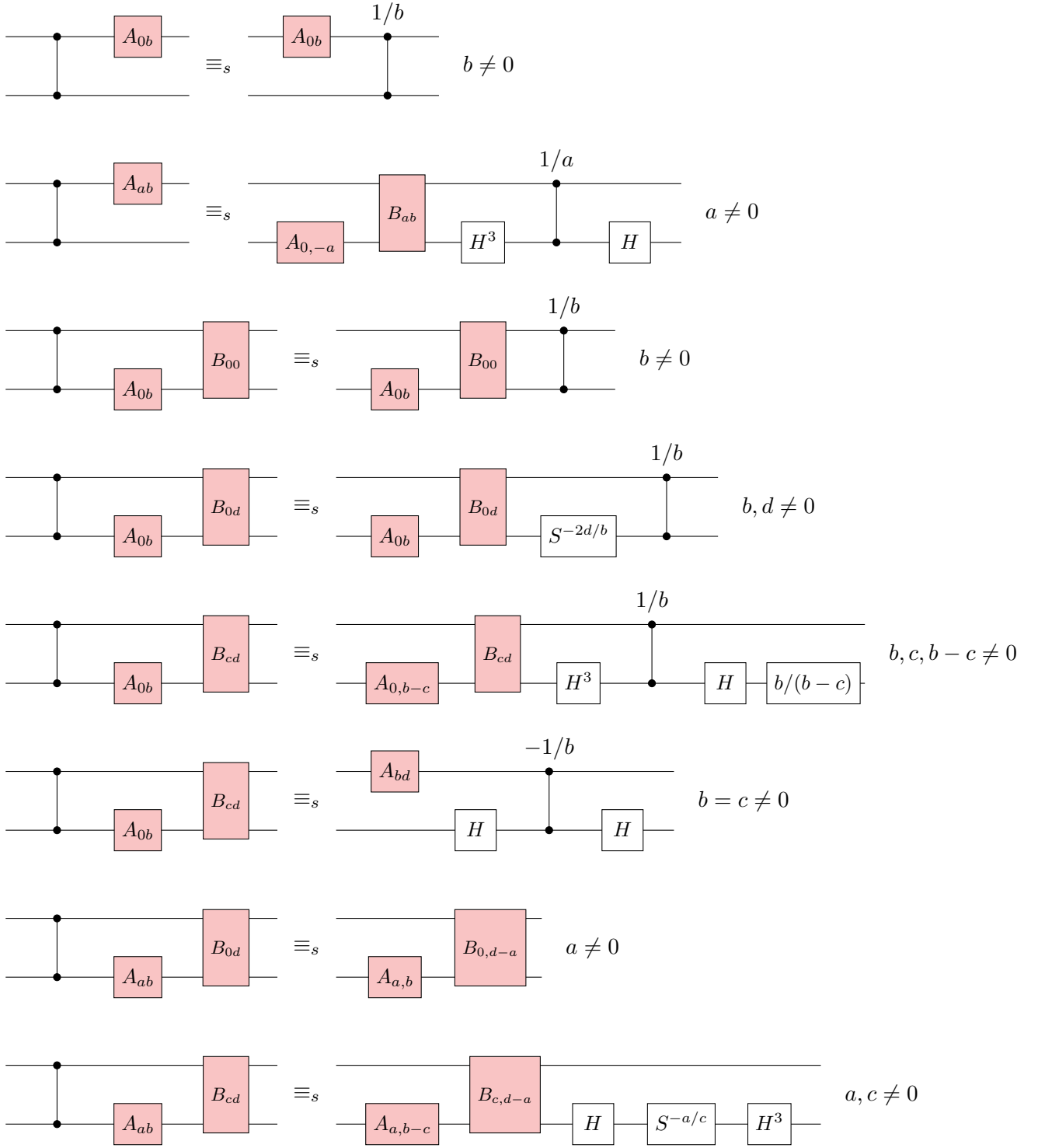}}
    \]
    \caption{Two-qudit box relations: Up to Pauli correction, pushing $\CZ$ through an $A$ box or $A.B$ boxes.}
    \label{fig:two-qudit-relations-CZ}
\end{figure}

\begin{figure}[H]
    \centering
    \[
    \scalebox{1}{\input{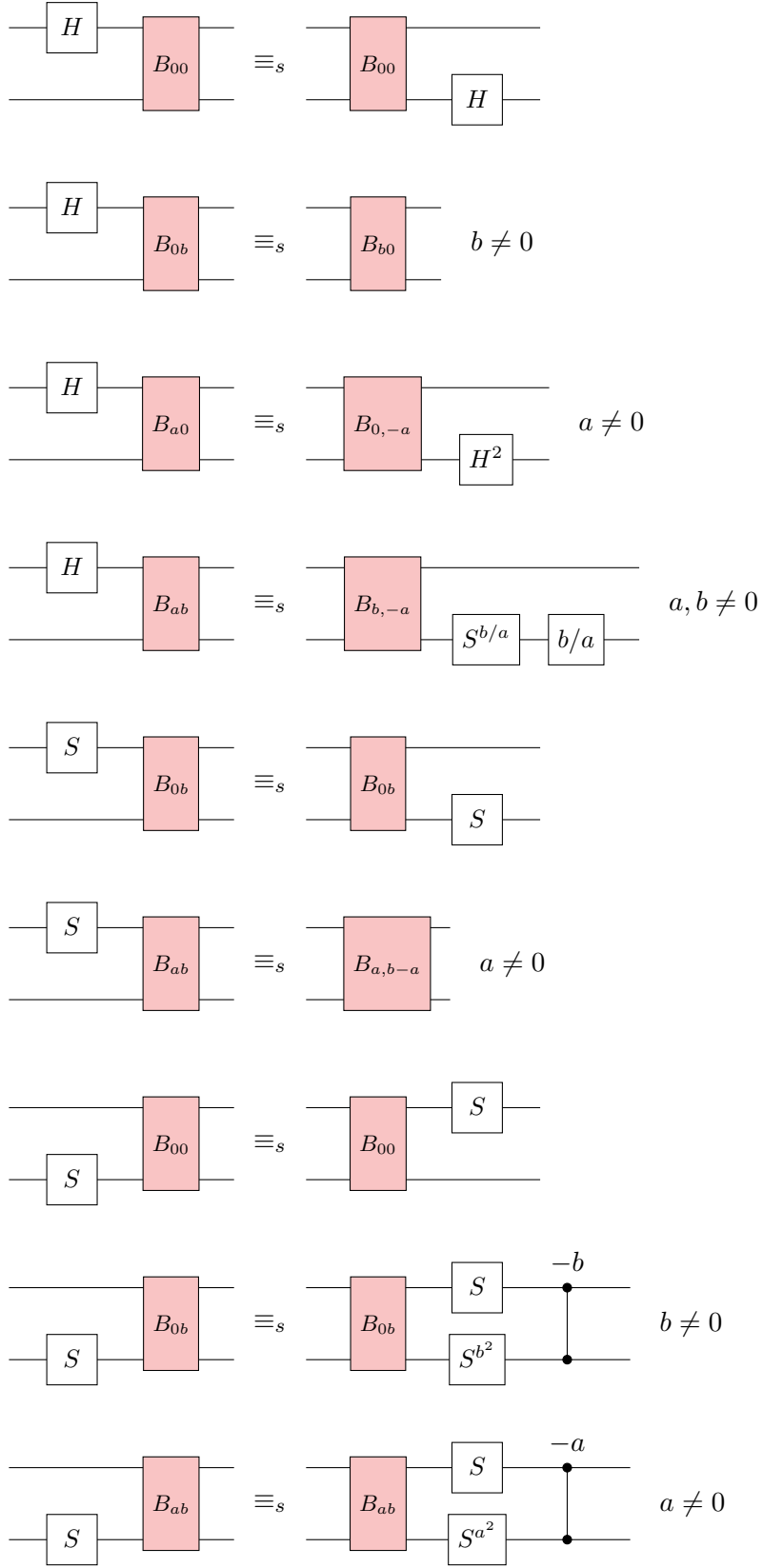}}
    \]
    \caption{Two-qudit box relations: Up to Pauli correction, pushing H and S gates through a B box.}
    \label{fig:two-qudit-relations-H-S-B}
\end{figure}

\begin{figure}[H]
    \centering
    \[
    \scalebox{1}{\input{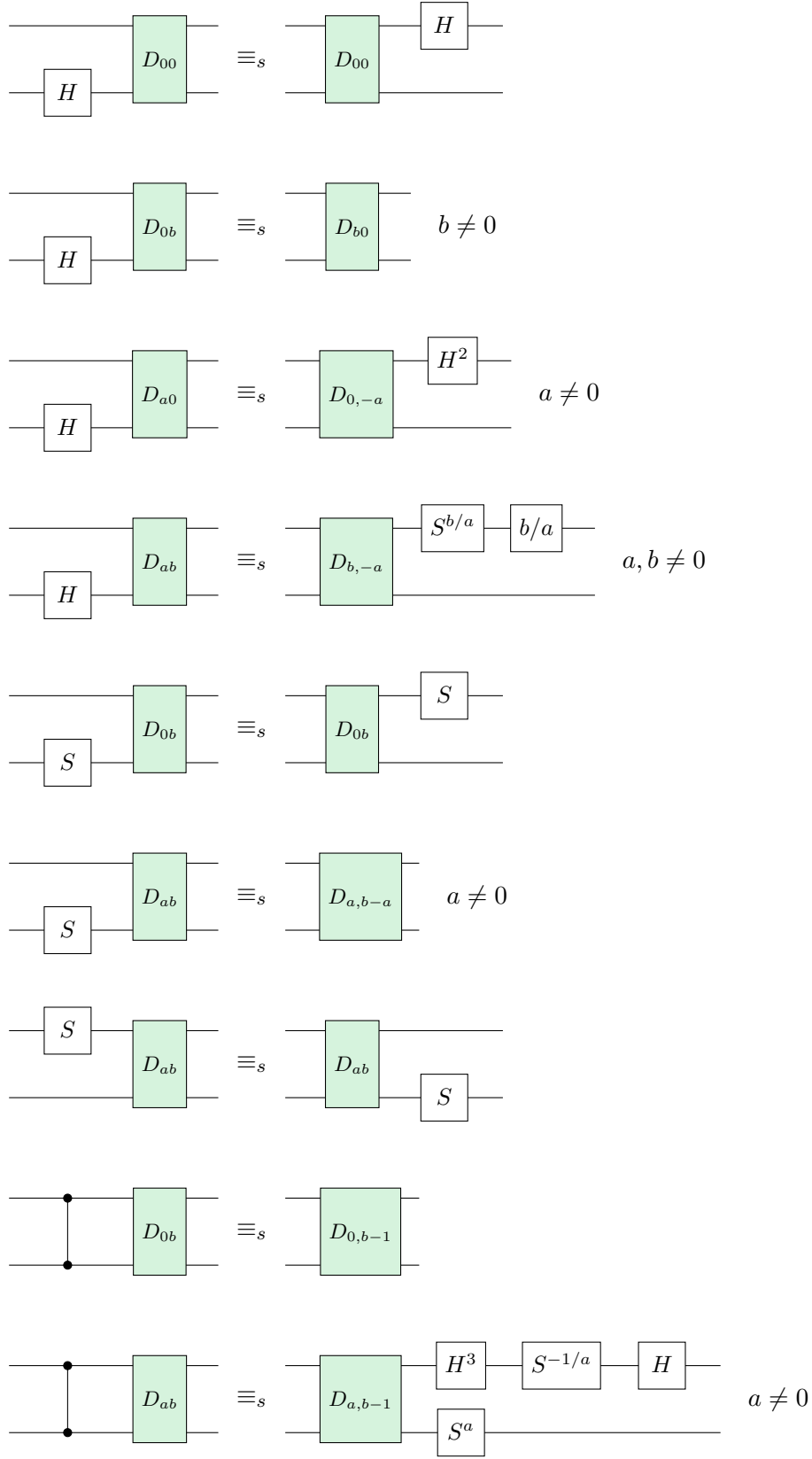}}
    \]
    \caption{Two-qudit box relations: Up to Pauli correction, pushing H $\downarrow$, S $\uparrow$, S $\downarrow$, and $\CZ$ through a D box.}
    \label{fig:two-qudit-relations-H-S-CZ-D}
\end{figure}

\begin{figure}[H]
    \centering
    \[
    \scalebox{1}{\input{figures/BoxRelations/br3CZuBB.tikz}}
    \]
    \caption{Three-qudit box relations: Up to Pauli correction, pushing CZ $\uparrow$ through two B boxes.}
    \label{fig:three-qudit-relations-CZ-B-B}
\end{figure}

\begin{figure}[H]
    \centering
    \[
    \scalebox{1}{\input{figures/BoxRelations/br3CZdBu.tikz}}
    \]
    \caption{Three-qudit box relations: Up to Pauli correction, pushing CZ $\downarrow$ through a B box.}
    \label{fig:three-qudit-relations-CZ-B}
\end{figure}

\begin{figure}[H]
    \centering
    \[
    \scalebox{1}{\input{figures/BoxRelations/br3CZdDD.tikz}}
    \]
    \caption{Three-qudit box relations: Up to Pauli correction, pushing CZ $\downarrow$ through two D boxes.}
    \label{fig:three-qudit-relations-CZ-D-D}
\end{figure}